\documentclass[11pt]{article} 
\usepackage[utf8]{inputenc}
\usepackage{fullpage}
\usepackage{upref}
\usepackage{enumerate}
\usepackage{latexsym}
\usepackage{pdfpages}
\usepackage{ifpdf}
\ifpdf % Do nothing or load pdflatex-specific packages if needed
\else % Only load this package if not using pdflatex
\usepackage[hyphenbreaks]{breakurl}
\fi
\usepackage[bookmarks, colorlinks=true, urlcolor=violet, linkcolor=blue, citecolor=red, hyperindex=true, linktocpage=true, pagebackref=true, final=true]{hyperref}
\usepackage{stmaryrd}
\usepackage{dsfont}
\usepackage{mleftright}\mleftright
\makeatletter
\def\@fnsymbol#1{\ensuremath{\ifcase#1\or *\or \mathsection\or \mathparagraph\or \|\or **\or \ddagger \else\@ctrerr\fi}}
\makeatother

\usepackage[boxed]{algorithm2e}
\usepackage[noend]{algpseudocode}

\usepackage{pdflscape}

\usepackage{color,graphics}
\usepackage{comment} 
\usepackage{caption}
\usepackage{subcaption}
\usepackage{wrapfig}
\usepackage{braket}
\usepackage{amssymb}
\usepackage{amsmath}
\usepackage{mathtools}
\usepackage{mathabx}
\usepackage{amsthm}
\usepackage{mathrsfs}
\usepackage{mathtools}
\usepackage{commath}
\usepackage{thmtools,thm-restate}

\DeclarePairedDelimiterX\ketbra[2]{|}{|}{#1 \delimsize\rangle\delimsize\langle #2}
\DeclarePairedDelimiter{\nrm}{\lVert}{\rVert}

\newenvironment{claim}[1]{\par\noindent\underline{Claim:}\space#1}{}
\newenvironment{claimproof}[1]{\par\noindent\underline{Proof:}\space#1}{\hfill $\blacksquare$}

\usepackage{scalerel,stackengine}
\stackMath
\newcommand\reallywidehat[1]{%
\savestack{\tmpbox}{\stretchto{%
  \scaleto{%
    \scalerel*[\widthof{\ensuremath{#1}}]{\kern-.6pt\bigwedge\kern-.6pt}%
    {\rule[-\textheight/2]{1ex}{\textheight}}%WIDTH-LIMITED BIG WEDGE
  }{\textheight}% 
}{0.5ex}}%
\stackon[1pt]{#1}{\tmpbox}%
}

\usepackage{amsfonts}
\usepackage{varioref}
\usepackage[nameinlink]{cleveref}
\usepackage[many]{tcolorbox}

\usepackage{xcolor}
\hypersetup{
	colorlinks,
	linkcolor={red!75!black},
	citecolor={blue!75!black},
	urlcolor={blue!75!black}
}
\usepackage[hyperpageref]{backref}

\newtheorem{theorem}{Theorem}[section]
\newtheorem{lemma}[theorem]{Lemma}
\newtheorem{corollary}[theorem]{Corollary}
\newtheorem{proposition}[theorem]{Proposition}
\newtheorem{prop}[theorem]{Proposition}

\theoremstyle{remark}

\theoremstyle{definition}

\AtBeginDocument{}

\newcommand{\E}{\mathbb{E}} 	%Expectation

\newcommand{\inProd}[2]{\langle{#1},{#2}\rangle}

\newcommand{\poly}{\mathrm{poly}}
\newcommand{\polylog}{\mathrm{polylog}}

\newcommand{\vect}[1]{\pmb{#1}}
\newcommand{\eps}{\varepsilon}
\renewcommand{\epsilon}{\varepsilon}
\renewcommand{\vec}[1]{\vect{#1}}

\newcommand{\R}{\mathbb{R}}
\newcommand{\C}{\mathbb{C}}
\newcommand{\propDens}{p}

\newcommand{\ipd}[2]{\left \langle {#1}, {#2} \right \rangle}

\newcommand{\bigO}[1]{\mathcal{O}\left( #1 \right)}
\newcommand{\bigOt}[1]{\tilde{\mathcal{O}}\left( #1 \right)}

\newcommand{\Exp}{\mathbb{E}}
\newtheorem{definition}[theorem]{Definition}

\usepackage{ifdraft}
\ifdraft{
	\newcommand{\authnote}[3]{{\footnotesize\color{#3} ({\bf #1:} #2)}}
}{
	\newcommand{\authnote}[3]{}
}
\newcommand{\rnote}[1]{\authnote{Ronald}{#1}{blue}}
\newcommand{\ynote}[1]{\authnote{Yanlin}{#1}{red}}
\newcommand{\anote}[1]{\authnote{András}{#1}{teal}}

\title{A Quantum Speed-Up\\ 
for Approximating the Top Eigenvectors of a Matrix}
\author{Yanlin Chen\thanks{QuSoft, CWI, the Netherlands. {\tt yanlin.chen@cwi.nl}} 
\and
Andr\'{a}s Gily\'{e}n\thanks{HUN-REN Alfr\'{e}d R\'{e}nyi Institute of Mathematics, Budapest, Hungary. 
Funded by the EU's Horizon 2020 Marie Sk{\l}odowska-Curie program QuantOrder-891889 and the QuantERA II project \href{https://quantera.eu/hqcc/}{HQCC}-101017733 in coordination with the national funding organisation NKFIH.
{\tt gilyen@renyi.hu}}
\and 
Ronald de Wolf\thanks{QuSoft, CWI and University of Amsterdam, the Netherlands. Partially supported by the Dutch Research Council (NWO) through Gravitation-grant Quantum Software Consortium, 024.003.037. {\tt rdewolf@cwi.nl}} 
}
\date{}

\begin{document}

\maketitle

\begin{abstract}
Finding a good approximation of the top eigenvector of a given $d\times d$ matrix~$A$ is a basic and important computational problem, with many applications. We give two different quantum algorithms that, given query access to the entries of a Hermitian matrix $A$ and assuming a constant eigenvalue gap, output a classical description of a good approximation of the top eigenvector: one algorithm with time complexity $\mathcal{\tilde{O}}(d^{1.75})$ and one with time complexity $d^{1.5+o(1)}$ (the first algorithm has a slightly better dependence on the $\ell_2$-error of the approximating vector than the second, and uses different techniques of independent interest).
Both of our quantum algorithms provide a polynomial speed-up over the best-possible classical algorithm, which needs $\Omega(d^2)$ queries to entries of $A$, and hence $\Omega(d^2)$ time.
We extend this to a quantum algorithm that outputs a classical description of the subspace spanned by the top-$q$ eigenvectors in time $qd^{1.5+o(1)}$.
We also prove a nearly-optimal lower bound of $\tilde{\Omega}(d^{1.5})$ on the quantum query complexity of approximating the top eigenvector.

Our quantum algorithms run a version of the classical power method that is robust to certain benign kinds of errors, where we implement each matrix-vector multiplication with small and well-behaved error on a quantum computer, in different ways for the two algorithms. 
Our first algorithm estimates the matrix-vector product one entry at a time, using a new ``Gaussian phase estimation'' procedure.
Our second algorithm uses block-encoding techniques to compute the matrix-vector product as a quantum state, from which we obtain a classical description by a new time-efficient unbiased pure-state tomography procedure. This procedure uses an essentially optimal number $\bigO{d\log(d)/\eps^2}$ of ``conditional sample states''; if we have a state-preparation unitary available rather than just copies of the state, then this $\eps$-dependence can be improved further quadratically. Our procedure comes with improved statistical properties and faster runtime compared to earlier pure-state tomography algorithms.
We also develop an almost optimal time-efficient process-tomography algorithm for reflections around bounded-rank subspaces, providing the basis for our top-eigensubspace estimation algorithm, and in turn providing a pure-state tomography algorithm that only requires a reflection about the state rather than a state preparation unitary as input.
\end{abstract}

\thispagestyle{empty}
\setcounter{page}{-1}

\newpage\tableofcontents\thispagestyle{empty}\newpage

\section{Introduction}

Arguably the most important property of a diagonalizable $d\times d$ matrix~$A$ is its largest eigenvalue~$\lambda_1$, with an associated eigenvector $v_1$. This top eigenvector~$v_1$ can be thought of as the most important ``direction'' in which the matrix~$A$ operates.
The ability to efficiently find $v_1$ is an important tool in many applications, for instance in the PageRank algorithm of Google's search engine, as a starting point for principal component analysis (for clustering or dimensionality-reduction), for Fisher discriminant analysis, or in continuous optimization problems where sometimes the best thing to do is to move the current point in the direction of the top eigenvector of an associated matrix~\cite{Jol86PCA,KV09spectral}. \ynote{added two references in response to Reviewer 1}

One way to find the top eigenvector of $A$ is to diagonalize the whole matrix. Theoretically this takes matrix multiplication time: $\mathcal{O}(d^{\omega})$ where $\omega\in[2,2.37\ldots)$ is the still-unknown matrix multiplication exponent~\cite{williams2023newboundsmatrixmultiplication,alman2024asymmetryyieldsfastermatrix}. In practice Gaussian elimination (which takes time $\mathcal{O}(d^3)$) is typically faster, unless $d$ is enormous. Diagonalization gives us not only the top eigenvector but a complete orthonormal set of $d$ eigenvectors. However, this is doing too much if we only care about finding the top eigenvector, or the top-$q$ eigenvectors for some $q\ll d$, and better methods exist in this case (see e.g.~\cite{parlett:symeigenvalue} for a whole book about this).

\subsection{The power method for approximating a top eigenvector}

A quite efficient method for (approximately) finding the top eigenvector is the iterative ``power method''. This uses simple matrix-vector multiplications instead of any kind of matrix decompositions, and works as follows.
We start with a random unit vector $w_0$ (say with i.i.d.\ Gaussian entries). This is a linear combination $\sum_{i=1}^d \alpha_i v_i$ of the $d$ unit eigenvectors $v_1,\ldots,v_d$ of the Hermitian matrix $A$, with coefficients of magnitude typically around $1/\sqrt{d}$. Then we apply $A$ to this vector some $K$ times, computing $w_1=Aw_0$, $w_2=Aw_1$, etc., up to $w_K=A^K w_0$.  This has the effect of multiplying each coefficient $\alpha_i$ with the powered eigenvalue $\lambda_i^K$. If there is some ``gap'' between the first two eigenvalues (say $|\lambda_1|-|\lambda_2|\geq\gamma>0$, or $|\lambda_1/\lambda_2|>1$), then the relative weight of the coefficient of $v_1$ will start to dominate all the other coefficients even already for small $K$, and the renormalization of the final vector $w_K=A^K w_0$ will be close to $v_1$, up to global phase. Specifically, if $A$ has bounded operator norm and eigenvalue gap $\gamma$, then $K=\bigO{\log(d)/\gamma}$ iterations suffice to approximate $v_1$ up to $1/\poly(d)$ $\ell_2$-error (see e.g.~\cite[Section 8.2.1]{golub&vanloan:matrixcomp} for details).

The cost of this algorithm is dominated by the $K$ matrix-vector multiplications, each of which costs $\bigOt{d^2}$ time classically (or $\bigOt{m}$ time if $A$ is sparse, with only $m$ nonzero entries given in some easily-accessible way like lists of nonzero entries for each row and column). Hence if the eigenvalue gap $\gamma$ is not too small, say constant or at least $1/\polylog(d)$, then the power method takes $\mathcal{\tilde{O}}(d^2)$ time to approximate a top eigenvector.\footnote{The $\bigOt{\cdot}$ notation suppresses polylog factors in $d,\eps,\delta$: $\bigOt{T}=\bigO{T\polylog(d/(\eps\delta))}$.}  Unsurprisingly, as we show later in this paper, $\Omega(d^2)$ queries to the entries of $A$ are also \emph{necessary} for classical algorithms for this.

\subsection{Our results: quantum algorithms}\label{ssec:resultsalgos}

Our main results in this paper are faster quantum algorithms for (approximately) finding the top eigenvector of a Hermitian matrix~$A$, and nearly matching lower bounds.\footnote{If the matrix $A$ is non-Hermitian, then we could instead use the Hermitian matrix $A'=\left[\begin{array}{cc}0 & A\\ A^\dagger & 0\end{array}\right]$ for finding the largest (left or right) singular value of $A$, replacing the eigenvalue gap with the singular value gap of $A$.} 

\paragraph{Quantum noisy power method.}
On a high level we just run the power method to find a good approximation for $v_1$, with classical representations of $w_0$ and all intermediate vectors, but we perform each matrix-vector multiplication \emph{approximately} using a quantum computer.\footnote{\label{note:LBexactmatrixvector}\emph{Exact} matrix-vector multiplication takes $\Omega(d^2)$ quantum queries to entries of $A$ (and hence does not give speed-up). This is easy to see by taking $A\in\{0,1/d\}^{d\times d}$ and $w=(\frac{1}{\sqrt{d}},\ldots,\frac{1}{\sqrt{d}})^T$, because then $d^{1.5}Aw$ gives the number of nonzero entries in~$A$. It is well-known that $\Omega(d^2)$ quantum queries are needed to count this number exactly.}

We give two different quantum algorithms for approximate matrix-vector multiplication.
Our first algorithm uses that each entry of the vector $Aw$ is an inner product between a row of $A$ and the column-vector $w$. Because such an inner product is the sum of $d$ numbers, we may hope to approximate it well via some version of amplitude estimation or quantum counting, using roughly $\sqrt{d}$ time per entry and $d^{1.5}$ time for all $d$ entries of $Aw$ together. 
This approach is easier said than done, because basic quantum-counting subroutines produce small errors in the approximation of each entry, and those errors might add up to a large $\ell_2$-error in the $d$-dimensional vector $Aw$ as a whole. To mitigate this issue we develop a ``Gaussian phase estimation'' procedure that can estimate one entry of $Aw$ 
with a complexity that is similar to standard phase estimation, but with well-behaved \emph{sub-Gaussian} error. These well-behaved errors in individual entries typically still add up to a large $\ell_2$-error for the vector $Aw$ as a whole. However, with very high probability the error remains small in one or a few fixed directions---including the direction of the unknown top eigenvector. 
To speed up the computation of each entry, we split the rows into ``small'' and ``large'' entries, and handle them separately.
This divide-and-conquer approach uses $\bigOt{d^{1.75}}$ time in total for the $d$ entries of $Aw$ (\Cref{thm:qNPM_top_Gaussian}). This is asymptotically worse than our second algorithm (described below), but we still feel it merits inclusion in this paper because it uses an intuitive entry-by-entry approach, it has a slightly better dependence on the precision than our second algorithm, and most importantly our new technique of Gaussian phase estimation may find applications elsewhere.

Our second algorithm is faster and more ``holistic''. It does not approximate the matrix-vector product $Aw$ entry-by-entry. Instead it implements a block-encoding of the matrix $A$, uses that to (approximately) produce $Aw$ as a $\log(d)$-qubit state, and then applies a subtle tomography procedure to obtain a classical estimate of the vector $Aw$ with small $\ell_2$-error.\footnote{In fact our second algorithm first implements a block-encoding of the rank-1 projector $\Pi=v_1v_1^\dagger$ using $\bigOt{1/\gamma}$ applications of an approximate block-encoding of the matrix $A$ (by applying QSVT \cite[Theorem 31]{gilyen2018QSingValTransf}), and then applies our state tomography algorithm to obtain a classical estimate of the vector $\Pi w$, which is proportional to $v_1$. This increases the spectral gap from $\gamma$ to $\Theta(1)$ and hence further improves our $\gamma$-dependency.}
When used in our version of the power method, this classical vector is then stored in a QRAM-data structure (see below) that makes it easy to prepare $Aw$ as a quantum state in the next iteration,  which applies $A$ again.
Our new tomography procedure incorporates ideas from~\cite{kerenidis2018QIntPoint} and~\cite{apeldoorn2022QTomographyWStatePrepUnis}; it matches the latter's essentially optimal query complexity, but improves upon both their time complexities. Doing tomography at the end of each iteration of the power method is somewhat expensive, but still leaves us with a time complexity of only $d^{1.5+o(1)}$ (\Cref{cor:qeigenspacesparse} with $q=1$ and sparsity set to $s=d$), which turns out to be a near-optimal quantum speed-up over classical, as our lower bounds (discussed below) imply.

It is worth highlighting our result about preparing the top eigenvector of $A$ as a $\log(d)$-qubit quantum state $\ket{v_1}$, which we have ``under the hood'' in our second algorithm. This preparation can (with high success probability and small $\ell_2$-error) be done in time roughly $(d/\gamma)^{1+o(1)}$ (\Cref{cor:prepare_v1} with sparsity set to $s=d$). This is optimal up to the $o(1)$ in the exponent; maybe surprisingly, for constant gap preparing $\ket{v_1}$ is not significantly more expensive than the easier task of just approximating the top eigen\emph{value} $\lambda_1$, for which we prove an $\Omega(d)$ lower bound on the required number of quantum queries to entries of~$A$ (\Cref{prop:LB_top_eigenvalue} with sparsity set to $s=d$).
Note that putting some form of quantum state tomography on top of the preparation of $\ket{v_1}$ is \emph{not} enough to achieve the time complexity of our second algorithm, as it would result in a quantum algorithm that produces a classical description of (an approximation of) $v_1$ in time roughly $d^2$ rather than roughly $d^{1.5}$.

In both of our algorithms, the vector resulting in each iteration from our approximate matrix-vector multiplication will have small errors compared to the perfect matrix-vector product.
In the basic power method we cannot tolerate small errors in adversarial directions: if $w_0$ has roughly $1/\sqrt{d}$ overlap with the top eigenvector $v_1$, and we compute $Aw_0$ with $\ell_2$-error $>\lambda_1/\sqrt{d}$, then our approximation to the vector $Aw_0$ may have \emph{no overlap with $v_1$ at all anymore}! If this happens, if we lose the initially-small overlap with the top eigenvector, then the power method fails to converge to $v_1$ even if all later matrix-vector multiplications are implemented perfectly.
Fortunately, Hardt and Price~\cite{HP15NoisyPowerMethod} have already shown that the power method is robust against errors in the matrix-vector computations if they are sufficiently well-behaved (in particular, entrywise sub-Gaussian errors with moderate variance suffice). Much of the technical effort in our quantum subroutines for matrix-vector multiplication is to ensure that the errors in the resulting vector are indeed sufficiently well-behaved not to break the noisy power method.%
\footnote{For simplicity we assume here that $\gamma$ (or some sufficiently good approximation of it) is known to our algorithm. However, because we can efficiently approximate $|\lambda_1|$ 
(by~\cite[Lemma 50]{apeldoorn2017QSDPSolvers}, one can estimate $\lambda_1$ with additive error $\gamma/4$ using $\mathcal{\tilde{O}}(d^{1.5}/\gamma)$ time) and verify whether the output of our algorithm is approximately an eigenvector for this eigenvalue by one approximate matrix-vector computation, we can actually try exponentially decreasing guesses for $\gamma$ until the algorithm returns an approximate top eigenvector.

It should be noted that our algorithm has polynomial dependence on the precision $\eps$, namely linear in $1/\eps$, which is worse than the $\log(1/\eps)$ dependence of the classical power method with perfect matrix-vector calculations. This is the price we pay for our polynomial speed-up in terms of the dimension. For applications where the precision need not be extremely small, our polynomial dependence on~$\eps$ would be an acceptable price to pay.}

\paragraph{Finding the top-\texorpdfstring{$q$}{q} eigenvectors.}
Going beyond just the top eigenvector, the ability to find the top-$q$ eigenvectors (with $q\ll d$) is crucial for many applications in machine learning and data analysis, such as spectral clustering, principal component analysis, low-rank approximation of $A$, and dimensionality reduction of $d$-dimensional data vectors $w$ (e.g., projecting the data vectors onto the span of the top-$q$ eigenvectors of the covariance matrix $A=\Exp[ww^T]$). 
In most cases it suffices to (approximately) find the subspace spanned by the top-$q$ eigenvectors of $A$ rather than the individual top-$q$ eigenvectors $v_1,\ldots,v_q$, which is fortunate because distinguishing $v_1,\ldots,v_q$ can be quite expensive if the corresponding eigenvalues $\lambda_1,\ldots,\lambda_q$ are close together.

The noisy power method can also  (approximately) find the subspace spanned by the top-$q$ eigenvectors of $A$, assuming some known gap~$\gamma$ between the $q$th and $(q+1)$th eigenvalue. 
Using our knowledge of the gap, we give an algorithm to approximate $\lambda_q$ (\Cref{cor:est_q}).
Knowing $\lambda_q$ and this gap (at least approximately), we then show how a block-encoding of $A$ can be efficiently converted using quantum singular-value transformation~\cite{gilyen2018QSingValTransf} into a block-encoding of the projector $\Pi$ that projects onto the subspace spanned by the top-$q$ eigenvectors. In \Cref{ssec:tomolowrankreflections} we give a new almost optimal process-tomography algorithm for recovering the projector $\Pi$, assuming only the ability to apply (controlled) reflections $2\Pi-I$ about the rank-$q$ subspace that we are trying to recover (\Cref{thm:LearnProj}). This algorithm, applied to $\Pi$, gives us the subspace corresponding to the top-$q$ eigenvectors of~$A$. For constant eigenvalue gap and desired precision, it uses time $qd^{1.5+o(1)}$. In fact, what we above called our ``second algorithm'' for finding the top eigenvector is just the special case $q=1$. In the case where $A$ is $s$-sparse (meaning each row and column of $A$ has $\leq s$ nonzero entries) and we have sparse-query-access to it, the time complexity becomes $q\sqrt{s}d^{1+o(1)}$ (\Cref{cor:qeigenspacesparse}; this result implies the claim of the previous sentence by setting $s=d$).
If the pairwise spacing between the first $q$ eigenvalues is at least $\Omega(1/q)$, then we can also (approximately) \emph{find} each of the eigenvectors $v_1,\ldots,v_q$ individually, at the expensive of $\poly(q)$ more time.

As a byproduct of this algorithm we also obtain a qualitatively improved tomography procedure that works assuming the ability to reflect around the state that we want to estimate, but does not need the stronger assumption of being able to prepare that state.
%
%Say something about Lanczos?

\paragraph{Our computational model.} 
The computational model for our quantum algorithms is that we have query access to the entries of $A$,\rnote{edited the paragraphs below in response to Review 3} which can be, e.g., stored in quantum-accessible classical memory (``quantum read-only memory'', a.k.a.\ QROM). \anote{Made it a bit more general hinting at the possibility of computing matrix entries on the fly.}
The runtime of a quantum algorithm, or of a classical algorithm with quantum subroutines, is measured by the total number of queries and other elementary operations (one- and two-qubit gates, classical RAM-operations) on the worst-case input. 

Our algorithms also use $\bigOt{d}$ bits of quantum-accessible classical-writable memory (QRAM, also sometimes called QCRAM) in order to store classical descriptions of the intermediate $d$-dimensional approximating vectors. In analogy with classical RAM, a QRAM is a device that stores some $m$-bit string $z=z_0\ldots z_{m-1}$ and allows efficient access to the individual bits of $z$, also to several bits in superposition: a ``read'' operation corresponds to the unitary $O_z$ that maps $\ket{i,b}\to\ket{i,b\oplus z_i}$, for $i\in[m]-1$ and $b\in\{0,1\}$.
Note that the QRAM (and QROM) memory content itself is classical throughout the algorithm: it is just a string $z$, not a superposition of such strings. A ``write'' operation for the QRAM corresponds to changing a bit of $z$; this is a purely classical operation, and cannot happen in superposition.
If only ``read'' operations are allowed, this is also sometimes called a QROM (``quantum read-only memory''), and typically classical input is assumed to be given in this form for the type of algorithms we and many others consider. The intuition is that read and write operations should be executable cheaply, for instance in time $\mathcal{O}(\log m)$ if we put the $m$ bits of $z$ on the leaves of a depth-$\log m$ binary tree, and we treat an address $i$ as a bit-by-bit route from the root to the addressed leaf. 

It should be noted that QRAM and QROM are controversial notions in some corners of quantum computing. In the case of QROM one could circumvent the controversy by assuming that the entries to our input matrix $A$ are computed for us by a small circuit,\footnote{Such an efficient procedure will, however, have to rely on structure present in the data, and unfortunately does not apply in general to problems where one aims to process some real-world data.} but for our use of QRAM this is not an option. The issue is not so much that a circuit for $O_z$ for querying the $m$-bit string $z$ uses roughly $m$ gates (because the same is true for classical RAM and is not really considered an issue there) but that running such an operation fault-tolerantly on a superposition like $\sum_i \alpha_i\ket{i,0}$ seems to induce large overheads.
Classical RAM is not considered a problematic notion because nowadays it is implemented in very fast and practically error-free hardware without the need to do error-correction. It is conceivable that in the distant future quantum hardware implementations will become so good that one can similarly implement classical RAM on them with similar efficiency. Since one has to allow quantum superposition anyway in order to do anything in quantum computing, we feel that assuming QRAM is conceptually acceptable---especially for a theoretical computer science paper such as this one---though it is clearly not something for the small and noisy quantum computers we have today and in the near future.

We note, however, that in some situations we can actually avoid the use of QRAM altogether. In particular, we can make our second algorithm QRAM-free if we have some way of preparing the quantum state $\ket{w}$ corresponding to the intermediate approximating vectors (instead of using a KP-tree stored in QRAM for state-preparation). Given such a state-preparation procedure, the number of applications of the block-encoding $U_\Pi$ in \Cref{thm:LearnProj} remains the same.
Accordingly, if we only care about the number of queries to entries of the input matrix (instead of time complexity), then we can get an $\tilde{\mathcal{O}}(d^{1.5+o(1)})$-vs-$\Omega(d^2)$ quantum-classical query-complexity separation for approximating the top eigenvector without using any QRAM, because we can prepare $\ket{w}$ from its classical description using a circuit of $\bigOt{d}$ gates that uses no QRAM and no queries to entries of the input matrix. Also, our state-tomography procedure starting from ``conditional samples'' does not need any QRAM by itself.
Lastly, our result about preparing the top eigenvector of $A$ as a $\log d$-qubit quantum state $\ket{v_1}$ can be done without QRAM (see end of \Cref{sec:qspaceapproximate}).

\subsection{Our results: lower bounds}

We also show that our second quantum algorithm for finding the top eigenvector is essentially optimal, by proving an $\tilde{\Omega}(d^{1.5})$ quantum query lower bound for this task. We do this by analyzing a hard instance $A=\frac{1}{d}uu^T + N$, which hides a vector $u\in\{-1,1\}^d$ using a $d\times d$ matrix $N$ with i.i.d.\ Gaussian entries of mean~0 and standard deviation $\sim 1/\sqrt{d}$. Note that the value $u_iu_j/d$ has magnitude $1/d$, but is  ``hidden'' in the entry $A_{ij}$ by adding noise to it of much larger magnitude $\sim 1/\sqrt{d}$. One can show that (with high probability) this matrix has a constant eigenvalue gap, and its top eigenvector is close to~$u/\sqrt{d}$. 

First, this hard instance provides the above-mentioned unsurprising\footnote{A simpler way to see this lower bound is to consider the problem of distinguishing the all-0 matrix from a matrix that has a~1 in one of the $d^2$ positions and 0s elsewhere. This is just the $d^2$-bit OR problem, for which we have an easy 
and well-known $\Omega(d^2)$ classical query bound. However, the quantum analogue of this approach  only gives an $\Omega(d)$ lower bound, since the quantum query complexity of $d^2$-bit OR is $\Theta(d)$. Therefore we present a more complicated argument for the classical lower bound whose quantum analogue \emph{does} provide an essentially tight bound of $\tilde{\Omega}(d^{1.5})$.} $\Omega(d^2)$ query lower bound for \emph{classical} algorithms, as follows. To approximate the top eigenvector (and hence $u$), an algorithm has to recover most of the $d$ signs $u_i$ of~$u$. Note that the entries of the $i$th row and column of~$A$ are the only entries that depend on $u_i$. 
If the algorithm makes $T$ queries overall, then there is an index $\mathbf{i}$ such that the algorithm makes at most $4T/d$ queries to entries in the $\mathbf{i}$th row and column, while still recovering $u_{\mathbf{i}}$ with good probability. Slightly simplifying, one can think of each of those entries as a sample from either the distribution $N(1/d,1/{d})$ or the distribution $N(-1/d,1/{d})$, with the sign of the mean corresponding to $u_{\mathbf{i}}$. It is well known that $\Omega(d)$ classical samples are necessary to estimate the mean of the distribution to within $\pm 1/d$ and hence to distinguish between these two distributions. This implies $4T/d=\Omega(d)$, giving us the $T=\Omega(d^2)$ classical query lower bound for approximating the top eigenvector of~$A$ (\Cref{cor:classicalLB}).

Second, a similar but more technical argument works to obtain the $\tilde{\Omega}(d^{1.5})$ quantum query lower bound (\Cref{cor:quantumLB}), as follows.
A good algorithm recovers most $u_i$-s with good probability.
If it makes $T$ quantum queries overall, then there is an index $\mathbf{i}$ such that the algorithm has at most $4T/d$ ``query mass'' on the entries of the $\mathbf{i}$th row and column (in expectation over the distribution of $A$), while still recovering $u_{\mathbf{i}}$ with good probability. It then remains to show that distinguishing between either the distribution $N(1/d,1/{d})$ or the distribution $N(-1/d,1/{d})$, with the ability to query multiple samples from that distribution in quantum superposition, requires $\tilde{\Omega}(\sqrt{d})$ quantum queries.
This we prove by a rather technical modification of the adversary bound of Ambainis~\cite{Amb02AdvMethod,Amb06weightedAdv}, adjusted to inputs which are vectors of samples from a continuous distribution, using expectations under a joint distribution $\mu$ on pairs of matrices (the two marginal distributions of $\mu$ are our hard instance conditioned on $u_{\mathbf{i}}=1$ and $u_{\mathbf{i}}=-1$, respectively) of Hamming distance roughly $\sqrt{d}$ in the $\mathbf{i}$th row and column. 

\subsection{Related work}
\rnote{made some edits in the first two paragraphs and added the last paragraph in response to Review 1}

Our algorithms produce classical descriptions of the top eigenvector(s). This is very different from HHL-style~\cite{harrow2009QLinSysSolver} algorithms that return the vectors in the form of a $\log(d)$-qubit state whose vector of amplitudes is (proportional to) the desired vector.  Our approach contrasts for instance with algorithms like quantum PCA~\cite{lloyd2013QPrincipalCompAnal}, which can efficiently find the top-$q$ eigenvectors \emph{as quantum states} assuming the ability to prepare $A$ as a mixed quantum state. It also contrasts with the recent work of Seki and Yunoki~\cite{SY21QPM}, who show how to apply a given Hermitian matrix many times to a given quantum state $\ket{\psi}$, giving rise to a quantum version of the power method that outputs quantum states.\anote{Changed the following sentence since there are plenty ground state estimation algorithms.}
The most closely related quantum algorithm we are aware of is producing classical descriptions of eigenvectors in the special case where $A$ is symmetric and diagonally-dominant (SDD): Apers and de Wolf~\cite[Claim 7.12]{AdW19graphSpar} show how to approximately find the top-$q$ eigenvalues and eigenvectors of a dense SDD matrix $A$ in time $\bigOt{d^{1.5}+qd}$, using their quantum speed-up for Laplacian linear solving.\footnote{One can actually reduce the general symmetric $A$ to the case of an SDD matrix by defining $A'=A+cI$ for sufficiently large $c$ to make $A'$ SDD, and then renormalizing~$A'$ to operator norm $\leq 1$. The problem with this reduction is that $c$ could be as big as $d$ and then the eigenvalue gap of the new matrix is much smaller than that of~$A$.}

There is a fair amount of work on finding the largest (or smallest) \emph{eigenvalue} of a given Hamiltonian $A$, but the setting there is usually different and incomparable to ours: $A$ is viewed as acting on $\log(d)$ qubits, and classically specified as the sum of a small number of terms, each acting non-trivially on only $\mathcal{O}(1)$ of the qubits (this is the canonical QMA-complete problem). There is also work on finding the top eigenvalue of a given matrix in general (not necessarily a local Hamiltonian), for instance Lemma~50 of~\cite{apeldoorn2017QSDPSolvers}; some of these works even involve a version of the power method~\cite{NWlargesteigenvalue}. However, none of these methods readily generalizes to finding a classical description of the top eigenvector itself. 

A recent paper by Apers and Gribling~\cite[Theorem~5.1]{apers&gribling:QIPM} also gives a quantum algorithm for approximate matrix-vector multiplication. Their result is incomparable to ours: it uses a different norm to measure the approximation, and it is geared towards the case of tall-and-skinny sparse matrices~$A$; if instead the matrix~$A$ is $d\times d$ and dense (the regime we care about in this paper), then their time complexity can be roughly $d^{3.5}$, which is much worse than ours.
Their application area is also different from ours: it is to speed up interior-point methods for linear programs where the number of constraints is much larger than the number of variables.

One of our main tools for our upper bounds is a novel, essentially unbiased tomography procedure to estimate a $d$-dimensional pure state (``essentially unbiased'' here means the error vector's expectation is exponentially close to~0) from $\bigOt{d}$ ``conditional samples'' of the state, see our \Cref{sec:qTomoCond}. There have already been quantum algorithms for essentially unbiased mean estimation for $d$-dimensional random variables~\cite{CHJ:multivar,SZ:qspeedupstochastic}, and one might hope to use these to recover our tomography procedure in an easier way. \ynote{modified:} However, their input models are somewhat different from ours and it is not clear how to tweak their algorithms to recover an unbiased state-tomography procedure that works in our input model.
 In a nutshell, \cite{CHJ:multivar,SZ:qspeedupstochastic} require an oracle that outputs the random vectors in binary, i.e., with each coordinate explicitly written, while our algorithm only requires an exponentially smaller quantum state whose amplitudes represent the vector of interest.

\subsection{Future work}

Here we mention some questions for future work.
First, can we improve the $d$-dependence of our second algorithm from $d^{1.5+o(1)}$ to $d^{1.5}$? Note that the $o(1)$ comes from Low's Hamiltonian simulation result~\cite{low2018HamSimNearlyOptSpecNorm}; in his context it is also still an open question whether the $o(1)$ can be removed. Can we also improve our algorithm to use fewer or even no QRAM bits, while retaining the current dependence on $d$, $\gamma$, $\eps$?

Second, our upper bound of roughly $qd^{1.5}$ for finding the subspace spanned by the top-$q$ eigenvectors is essentially optimal for constant $q$, but it cannot be optimal for large $q$ (i.e., $q=\Omega(d)$) because diagonalization finds all $d$ eigenvectors exactly in time roughly $d^{2.37}$, which is less than $d^{2.5}$. We should try to improve our algorithm for large $q$.

Third, matrix-vector multiplication is a very basic and common operation in many algorithms. So far there has not been much work on speeding this up quantumly, possibly because easy lower bounds preclude quantum speed-up for \emph{exact} matrix-vector multiplication  (Footnote~\ref{note:LBexactmatrixvector}). Can we find other applications of our polynomially faster \emph{approximate} matrix-vector multiplication? One such application is
computing an approximate matrix-matrix product $AB$ in time roughly $d^{2.5}$, by separately computing $AB_i$ for each of the $d$ columns~$B_i$ of $B$. This would not beat the current-best (but wholly impractical) matrix-multiplication techniques, which take time $d^{2.37\ldots}$, but it would be a very different approach for going beyond the basic $\bigO{d^3}$ matrix-multiplication algorithm.
A related application is to matrix-product \emph{verification}: we can decide whether $AB$ is close to $C$ in Frobenius norm for given $d\times d$ matrices $A,B,C$, in quantum time roughly $d^{1.5}$, by combining our approximate matrix-vector computation with Freivalds's algorithm~\cite{freivalds:matrixmult}. This should be compared with the quantum algorithm of Buhrman and \v{S}palek~\cite{buhrman2006MatrixProduct} that tests if $AB$ is \emph{equal} to~$C$ (over an arbitrary field) using $\bigOt{d^{5/3}}$ time.

\section{Preliminaries}

Throughout the paper, $d$ always denotes the dimension of the ambient space $\mathbb{R}^d$ or $\mathbb{C}^d$, $\log$ without a base means the binary logarithm, $\ln=\log_e$ is the natural logarithm, and $\exp(f)=e^f$. 
We let $[d]$ denote the set $\{1,\ldots,d\}$ and $[d]-1$ denote the set $\{0,\ldots,d-1\}$.  
We denote by $\nrm{v}$ the $\ell_2$-norm of a vector, and by $\nrm{v}_p$ its $\ell_p$-norm.
For $A\in \mathbb{C}^{d\times d'}$, we define the spectral norm $\|A\|=\max\limits_{v\in\mathbb{C}^{d'}}\frac{\nrm{Av}}{\nrm{v}}$. For a set $S$, we define the indicator function $\mathds{1}_{S}$ as 
\begin{align*}
\mathds{1}_{S}(x)
	=\begin{cases}
		1 \text{ if } x\in S,\\
		0 \text{ otherwise.}
	\end{cases}
\end{align*}

The \emph{total variation distance} between probability distributions $P$ and $Q$ is defined as $d_{TV}(P,Q)=\sup_{A}P(A)-Q(A)$, where the supremum is over events $A$.
In particular, for discrete distributions we have $d_{TV}(P,Q)=\frac{1}{2}\sum_x |P(x)-Q(x)|$.
%, and for continuous distributions with respective probability density functions (pdf) $p(x),q(x)$ we have $d_{TV}(P,Q)=\frac{1}{2}\int |p(x)-q(x)|\text{d}x$.
We say that two random variables are \emph{$\delta$-close} to each other if the total variation distance between their distributions is at most $\delta$. 
For two distributions $P,Q$ over the same space, the \emph{relative entropy} $D_{KL}(P\,\|\,Q)$ (also called \emph{Kullback-Leibler divergence} or \emph{KL-divergence}) from $P$ to $Q$ is defined as
$$
D_{KL}(P\,\|\,Q)=\int p(x)\cdot\ln\frac{p(x)}{q(x)}\text{d}x=\E_{p}\left[\ln\frac{p(x)}{q(x)}\right],
$$
where $p(x),q(x)$ are the probability density functions (pdf) of $P$ and $Q$, respectively. In case $P$ is not absolutely continuous with respect to $Q$ we define $D_{KL}(P\,\|\,Q)=\infty$.

For us a projector is always a matrix $\Pi$ which is idempotent ($\Pi^2=\Pi$) and Hermitian ($\Pi^\dagger= \Pi$). This is sometimes called an ``orthogonal projector'' in the literature, but we drop the adjective ``orthogonal'' in order to avoid confusion with orthogonality between a pair of projectors. For a subspace $S$ we denote the unique (orthogonal) projector to $S$ by $\Pi_S$.

\subsection{Computational model and quantum algorithms}\label{ssec:compmodel}
Our computational model is a classical computer (a classical random-access machine) that can invoke a quantum computer as a subroutine.  The input is stored in quantum-readable read-only memory (a QROM), whose bits (or more generally, entries, if one entry is a number with multiple bits) can be queried.\footnote{This is just the standard quantum query model. In fact, the QROM input assumption can be relaxed and it suffices to have a black-box that computes the input bits on demand.} The classical computer can also write bits to a quantum-readable classical-writable classical memory (a QRAM). Such a QRAM, storing some $m$-bit string $w=w_0\ldots w_{m-1}$ admits \emph{quantum queries} (a.k.a.\ quantum read operations), which correspond to the unitary $O_w$ that maps $\ket{i,b}\to\ket{i,b\oplus w_i}$, where $i\in\{0,\ldots,m-1\}$ and $b\in\{0,1\}$. As already mentioned in the introduction, we think of one QRAM query as relatively ``cheap'', in the same way that a classical RAM query can be considered ``cheap'': imagine the $m$ bits as sitting on the leaves of a $\log(m)$-depth binary tree, then reading the bit at location $i$  intuitively just means going down the $\log(m)$-length path from the root to the leaf indicated by~$i$.

The classical computer can send a description of a quantum circuit to the quantum computer; the quantum computer runs the circuit (which may include queries to the input bits stored in QROM and to the bits stored by the computer itself in the QRAM), measures the full final state in the computational basis, and returns the measurement outcome to the classical computer. In this model, an algorithm has \emph{time complexity} $T$ if in total it uses at most $T$ elementary classical operations and quantum gates, quantum queries to the input bits stored in QROM, and quantum queries to the QRAM. The \emph{query complexity} of an algorithm only counts the number of queries to the input stored in QROM. We call a (quantum) algorithm \emph{bounded-error} if (for every possible input) it returns a correct output with probability at least $2/3$ (standard methods allow us to change this $2/3$ to any constant in $(1/2,1)$ by only changing the complexity by a constant factor). 

We will represent real numbers in computer memory using a number of bits of precision that is polylogarithmic in $d/\eps$ (i.e., $\bigOt{1}$ bits). This ensures that all numbers are represented throughout our algorithms with negligible approximation error and we will ignore those errors later on for ease of presentation.
In this paper, we are mainly interested in approximating the top eigenvector of a given matrix $A$. In the quantum case, we assume entries of $A$ are stored in a QROM, which we can access by means of queries to the oracle $O_A:\ket{ij}\ket{0}\rightarrow\ket{ij}\ket{A_{ij}}$.
In the special case where $A$ is $s$-sparse, meaning that each of its rows and columns has at most $s$-nonzero entries, we additionally assume we can also query the location of the $\ell$th nonzero entry in the $j$th column. This is called ``sparse-query-access to $A$'', and is a common assumption for quantum algorithms working on sparse matrices (for instance in Hamiltonian simulation). This corresponds to storing the matrix $A$ using an ``adjacency list'', i.e., the locations and values of the nonzero entries for each row and column in, a QROM.

\subsection{Quantum subroutines}

In this section we describe a few known quantum algorithms that we invoke as subroutines.

\begin{theorem}[\cite{gilyen2018QSingValTransf,yoder2014FixedPointSearch}, fixed-point amplitude amplification]\label{thm:Fixed_AA}
Let $a,\delta>0$, $U$ be a unitary that maps $\ket{0}\rightarrow \ket{\psi}$ and $R_\mathcal{A},R_{\ket{0}}$ be quantum circuits that reflect through subspaces $\mathcal{A}$ and (the span of) $\ket{0}$ respectively. Suppose $\|\Pi_{\mathcal{A}}\ket{\psi}\|\geq a$. There is a quantum algorithm that prepares $\ket{\psi'}$ satisfying $\|\ket{\psi'}-\frac{\Pi_{\mathcal{A}}\ket{\psi}}{\|\Pi_{\mathcal{A}}\ket{\psi}\|}\|\leq \delta$ using a total number of $\mathcal{O}(\log(1/\delta)/a)$ applications of $U,U^{-1}$, controlled $R_{\mathcal{A}},R_{\ket{0}}$ and additional single-qubit gates.
\end{theorem}

We also use the following theorem to help us find all marked items in a $d$-element search space with high probability. The theorem was implicit in~\cite{grover1996QSearch,boyer1998TightBoundsOnQuantumSearching}. For more details and for a better $d,\delta$-dependency, see \cite[Section 3]{vAGN23multifinding}.

\begin{theorem}\label{thm:findallsolutions}
    Let $f:[d] \rightarrow\{0, 1\}$ be a function that marks a set of elements $F = \{j \in [d] : f(j) = 1\}$, and $\delta\in(0,1)$. Suppose we know an upper bound $u$ on the size of $F$ and we have a quantum oracle $O_f$ such that $O_f:\ket{j}\ket{b} \rightarrow \ket{j}\ket{b\oplus f(j)}$. Then there exists a quantum algorithm that finds $F$ with probability at least $1-\delta$, using $\mathcal{O}(\sqrt{du}\cdot\poly\log(d/\delta))$ time.
\end{theorem}

\begin{theorem}[follows from Section~4 of \cite{brassard2002AmpAndEst}, amplitude estimation]\label{thm:amplitude_estimation}
Let $\delta \in(0,1)$. Given a natural number $M$ and access to an $(n + 1)$-qubit unitary $U$ satisfying
\[
U\ket{0^n}\ket{0}= \sqrt{a}\ket{\phi_0}\ket{0} +\sqrt{1-a}\ket{\phi_1}\ket{1},
\]
where $\ket{\phi_0}$ and $\ket{\phi_1}$ are arbitrary $n$-qubit states and $ a \in [0,1]$,
there exists a quantum algorithm that uses $\mathcal{O}(M\log(1/\delta))$ applications of $U$ and $U^\dagger$ and $\tilde{\mathcal{O}}(M\log(1/\delta))$ elementary gates, and outputs an estimator $\lambda$ such that, with probability $\geq 1-\delta$,
\[
|\sqrt{a}-\lambda| \leq \frac{1}{M}.
\]
\end{theorem}

\subsection{KP-tree and state-preparation}\label{sec:KPtree}

Here we introduce a variant of the QRAM data structure developed by Prakash and Kerenidis~\cite{Pra14,kerenidis2016QRecSys} for efficient state-preparation. We call this data structure a ``KP-tree'' (or $KP_v$ if we are storing the vector $v$) in order to credit Kerenidis and Prakash. The variant we use is very similar to the one used in~\cite{CdW21QLasso}.

\begin{definition}[KP-tree]\label{def:KPtree}
Let $v\in\mathbb{C}^d$. We define a KP-tree $KP_v$ of $v$ as follows:
\begin{itemize}
    \item The root stores the scalar $\nrm{v}$ and the size of the support $t=|supp(v)|$ of $v$.
    \item $KP_v$ is a binary tree on $\mathcal{O}(t\log d)$ vertices with depth $\lceil\log d\rceil$.    
    \item The number of leaves is $t$; for each $j\in supp(v)$ there is one corresponding leaf storing $v_j$.
    \item Each edge of the tree is labeled by a bit; the bits on the edges of the path from the root to the leaf corresponding to the $j^{th}$ entry of $v$ form the binary description of $j$.
    \item Intermediate nodes store the square root of the sum of their children's squared absolute values.
\end{itemize}

For $\ell\leq \lceil\log d\rceil$ and $j\in [2^\ell]-1$, we define $KP_v(\ell,j)$ as the scalar value stored in the $j^{th}$ node in the $\ell^{th}$ layer, i.e., the node that we can reach from the root by the path according to the binary representation of $j$. 
If there is no corresponding $j^{th}$ node in the $\ell^{th}$ layer (that is, we cannot reach a node by the path according to the binary representation of $j$ from the root), then $KP_v(\ell,j)$ is defined as $0$. Note that both the numbering of the layer and the numbering of nodes start from $0$. When $v$ is the all-0 vector, the corresponding tree consists of a single root node with $t=0$.
\end{definition}

If we have a classical vector $v\in \mathbb{C}^d$, we can build a KP-tree for $v$ using $\bigOt{d}$ time and QRAM bits. Given a KP-tree KP$_v$, we can efficiently query entries of $v$ and prepare the quantum state $\sum\limits_{j\in [d]-1}{\frac{v_j}{\|v\|}}\ket{j}$:

\begin{theorem}[Modified Theorem 2.12 of \cite{CdW21QLasso}]\label{thm:StatePrepare}
Suppose we have a KP-tree $KP_v$ of $v\in \mathbb{C}^d$, and we can apply a unitary $O_{KP_v}$ that maps $\ket{\ell,k}\ket{0}\rightarrow \ket{\ell,k}\ket{KP_v(\ell,k)}$. Then one can implement a unitary $U_v$ that maps $\ket{0}$ to $\ket{\psi_v}=\sum\limits_{j\in [d]-1}{\frac{v_j}{\|v\|}}\ket{j}$ up to error $\eps$ by using $\bigO{\log d}$ applications of $O_{KP_v}$, $O^\dagger_{KP_v}$, and $\bigOt{1}$ additional elementary gates.
\end{theorem}

Note that if $\norm{v}< 1$, then we can similarly prepare the quantum state $\ket{v}=\ket{0}\sum\limits_{j\in [d]-1}{v_j}\ket{j}+\ket{1}\ket{0}$ using $\bigOt{1}$ time and queries to KP$_v$.

\subsection{Block-encoding and Hamiltonian simulation}

Block-encoding embeds a scaled version of a (possibly non-unitary) matrix~$A$ in the upper-left corner of a bigger unitary matrix~$U$.

\begin{definition}
    Suppose that $A$ is a $2^w$-dimensional matrix, $\alpha ,\epsilon >0$, and $a \in \mathbb{N}$. 
    We call an $(a + w)$-qubit unitary $U$ an $(\alpha,a,\epsilon)$-block-encoding of $A$ if
$$
\|A -\alpha \bra{0^a} \otimes I_{2^w} )U(\ket{0^a}\otimes I_{2^w} )\|\leq \epsilon.
$$
\end{definition}

\begin{theorem}[\cite{low2016HamSimQubitization}]\label{thm:exp_simluation}
    Suppose that $U$ is an $(\alpha,a,\epsilon/|2t|)$-block-encoding of the Hamiltonian $H$. Then we can implement an $\epsilon$-precise Hamiltonian-simulation unitary $V$ which is a $(1,a+2,\epsilon)$-block-encoding of $e^{it H}$, with $\mathcal{O}(|\alpha t|+\log(1/\epsilon))$ uses of controlled-$U$ and its inverse, and with $\mathcal{O}(a|\alpha t|+a\log(1/\epsilon))$ additional elementary gates.
\end{theorem}

\begin{theorem}[{\cite[Theorem 2]{low2018HamSimNearlyOptSpecNorm}}]\label{thm:Low19OptimalHsim}
    Let $A$ be a $d\times d$ Hermitian matrix with operator norm~$\leq 1$, and $t>0$. Suppose $A$ has sparsity $s$ and we have sparse-query-access to~$A$.
      Then we can implement a unitary $U$ such that $\|U-\exp(iAt)\|\leq \epsilon$ using $\bigOt{t\sqrt{s}(t\sqrt{s}/\epsilon)^{o(1)}}$ time and queries.
\end{theorem}

If we do not have a sparse oracle or if $A$ is dense, then the time complexity of the above theorem simply becomes $\mathcal{\tilde{O}}((\sqrt{d}t)^{1+o(1)}/\epsilon^{o(1)})$ by setting $s=d$.

\begin{theorem}[{\cite[Corollary 71]{gilyen2018QSingValTransf}}]
    Let $\epsilon\in(0,1/2)$, $A$ be a $d\times d$ Hermitian matrix with operator norm~$\leq 1/2$, and $U=\exp(iA)$. Then we can implement a $(2/\pi,2,\epsilon)$-block-encoding of~$A$, using $\mathcal{O}(\log(1/\epsilon))$ applications of controlled-$U$, controlled-$U$ inverse, $\mathcal{O}(\log(1/\epsilon))$ time, and one auxiliary qubit.
\end{theorem}

Combining the above two theorems, we have the following theorem.

\begin{theorem}\label{thm:EntrytoBlock}
    Let $A$ be a $d\times d$ Hermitian matrix with operator norm~$\leq 1$. Suppose $A$ has sparsity $s$ and we have sparse-query-access to $A$. Then we can implement a unitary $U$ which is a $(4/\pi, 2, \epsilon)$-block-encoding of $A$ with $\bigOt{\sqrt{s}(s/\epsilon)^{o(1)}}$ time and queries. 
\end{theorem}

\begin{theorem}[\cite{apeldoorn2022QTomographyWStatePrepUnis}, Lemma 6]\label{thm:blockEncodeIP}
    Let $U=\sum_{x} U_x\otimes \ketbra{x}{x}$ and $V=\sum_{x} V_x\otimes \ketbra{x}{x}$ be controlled (by the second register) state-preparation unitaries, where $U_x:\ket{0}\ket{0^{\otimes a}}\rightarrow \ket{0}\ket{\psi_x}+\ket{1}\ket{\tilde{\psi}_x}$ and $V_x:\ket{0}\ket{0^{\otimes a}}\rightarrow \ket{0}\ket{\phi_x}+\ket{1}\ket{\tilde{\phi}_x}$ are $(a+1)$-qubit state-preparation unitaries for some (sub-normalized) $a$-qubit quantum states $\ket{\psi_x},\ket{\phi_x}$. Then $(I\otimes V^\dagger)(\text{SWAP}\otimes I_{2^{a+1}})(I\otimes U)$ is a $(1,a+2,0)$-block-encoding of the diagonal matrix diag$(\{\langle \psi_x|\phi_x\rangle\})$, where the SWAP gate acts on the first and second qubits.
\end{theorem}

\subsection{Singular-value and singular-vector-perturbation bounds}\label{sec:SVPert}

We invoke some tight bounds on the perturbation of singular values and singular vectors of a matrix.
In the following we order the singular values of a matrix $A$ in decreasing order, such that $i < j \Rightarrow \varsigma_i(A)\geq \varsigma_j(A)$.

\begin{theorem}[Weyl's singular value perturbation bound {\cite[Corollary III.2.6, Problem III.6.13]{bhatia1997MatrixAnalysis}}]\label{thm:Weyl}
	Let $A,B \in \mathbb{C}^{n\times m}$ be any matrices, then for all $i\in[n]$ we have
	\begin{align*}
		|\varsigma_i(A)-\varsigma_i(B)|\leq \nrm{A-B}.
	\end{align*}
\end{theorem}

In order to state the following perturbation bound on the singular-value subspaces we define $\Pi_{S}^{X}$ to be the projector onto the subspace spanned by the left-singular vectors of $X$ having singular values in $S$.

\begin{theorem}[Wedin-Davis-Kahan $\sin(\theta)$ theorem \cite{wedin1972DavisKahanSinThSingularVec}]\label{thm:Wedin}
	Let $A,B \in \mathbb{C}^{n\times m}$ be any matrices, and $\alpha,\delta\geq 0$, then\footnote{This bound is tight for any rank-$r$ projectors $A,B$, when $\alpha=0$ and $\delta=1$ due to \Cref{lem:ProjDiff}. Wedin's paper proves the statement for singular-vector subspaces that we use here, and the analogous statement for normal matrices is proven in Bhatia's book~\cite[Theorem VII.3.1]{bhatia1997MatrixAnalysis}, which actually also implies \Cref{thm:Wedin} with a bit of work.}
	\begin{align*}
		\nrm{(I-\Pi_{> \alpha}^{A})\Pi_{\geq \alpha + \delta}^{B}}\leq\frac{\nrm{A-B}}{\delta}.
	\end{align*}
\end{theorem}

\begin{lemma}[Operator norm equivalence to $\sin(\theta)$ between subspaces {\cite[after Exercise VII.1.11]{bhatia1997MatrixAnalysis}}]\label{lem:ProjDiff}
	Let $P,Q\in\mathbb{C}^{n \times n}$ be projectors with equal rank, then $\nrm{P-Q}=\nrm{P(I-Q)}=\nrm{(I-P)Q}$.
\end{lemma}

\subsection{Concentration inequalities}

Repeated sampling is very important for our quantum tomography algorithms, and here we describe some of the tail bounds we need.

\begin{proposition}[Bennett-Bernstein Bound {\cite[Theorem 2.9 \& Eqn.\ 2.10]{stephane2013ConcIneqThIndep}}]\label{prop:BennettB}
	Let $X^{(i)}\colon i\in [n]$ be independent random variables with finite variance such that, for each $i$,  $X^{(i)}\leq b$ for some $b>0$ almost surely (i.e., this event has probability measure~1). Let 
	\begin{align*}
		S=\sum_{i=1}^n X^{(i)}-\mathbb{E}[X^{(i)}], \qquad v= \sum_{i=1}^n \mathbb{E}[(X^{(i)})^2],
	\end{align*}
	then for any $t>0$,
	\begin{align*}
		\Pr[S\geq t]
		\leq\exp\left(-\frac{v}{b^2}h\left(\frac{bt}{v}\right)\right)
		\leq\exp\left(-\frac{t^2}{2v+\frac{2}{3}bt}\right),
	\end{align*}	
	where $h(x)=(1+x)\ln(1+x)-x$.
\end{proposition}

\begin{proposition}[{Chernoff-Hoeffding Bound \cite{chernoff1952Bound},
		\cite[Theorem~1]{hoeffding1963ProbIneqSumsOfBoundedRVs},
		\cite[Section~2.6]{stephane2013ConcIneqThIndep}}]
	\label{prop:ChernoffH}
	Let $0\leq\! X\!\leq 1$ be a bounded random variable and $p:=\mathbb{E}[X]$. Suppose we take $n$ i.i.d.\ samples $X^{(i)}$ of $X$ and denote the normalized outcome by $s=\frac{X^{(1)}+X^{(2)}+\cdots+X^{(n)}}{n}$. Then we have for all $\eps > 0$
	\begin{align}
		\Pr[s\geq p+\eps]	&\leq e^{-D_{KL}(p+\varepsilon\parallel p) n} \leq \exp\left(- \frac{\eps^2}{2(p+\eps)}n\right),\label{eq:chernoffMore}\\
		\Pr[s\leq p-\eps]	&\leq e^{-D_{KL}(p-\varepsilon\parallel p) n} \leq \exp\left(- \frac{\eps^2}{2p}n\right),\label{eq:chernoffLess}
	\end{align}
	where $D_{KL}(x\parallel p) = x \ln \frac{x}{p} + (1-x) \ln \left (\frac{1-x}{1-p} \right )$ is the Kullback–Leibler divergence between Bernoulli random variables with mean $x$ and $p$ respectively.
\end{proposition}

\begin{proof}
	The first inequality is \cite[Theorem 1]{hoeffding1963ProbIneqSumsOfBoundedRVs}, while \eqref{eq:chernoffLess} follows from \eqref{eq:chernoffMore} by considering $1-X^{(i)}$. The rightmost inequalities come from the observation that 
	$\forall x, y\geq 0\colon D_{KL}(x \parallel y) \geq \frac{(x-y)^2}{2\max\{x,y\}}$. 
\end{proof}

\begin{corollary}[Okamoto-Hoeffding Bound]\label{cor:SqrtChernoff}
	Let $X$, $s$ be as in \Cref{prop:ChernoffH}, then we have %that
	\begin{align*}
		\Pr[\sqrt{s}\geq \sqrt{p}+\eps)]	&\leq \exp\left(- 2\eps^2 n\right),\\
		\Pr[\sqrt{s}\leq \sqrt{p}-\eps]	&\leq \exp\left(- \eps^2 n\right).
	\end{align*}	
\end{corollary}

\begin{proof}
	This directly follows from \eqref{eq:chernoffMore}-\eqref{eq:chernoffLess} using the observation \cite{okamoto1958IneqSqrtBinomial} that 
	\begin{align*}
		D_{KL}(x\parallel p) &\geq  2\left(\sqrt{x}-\sqrt{p}\right)^2 \qquad \forall\,\, 0\leq p \leq x \leq 1,\\
		D_{KL}(x\parallel p) &\geq \phantom{2}\left(\sqrt{p}-\sqrt{x}\right)^2 \qquad \forall\,\, 0\leq x \leq p \leq 1.\tag*{\qedhere}
	\end{align*}
\end{proof}

\subsection{Matrix concentration inequalities}\label{sec:MatConc}

We state some random matrix concentration results for tall matrices (i.e., matrices having more rows than columns), but in our case we mostly apply them to flat matrices (having more columns than rows), thus we effectively apply the statements to $G^\dagger$. We will use the following non-asymptotic bounds.

\begin{theorem}[Well-conditioned tall Gaussian matrices {\cite[Theorem 5.39 \& Footnote 25]{vershynin2012IntroNonAsyAnalyRandMat}}]\label{thm:wellConditioned}
	There exist absolute constants\footnote{\label{foot:realVsComplexBound}In the real case \cite[Theorem II.13]{davidson2001RandMatBanachSpaces} gives $c,C=1$. While \cite[Footnote 25]{vershynin2012IntroNonAsyAnalyRandMat} asserts that the statement can be adapted to the complex case, there are no specifics provided, and the proof of \cite[Corollary 5.35]{vershynin2012IntroNonAsyAnalyRandMat} is borrowed from \cite[Theorem II.13]{davidson2001RandMatBanachSpaces}, where the adaptation of the statement to the complex case is presented as an open question. It is tempting to try and adapt the proof of the real case using the observation that $\varsigma_{\min}(G)=\min_{\nrm{u}=1}\max_{\nrm{v}=1}\Re\ipd{u^*}{Gv}$ together with the Slepian-Gordon lemma \cite[Theorem 1.4]{gordon1985InequalitiesGaussianProcesses}, however this approach seems to irrecoverably fail due to a banal issue: while $\nrm{\ket{u}\!\ket{v}-\ket{u'}\!\ket{v'}}^2 \leq \nrm{u-u'}^2+\nrm{v-v'}^2$ holds for real unit vectors, for complex unit vectors only the weaker $\nrm{\ket{u}\!\ket{v}-\ket{u'}\!\ket{v'}}^2 \leq 2\nrm{u-u'}^2+2\nrm{v-v'}^2$ holds in general. Indeed, for any $\alpha\in (0,2\pi)$ consider the complex numbers $u=1, v=i\exp(i \alpha), u'=\exp(-i \alpha), v'=i$, then we have $|uv-u'v'|^2/(|u-u'|^2+|v-v'|^2)=1+\cos(\alpha)$.} $c,C \geq 1$ such that the following holds for all $N\geq n$: if $G\in\mathbb{C}^{N\times n}$ is a random matrix whose matrix elements have i.i.d.\ real or complex standard normal distribution,\footnote{A complex standard normal random variable has independent real and imaginary parts each having centered normal distribution with variance $\frac{1}{2}$.} then its smallest singular value $\varsigma_{\min}(G)$ and largest singular value $\varsigma_{\max}(G)$ satisfy, for all $t\geq 0$
	\begin{align*}
		\Pr[\sqrt{N}-C\sqrt{n}-t < \varsigma_{\min}(G) \le \varsigma_{\max}(G) < \sqrt{N}+C\sqrt{n}+t] > 1-2\exp(-t^2/(2c)).
	\end{align*}
\end{theorem}

\begin{corollary}\label{cor:wellConditioned}
	In the setting of \Cref{thm:wellConditioned}, if $N\geq 16 C^2 n$, we have
	\begin{align*}
		\Pr\left[\frac{1}{4}\sqrt{N}< \varsigma_{\min}(G) \le \varsigma_{\max}(G)< \frac{7}{4}\sqrt{N}\right] > 1-2\exp(-N/(8c)).
	\end{align*}
\end{corollary}
\begin{proof}
	Apply \Cref{thm:wellConditioned} with $t=\sqrt{N}/2$.
\end{proof}

The following slightly tighter bound can be used for bounding the norm of individual columns or rows of Gaussian random matrices.

\begin{prop}\label{prop:rndVecLength}
	Let $v$ be an $n$-dimensional random vector whose coordinates have i.i.d.\ real or complex standard normal distribution, then $\mathbb{E}[\nrm{v}]\leq\!\sqrt{n}$ and $\Pr[\nrm{v} \geq\! \sqrt{n} + t]\leq \exp(-\frac{t^2}{2})$ $\forall t\geq 0$.
\end{prop}

\begin{proof}
	We have $\mathbb{E}[\nrm{v}]\leq\sqrt{\mathbb{E}[\nrm{v}^2]}=\sqrt{n}$ by Jensen's inequality. The function $v\mapsto \nrm{v}$ is $1$-Lipschitz due to the triangle inequality, and hence by the concentration of Lipschitz functions on vectors with the canonical Gaussian measure (Proposition 2.18 \& Equation (2.35) of~\cite{ledoux2001ConcentrationOfMeasure}) we have $\Pr[\nrm{v} \geq \sqrt{n} + t]\leq \exp(-t^2/2)$.\footnote{Actually, in the complex case the upper bound is even stronger: $\exp(-t^2)$.}
\end{proof}

Next, we invoke a concentration bound for the operator norm of a random matrix with independent bounded rows.

\begin{theorem}[Independent bounded rows {\cite[Theorem 5.44 \& Remark 5.49]{vershynin2012IntroNonAsyAnalyRandMat}}]\label{thm:heavyRows}
	There exists an absolute constant $c' > 0$ such that the following holds. Let $G\in\mathbb{C}^{N\times n}$ be a random matrix whose rows $G_i$ are independent, each having mean~$0$ and a covariance matrix\footnote{If $\psi\in \C^n$ is a mean-$0$ random vector and $C=\mathbb{E}[\psi^\dagger \psi]$ is its covariance matrix, then the covariance matrix of the complex conjugate random variable $\psi^*$ is $\mathbb{E}[\psi^{T} \psi^*]=C^*$. On the other hand we have $\mathrm{Cov}(\Re(\psi))+\mathrm{Cov}(\Im(\psi))=\frac{C+C^*}{2}$, and therefore $\nrm{\mathrm{Cov}(\Re(\psi))+\mathrm{Cov}(\Im(\psi))}\leq\nrm{C}$. Thus we can apply \cite[Theorem 5.44 \& Remark 5.49]{vershynin2012IntroNonAsyAnalyRandMat} separately to the real and imaginary parts of the random vectors $G_i$, whence the extra (possibly sub-optimal) factor of $2$ in \eqref{eq:A*A_rows}.} with operator norm at most $S^2$, and almost surely $\nrm{G_i}_2\leq B$ for all $i\in[N]$.
	Then for every $t \ge 0$, with probability at least $1 - 2n\exp(-c't^2)$ one has
	\begin{equation}\label{eq:A*A_rows}
		\nrm{ G } \le 2|S|\sqrt{N} + t B.
	\end{equation}
\end{theorem}

In case $G$ is a real-symmetric Gaussian matrix, we have the following non-asymptotic bound.

\begin{theorem}[Symmetric Gaussian matrix {\cite[Corollary~3.9 with $\epsilon=0.25$]{BH14Gaussianrandom}}]\label{thm:GauOperator}
    Let $G\in\mathbb{R}^{d\times d}$ be a symmetric matrix with $G_{ij}=b_{ij}\cdot g_{ij}$, where the random variables $\{g_{ij}:i\geq j\}$ are i.i.d.\ $\sim N(0,1)$ and the $\{b_{ij}:i\geq j\}$ are arbitrary real scalars. Denote $b_{max}=\max\limits_i\sqrt{\sum\limits_{j}b_{ij}^2}$ and $b^*_{max}=\max\limits_{ij}{|b_{ij}|}$. Then for every $t\geq 0$,
    \[
    \Pr\left[\|G\|\geq 2.5\cdot b_{max}+\frac{7.5}{\ln(1.25)}b^*_{max}\sqrt{\ln d}+t\right]\leq \exp(-t^2/(4{b^*}^2_{max})).
    \]
\end{theorem}

\subsection{Bounds on random variables with adaptive dependency structure}\label{sec:MartBounds}

Here we present some useful bounds on random variables, where the dependency structure follows some martingale-like structure.

The first bound gives an intuitive total variation distance bound for ``adaptive'' processes. The main idea is to couple two adaptive random processes that are step-wise similar. We use the following folklore \cite[p.~1]{angel2021PairwiseOptCoupling} observation:

\begin{theorem}[Total variation distance and optimal coupling]\label{thm:TVcoupling}
	Let $X,Y$ be random variables. Then $d_{TV}(X,Y)\leq\eps$ iff there exist $\eps$-coupled random variables $\widetilde{X}$ and $\widetilde{Y}$ (possibly dependent) with the same distribution as $X$ and $Y$, respectively, such that $\Pr[\widetilde{X}\!\neq \!\widetilde{Y}]\leq \eps$.
\end{theorem}

\begin{corollary}\label{cor:TVcoupling}
    Let $X$ be a random variable and $A$ an event of the underlying probability space such that $\Pr[A]>0$. Then for $Y' = X | A$, the conditioned version of $X$ (i.e., $\Pr[Y'\in S]=\Pr[A\, \&\, X\!\in\! S]/\Pr[A]$ for all measurable sets $S$), we have that $d_{TV}(X,Y')\leq 1-\Pr[A]$.
\end{corollary}
\begin{proof}
    Let $Y$ be a random variable that is independent of $A$, but its distribution is identical to that of $Y'$, i.e., $\Pr[Y'\in S]=\Pr[Y\in S]=\Pr[A\, \&\, X\!\in\! S]/\Pr[A]$ for all measurable sets $S$.
    In \Cref{thm:TVcoupling} take $\widetilde{X}:=X$, and $\widetilde{Y} := X$ on $A$ and $\widetilde{Y} := Y$ on the complement of $A$; by construction we have $\Pr[\widetilde{X}\!\neq \!\widetilde{Y}]\leq 1-\Pr[A]$.
    Finally, observe that for all measurable sets~$S$~we~have
    \begin{align*}
    \Pr[\widetilde{Y}\in S] & =\Pr[A\,\&\, X\in S]+\Pr[\bar{A}\,\&\, Y\in S]=
    \Pr[A\,\&\, X\in S]+\Pr[\bar{A}]\cdot\Pr[Y\in S]\\
    & =\left(1+\frac{\Pr[\bar{A}]}{\Pr[A]}\right)\Pr[A\,\&\, X\in S]
    =\frac{\Pr[A\,\&\, X\in S]}{\Pr[A]}=\Pr[X\in S\mid A]=\Pr[Y'\in S].\tag*{\qedhere}
    \end{align*}
\end{proof}

\begin{lemma}[Conditional total variation distance based bound]\label{lem:condTVBound}
	Let $X=(X_1,X_2)$ and $X'=(X'_1,X'_2)$ be two discrete random variables, and let $X_{2x}:=X_2|X_1\!=\!x$, $X'_{2x}:=X'_2|X'_1\!=\!x$.
	Suppose that $d_{TV}(X_1,X'_1)\!\leq\! \eps_1$ and $d_{TV}(X_{2x},X'_{2x})\!\leq\! \eps_2$ for all $x$ such that $\Pr[X_1=x]\Pr[X_1'=x]>0$, then $d_{TV}(X,X')\!\leq\! \eps_1+ (1-\eps_1)\eps_2$.
\end{lemma}

\begin{proof}
	Due to \Cref{thm:TVcoupling} we can find $\eps_1$-coupled random variables $\widetilde{X}_1$, $\widetilde{X}'_1$, and similarly $\eps_2$-coupled $\widetilde{X}_{2x}$, $\widetilde{X}'_{2x}$ for all $x$ in the range of $X_1,X'_1$ respectively. We can assume without loss of generality that $\widetilde{X}_1$ and $\widetilde{X}_{2x}$ are mutually independent for all $x$ in the range of $X_1,X'_1$ and likewise are $\widetilde{X}_1$ and $\widetilde{X}_{2x}$. 
 We then define $\widetilde{X}=(\widetilde{X}_1,\widetilde{X}_{2\widetilde{X}_1})$ and $\widetilde{X}'=(\widetilde{X}'_1,\widetilde{X}'_{2\widetilde{X}'_1})$, so that clearly $\Pr[X=(x_1,x_2)]=\Pr[\widetilde{X}=(x_1,x_2)]$ and $\Pr[X'=(x_1,x_2)]=\Pr[\widetilde{X}'=(x_1,x_2)]$. On the other hand due to the tight coupling of \Cref{thm:TVcoupling} we also get
	\begin{align*}
		d_{TV}(X,X') & \leq
		\Pr[\widetilde{X}\neq \widetilde{X}'] \\&
		=\Pr[\widetilde{X}_1\neq \widetilde{X}'_1] + \sum_x \Pr[\widetilde{X}_1 = \widetilde{X}'_1 = x \, \& \, \widetilde{X}_{2x}\neq\widetilde{X}'_{2x}] \\&
		=\Pr[\widetilde{X}_1\neq \widetilde{X}'_1] + \sum_x \Pr[\widetilde{X}_1 = \widetilde{X}'_1 = x] \Pr[\widetilde{X}_{2x}\neq\widetilde{X}'_{2x}] \\&
		\leq \eps_1 + \sum_x \Pr[\widetilde{X}_1 = \widetilde{X}'_1 = x] \eps_2 \\&    
		=\eps_1 + (1-\eps_1)\eps_2. \tag*{\qedhere}
	\end{align*}
\end{proof}

The following is essentially a martingale property, which could be stated more generally, but here we prove a simple version for completeness. 

\begin{lemma}[Martingale-like covariance sum]\label{lem:condCovBound}
	If $X,Y,Z\in \mathbb{C}^d$ are vector-valued discrete random variables such that $\mathbb{E}[Z|(X,Y)]=0$ (i.e., $\mathbb{E}[Z|(X=x,Y=y)]=0$ for all $x,y$), then $\mathrm{Cov}(X+Z) = \mathrm{Cov}(X) + \mathrm{Cov}(Z)$.
\end{lemma}

\begin{proof}
It is easy to see that $\mathbb{E}[Z]=0$, and we can assume without loss of generality that $\mathbb{E}[X]=0$. 
	\begin{align*}
		\mathrm{Cov}(X+Z)&=\mathbb{E}[\ketbra{X+Z}{X+Z}]\\&
		=\sum_{\underset{\Pr[(X,Y)=(x,y)]>0}{(x,y)\colon}}\Pr[(X,Y)=(x,y)]\mathbb{E}[\ketbra{x+Z}{x+Z}\,|(X,Y)=(x,y)]\\&
		=\sum_{\underset{\Pr[(X,Y)=(x,y)]>0}{(x,y)\colon}}\Pr[(X,Y)=(x,y)]\left(\ketbra{x}{x}+\mathbb{E}[\ketbra{Z}{Z}\,|(X,Y)=(x,y)]\right)\\&    
		= \sum_{x}\Pr[X=x]\ketbra{x}{x}     
		+\sum_{\underset{\Pr[(X,Y)=(x,y)]>0}{(x,y)\colon}}\Pr[(X,Y)=(x,y)]\mathbb{E}[\ketbra{Z}{Z}\,|(X,Y)=(x,y)]\\& 
		= \mathrm{Cov}(X) + \mathrm{Cov}(Z). \tag*{\qedhere}
	\end{align*}
\end{proof}

\subsection{Gaussian, Sub-Gaussian, and discrete Gaussian distributions}\label{sec:subGaussian}

A random variable $X$ over $\mathbb{R}$ has Gaussian distribution with mean $\mu=\E[X]$ and variance $\sigma^2=$Var$(X)$, denoted $X\sim N(\mu,\sigma^2)$, if its probability density function is
$$
p(x)=\frac{1}{\sqrt{2\pi \sigma^2}}\exp\Big(-\frac{(x-\mu)^2}{2\sigma^2}\Big), \text{ for }x \in \mathbb{R}.
$$
A well-known property of a Gaussian is that its tail decays rapidly, which can be quantified as:
\begin{itemize}
    \item If $X\sim N(\mu,\sigma^2)$, then for any $t > 0$, it holds that
    $$
    \Pr[X-\mu > t]\leq \frac{1}{\sqrt{2\pi}t}\exp(-t^2/(2\sigma^2)).
    $$
\end{itemize}
A random variable $X$ over $\R$ is called $\tau$-sub-Gaussian with parameter $\alpha>0$, denoted in short by $X \sim\,\tau$-subG$(\alpha^2)$, if %$E(X)=0$ and 
its moment-generating function satisfies
$\E[\exp(tX)] \leq \exp(\tau)\exp(\alpha^2t^2/2)$ for all $t\in\mathbb{R}$.
%A nice reference on sub-Gaussian random variables is [Rig15], which shows they have many useful properties similar to Gaussian distributions, and 
We list and prove a few useful properties of sub-Gaussian distributions below.

\begin{itemize}
        \item The tails of $X \sim \, \tau$-subG$(\alpha^2)$ are dominated by a Gaussian with parameter $\alpha$, i.e., 
        \begin{align}\label{eq:subGTail}
            \Pr[|X|> t] \leq 2\exp(\tau)\exp(-t^2/(2\alpha^2))\quad \forall t> 0.
        \end{align}
        
        \item If $X_i \sim\,\tau$-subG$(\alpha_i^2)$ are independent, then for any $a=(a_1,\ldots,a_d)^T \in \mathbb{R}^d$, the weighted sum $\sum\limits_{i\in [d]} a_i X_i$ is $(\sum_{i\in[d]}\tau_i)$-sub-Gaussian with parameter $\tilde{\alpha}=\sqrt{\sum\limits_{i=1}^d a_i^2\alpha_i^2}$. 
\end{itemize}
To prove \eqref{eq:subGTail}, we first use Markov's inequality to obtain that for all $s>0$
\[
\Pr[X\geq t] =\Pr[\exp(sX)\geq \exp(st)]\leq \E[\exp(sX)]/\exp(st)\leq \exp(\tau)\exp(\alpha^2s^2/2-st).
\]
Since the above inequality holds for every $s>0$, we have $\Pr[X\geq t]\leq \exp(-t^2/(2\alpha^2))$ because $\min\limits_{s>0}(\alpha^2s^2/2-st)=-t^2/(2\alpha^2)$. The same argument applied to $-X$ gives the bound on $\Pr[X\leq-t]$.

The second property can be easily derived using the independence of the $X_i$'s:
\begin{align*}
    \E[\exp(t \sum_{i=1}^d a_i X_i )]=\prod_{i=1}^d\E[\exp(t a_i X_i )]\leq \prod_{i=1}^d\exp(\tau_i)\exp(\alpha_i^2a_i^2t^2/2)=\exp(\sum_{i=1}^d\tau_i)\exp(\sum_{i=1}^d\alpha_i^2a_i^2t^2/2).
\end{align*}

Next, we explain what a \emph{discrete} Gaussian is. For any $s > 0$, we define the function $\propDens_s \colon \mathbb{R} \rightarrow \mathbb{R}$ as $\propDens_s(x) = \exp(-\pi x^2/s^2)$. 
When $s = 1$, we simply write $\propDens(x)$. For a countable set $S$ we define $\propDens_s(S) = \sum\limits_{a\in S} \propDens_s(a)$; if $\propDens_s(S)< \infty$, we define $\mathcal{D}_{S,s}=\frac{1}{\propDens_s(S)}\propDens_s$ to be the discrete probability distribution over $S$ such that the probability of drawing $x \in S$ is proportional to $\propDens_s(x)$. 
We call this the discrete Gaussian distribution over $S$ with parameter $s$. The following theorem states that such distributions over $\mathbb{Z}$ are sub-Gaussian with parameter $s$:

\begin{theorem}[{\cite[Lemma 2.8]{MP11Trapdoors}}]\label{thm:subGaussianDGS} For any $\lambda,\nu>0$ and $\tau\in(0,0.1]$, if $s\geq \lambda\sqrt{\log(12/\tau)/\pi}$,\footnote{Here the lattice we consider is the one-dimensional lattice $\lambda\mathbb{Z}$, so the length of the shortest nonzero vector $\lambda_1(\lambda\mathbb{Z})$ is just $\lambda$. Also, by using the equation between the smoothing parameter and the length of the shortest vector (\cite[Lemma 3.3]{MR04GaussianMeasure}), we have $\eta_\tau(\lambda\mathbb{Z})\leq \sqrt{\log(2+2/\tau)/\pi}\cdot\lambda_1(\lambda\mathbb{Z})\leq \lambda\sqrt{\log(3/\tau)/\pi}$. Therefore, if $s\geq \lambda\sqrt{\log(3/\tau)/\pi}$, then $s\geq \eta_\tau(\lambda\mathbb{Z})$ and hence $\mathcal{D}_{\lambda\mathbb{Z}+\nu,s}$ is $\log((1+\tau)/(1-\tau))$-sub-Gaussian. By the fact $\log((1+\tau)/(1-\tau))\leq 4\tau$ and changing $\tau\rightarrow \tau/4$, we get the theorem as stated here.} then $\mathcal{D}_{\lambda\mathbb{Z}+\nu,s}$ is $\tau$-sub-Gaussian with parameter $s$. Moreover, for every $s > 0$, $\mathcal{D}_{\lambda\mathbb{Z},s}$ is $0$-sub-Gaussian with parameter $s$. 
\end{theorem}

We also consider truncations of the above infinite discrete Gaussian distributions, and define $\mathcal{D}_{\mathbb{Z},s}^{[-L,R]}$ as $\mathcal{D}_{\mathbb{Z}\cap[-L,R],s}$. 
For every $\delta\in(0,1]$, if $L,R\geq s\sqrt{2\ln (1/\delta)}$, then by \Cref{cor:TVcoupling} we have $d_{TV}(\mathcal{D}_{\mathbb{Z},s},\mathcal{D}_{\mathbb{Z},s}^{[-L,R]})\leq 2\delta$, because the tail of $\mathcal{D}_{\mathbb{Z},s}$ can be bounded by $2\delta$ using \eqref{eq:subGTail} due to \Cref{thm:subGaussianDGS}.
We also define $\mathcal{D}_{\mathbb{Z},s}^{\text{mod } N}$ as a modular version of the discrete Gaussian distribution $\mathcal{D}_{\mathbb{Z},s}$, which has probability $\mathcal{D}_{\mathbb{Z},s}^{\text{mod } N}(k)=\mathcal{D}_{\mathbb{Z},s}(N\cdot \mathbb{Z}+k)$ for every $k\in\{-\lfloor N/2\rfloor,\ldots,0,\ldots,\lceil N/2\rceil-1\}$, and probability~0 for all other~$k\in \mathbb{Z}$. Again, by using \Cref{thm:subGaussianDGS} and \eqref{eq:subGTail}, we have that for every $\delta\in(0,1/2)$, if $N\geq 2s\sqrt{2\ln (1/\delta)}$, then $d_{TV}(\mathcal{D}_{\mathbb{Z},s},\mathcal{D}_{\mathbb{Z},s}^{\text{mod }N})\leq 2\delta$. 
Combining these we get $d_{TV}(\mathcal{D}_{\mathbb{Z},s}^{\text{mod }N},\mathcal{D}_{\mathbb{Z},s}^{[-L,R]})\leq 4\delta$ if  $N,L,R\geq 2s\sqrt{2\ln (1/\delta)}$. Finally, these arguments can also be applied to $\mathcal{D}_{\mathbb{Z}+c,s}$ with large enough~$s$.

\begin{corollary}\label{cor:DGSvariants}
    Let $\delta\in(0,1]$. For any $c>0$ and $\tau\in(0,0.1]$, if $s\geq \sqrt{\log(12/\tau)/\pi}$ and $N,L,R\geq 10s\sqrt{2\ln (2/\delta)}$, then $\mathcal{D}_{\mathbb{Z}+c,s}$, $\mathcal{D}_{\mathbb{Z}+c,s}^{[-L,R]}$, and $\mathcal{D}_{\mathbb{Z}+c,s}^{\text{mod }N}$ are $4\delta\exp(\tau)$-close to each other in total variation distance. 
\end{corollary}
Using the discussion and corollary above, we can show the truncated discrete Gaussian state is close to the modular discrete Gaussian state when $N,s$ are both reasonably large.

\begin{theorem}\label{thm:truroughlymod}
    Let $\delta\in(0,0.1]$, $s\geq 8\sqrt{2\log(\frac{1}{\delta})}$, $N\geq 16s\sqrt{2\ln(\frac{1}{\delta})}$ even number, and $t\in [-\frac{N}{8},\frac{N}{8}]$. Let $f:\mathbb{R}\rightarrow \mathbb{C}$ be an arbitrary phase function such that $|f(x)|=1$ for every $x\in \mathbb{R}$ and 
    \begin{align*}
        \ket{G}=& \frac{1}{\sqrt{G}}\sum\limits_{x\in\mathbb{Z}}f(x+t)\propDens_s(x+t)\ket{x},\\  \ket{G^{tr}}= &\frac{1}{\sqrt{G_{tr}}}\sum\limits_{x\in\{-\frac{N}{2}\ldots,0,\ldots,\frac{N}{2}-1\}}f(x+t)\propDens_s(x+t)\ket{x}, \\
    \ket{G^{{mod}}}=&\frac{1}{\sqrt{G_{mod}}}\sum\limits_{x\in\mathbb{Z}}f(x+t)\propDens_s(x+t)\ket{(x+\frac{N}{2}\text{ mod } N)-\frac{N}{2}}, 
    \end{align*}
    where $G, G_{mod}, G_{tr}$ are normalizing factors. Then $\ket{G^{tr}}, \ket{G^{{mod}}},$ $\ket{G}$ are $9\delta$-close to each other.\footnote{Equivalently, if $t\in[-5N/8,-3N/8]$  (and the constraints for $N,s$ remain the same), then $\ket{G}$, $\ket{(G^{tr})'}= \frac{1}{\sqrt{G_{tr}}}\sum\limits_{x\in[N]}f(x+t)\propDens_s(x+t)\ket{x},$
          $\ket{(G^{mod})'}=\frac{1}{\sqrt{G_{mod}}}\sum\limits_{x\in\mathbb{Z}}f(x+t)\propDens_s(x+t)\ket{x \text{ mod } N)}$ are $9\delta$-close to each other.
    } 
\end{theorem}

\begin{proof}
    Let $[\![\pm a ]\!]$ denote the set $\{-\lceil a\rceil\ldots,0,\ldots,\lceil a\rceil-1\}$. We first show $\ket{G^{tr}}$ is close to $\ket{G}$. Define $\ket{\widetilde{G^{tr}}}=\sqrt{\frac{G_{tr}}{G}}\ket{G^{tr}}$, which has $\ell_2$-norm $\sqrt{\frac{G_{tr}}{G}}\leq 1$. We can see 
    \begin{align*}
        \|\ket{\widetilde{G^{tr}}}-\ket{G}\|^2=\frac{1}{G}\sum\limits_{x\in \mathbb{Z}\setminus [\![\pm \frac{N}{2}]\!]} \propDens_s^2(x+t)\leq \frac{1}{G}\sum\limits_{x\in \mathbb{Z}\setminus \{x:|x+t|\leq N/4\}} \propDens_s^2(x+t)\leq \delta^2,
    \end{align*}
    where the last equality holds because $\propDens_s^2=\propDens_{s/\sqrt{2}}$ and $\mathcal{D}_{\mathbb{Z}+t,s/\sqrt{2}}$ is $\delta^2$-sub-Gaussian with parameter $s/\sqrt{2}$ by \Cref{thm:subGaussianDGS} (note $s/\sqrt{2}\geq 8\sqrt{\log(1/\delta)}\geq \sqrt{\log(12/\delta^2)/\pi}$) and because of the first property of sub-Gaussians: $\Pr[|x+t|>N/4]\leq 2\exp(\delta^2)\exp(-(N/4)^2/(2\cdot s^2/2))\leq \delta^2$. Note that $1\geq \|\ket{\widetilde{G^{tr}}}\|\geq \|\ket{G}\|-\|\ket{\widetilde{G^{tr}}}-\ket{G}\|\geq 1-\delta$ and hence $\|\ket{\widetilde{G^{tr}}}-\ket{{G^{tr}}}\|=\|( \sqrt{\frac{G_{tr}}{G}}-1)\ket{{G^{tr}}}\|\leq\delta$. 
    Therefore, we obtain 
    \[
    \|\ket{{G^{tr}}}-\ket{G}\|\leq\|\ket{\widetilde{G^{tr}}}-\ket{{G^{tr}}}\|+\|\ket{\widetilde{G^{tr}}}-\ket{G}\|\leq 2\delta.
    \]
    To show $\ket{G^{mod}}$ is close to $\ket{G}$, let us similarly define $\ket{\widetilde{G^{mod}}}=\sqrt{\frac{G_{mod}}{G}}\ket{G^{mod}}$.
    We can see
    \begin{align*}
        \|\ket{\widetilde{G^{mod}}}-\ket{G}\|^2\leq \frac{1}{G}\sum\limits_{x\in [\![\pm \frac{N}{2} ]\!]}\Big( \sum\limits_{y\in \mathbb{Z}\setminus\{0\}}\propDens_s(x+t+Ny)\Big)^2+\frac{1}{G}\sum\limits_{x\in \mathbb{Z}\setminus[\![\pm \frac{N}{2} ]\!]}\Big( \propDens_s(x+t)\Big)^2.
    \end{align*}
   To upper bound the first term in the RHS above, we split the domain of $x$ into three disjoint parts: $D_1=\{x\in [\![\pm \frac{N}{2} ]\!]: |x+t|\leq N/4\}$, $D_2=\{x\in [\![\pm \frac{N}{2} ]\!]: N/2\geq|x+t|> N/4\}$, and $D_3=\{x\in [\![\pm \frac{N}{2} ]\!]: |x+t|> N/2\}$. 
    
    When $x\in D_1$, we can see $\sum\limits_{y\in \mathbb{Z}\setminus\{0\}}\propDens_s(x+t+Ny)\leq \propDens_s(x+t)\cdot 2\sum\limits_{y=1}^\infty \delta^{2y}= \propDens_s(x+t)\cdot 2\delta^2/(1-\delta^2)$, because for every $y \in\mathbb{N}\cup\{0\}$ and $x\in D_1$ both
    \begin{align}\label{eq:constraint1}
        \propDens_s(x+t+(y+1)N)/ \propDens_s(x+t+yN)= \exp({-\frac{\pi}{s^2}\cdot N(2(x+t)+(2y+1)N}))\leq  \delta^2
    \end{align}
    and 
        \begin{align}\label{eq:constrain2}
        \propDens_s(x+t-(y+1)N)/\propDens_s(x+t-yN)= \exp({-\frac{\pi}{s^2}\cdot N(-2(x+t)+(2y+1)N)})\leq  \delta^2
    \end{align}
    hold (because $\exp({-\frac{\pi}{s^2}\cdot N(N\pm2(x+t)})\leq \exp({-\frac{\pi}{s^2}\cdot \frac{N^2}{2}})\leq \delta^2$ for every $x\in D_1$). 
    
    When $x\in D_2$ we use a similar argument, the only difference is that \Cref{eq:constraint1,eq:constrain2} now hold for every $y\in \mathbb{N}$ (excluding $0$), so 
    \[
    \sum\limits_{y\in \mathbb{Z}\setminus\{0\}}\propDens_s(x+t+Ny)\leq \propDens_s(x+t)\cdot 2\sum\limits_{y=0}^\infty \delta^{2y}=\propDens_s(x+t)\cdot 2/(1-\delta^2). 
    \]
    When $x\in D_3$, we can see that either $x+N$ or $x-N$ is $N/2$-close to (but $3N/8$-far from) $-t$, and without loss of generality and for simplicity we can assume $|x+N-(-t)|\in [3N/8, N/2]$. Then using a similar argument as for the $x\in D_2$, we can show that for every $x\in D_3$, 
    \[
    \sum\limits_{y\in \mathbb{Z}\setminus\{0\}}\propDens_s(x+t+Ny)\leq \propDens_s(x+N+t)\cdot 2/(1-\delta^2). 
    \]
    Since  $\propDens_s^2=\propDens_{s/\sqrt{2}}$ and $\mathcal{D}_{\mathbb{Z}+t,s/\sqrt{2}}$ is $\delta^2$-sub-Gaussian with parameter $s/\sqrt{2}$, we have
        \begin{align*}
        &\|\ket{\widetilde{G^{mod}}}-\ket{G}\|^2\leq \frac{1}{G}\sum\limits_{x\in D_1\cup  D_2\cup D_3}\Big( \sum\limits_{y\in \mathbb{Z}\setminus\{0\}}\propDens_s(x+t+Ny)\Big)^2 +\frac{1}{G}\sum\limits_{x\in \mathbb{Z}\setminus[\![\pm \frac{N}{2} ]\!]}\Big( \propDens_s(x+t)\Big)^2\\
        \leq & \frac{1}{G}\Big( (\frac{2\delta^2}{1-\delta^2})^2\!\sum\limits_{x\in D_1} \propDens_{\frac{s}{\sqrt{2}}}(x+t) + (\frac{2}{1-\delta^2})^2\!\sum\limits_{x\in D_2} \propDens_{\frac{s}{\sqrt{2}}}(x+t)+(\frac{2}{1-\delta^2})^2\!\sum\limits_{x\in D_3} \propDens_{\frac{s}{\sqrt{2}}}(x+N+t)+\delta^2G\Big)\\
        \leq & \frac{1}{G}\Big(\frac{4\delta^4}{(1-\delta^2)^2}G+\frac{8}{(1-\delta^2)^2}({\delta^2}G)+\delta^2G\Big)=\frac{4\delta^4+8\delta^2}{1-\delta^2}+\delta^2\leq 10\delta^2,
    \end{align*}
    where the second equality holds because 
    \begin{align*}
    \sum\limits_{x\in D_2} &\propDens_{s/\sqrt{2}}(x+t)+\sum\limits_{x\in D_3} \propDens_{s/\sqrt{2}}(x+N+t)\\
    &\leq \sum\limits_{x\in \mathbb{Z}\setminus\{x:|x+t|\leq N/8\}}\propDens_{s/\sqrt{2}}(x+t)\leq G\cdot 2\exp(\delta^2)\exp(-(N/8)^2/(2\cdot s^2/2))\leq\delta^2G.
    \end{align*}
    Note that $\|\ket{\widetilde{G^{mod}}}\|\in \|\ket{G}\|\pm\|\ket{\widetilde{G^{mod}}}-\ket{G}\|$ (implying $\|\ket{\widetilde{G^{mod}}}\|\in(1\pm \sqrt{10}\delta)$) and hence $\|\ket{\widetilde{G^{mod}}}-\ket{{G^{mod}}}\|=\|( \sqrt{\frac{G_{mod}}{G}}-1)\ket{{G^{mod}}}\|\leq\sqrt{10}\delta$. 
    As a result, we obtain $\|\ket{{G^{mod}}}-\ket{{G}}\|\leq \|\ket{\widetilde{G^{mod}}}-\ket{G}\|+\|\ket{\widetilde{G^{mod}}}-\ket{{G^{mod}}}\|\leq 2\sqrt{10}\delta< 7 \delta$.
    And by triangle inequality, we obtain $\|\ket{{G^{mod}}}-\ket{{G^{tr}}}\|\leq\|\ket{{G^{mod}}}-\ket{{G}}\|+\|\ket{{G}}-\ket{{G^{tr}}}\|\leq 9\delta$.
\end{proof}

\subsection{Fourier transform}

The Fourier transform $\hat{h}: \mathbb{R} \rightarrow \mathbb{C}$ of a function $h: \mathbb{R} \rightarrow \mathbb{C}$ is defined as 
\[
\hat{h}(\omega)=\int_{-\infty}^{\infty} h(x)\exp(-2\pi i x\omega)\,dx.
\]
The next facts follow easily from the above definition. If $h$ is defined as $h(\omega)=g(\omega+\nu)$ for some
function $g$ and value $\nu$, then we have
\[
\hat{g}(\omega)=\hat{h}(\omega)\exp(2\pi i \nu \omega).
\]
On the other hand, if $h(x)= g(x)\exp(2\pi i x\nu)$, then 
\[
\hat{h}(\omega)=\hat{g}(\omega-\nu).
\]
Another important fact is that the Fourier transform of $\propDens_s$ is $s\cdot\propDens_{1/s}$ for all $s>0$. Also, the sum of $\propDens_s(x)$ over $C\cdot\mathbb{Z}$ satisfies the Poisson summation formula~\cite[Lemma 2.14]{Reg09lattices}:

\begin{theorem}\label{thm:Poisson}
	For any scalar $C>0$ and any Schwartz function $f:\mathbb{R}\rightarrow \mathbb{C}$ (i.e., $f$ and each of its derivatives go to 0 faster than every inverse polynomial as the absolute value of the argument goes to infinity), 
	\[
	\sum\limits_{j\in C\cdot\mathbb{Z}}f( j)=C^{-1}\sum\limits_{j\in C^{-1}\cdot\mathbb{Z}}\hat{f}( j).
	\]
\end{theorem}

\subsection{Area of hyperspherical cap}

Let $B_d(r)$ denote the $d$-dimensional ball with radius $r$. The surface area of $B_d(r)$ is well-known to be $\frac{2\pi^{\frac{d}{2}}}{\Gamma(\frac{d}{2})}r^{d-1}$~\cite[p.~66]{Li11}, where $\Gamma$ is the Gamma function. Let $A_d^r(\phi)$ be the surface area of a hyperspherical cap in $B_d(r)$ with spherical angle $\phi$. The area of this hyperspherical cap can be calculated by integrating the surface area of a $(d-1)$-dimensional sphere with radius $r\sin\theta$~\cite[p.~67]{Li11}: 
\[
A_d^r(\phi)=\int_0^\phi 2A^{r\sin\theta}_{d-1}(\pi/2)r\text{d}\theta=\frac{2\pi^{\frac{d-1}{2}}}{\Gamma(\frac{d-1}{2})}r^{d-1}\cdot \int_0^\phi \sin^{d-2}\theta\text{d}\theta. 
\]
We abbreviate  $A_d(\phi):=A^1_d(\phi)$ for simplicity. 

Following Ravsky's computation in his reply to a question on StackExchange~\cite{Rav21Stack},
we now use the area of the hyperspherical cap to upper bound the probability of the event that a uniformly random vector~$u$ on $S^{d-1}$ only has a small overlap with another (fixed) unit vector~$v$. 

\begin{theorem}\label{thm:sphericalcapratio}
    Let $d\geq 3$ be and integer, $v\in \mathbb{R}^d$ be a unit vector, and $a \in [0,1]$. Then we have
    \[
    \Pr\limits_{u\sim S^{d-1}}\big[|\inProd{v}{u}|< a\big]\leq\frac{2}{\sqrt{\pi}}\cdot \frac{\Gamma(\frac{d}{2})}{\Gamma(\frac{d-1}{2})}\cdot a.
    \]
\end{theorem}

\begin{proof}
   Let $\phi\in [0,\pi/2]$ such that $\cos \phi=a$. We can see that if $|\inProd{v}{u}|\geq a$, then $u$ will be in the hyperspherical cap (whose center is $v$) with spherical angle $\phi$, and the probability that a uniformly-random $u$ lands in that hyperspherical cap is $\frac{A_d(\phi)}{A_d(\pi/2)}$. Therefore,
    \begin{align*}
        \Pr\limits_{u\sim S^{d-1}}\big[|\inProd{v}{u}|< a\big]=& \frac{A_d(\pi/2)-A_d(\phi)}{A_d(\pi/2)}=\Big(\frac{\pi^{\frac{d}{2}}}{\Gamma(\frac{d}{2})}\Big)^{-1}\cdot \Big(\frac{2\pi^{\frac{d-1}{2}}}{\Gamma(\frac{d-1}{2})}\Big)\int_\phi^{\pi/2} \sin\theta^{d-2}\text{d}\theta\\
        \leq& \Big(\frac{\pi^{\frac{d}{2}}}{\Gamma(\frac{d}{2})}\Big)^{-1}\cdot \Big(\frac{2\pi^{\frac{d-1}{2}}}{\Gamma(\frac{d-1}{2})}\Big)\int_\phi^{\pi/2} \sin\theta\text{d}\theta=\frac{2\Gamma(\frac{d}{2})}{\sqrt{\pi}\Gamma(\frac{d-1}{2})}\cdot (-\cos\frac{\pi}{2}+\cos\phi)\\
        =&\frac{2}{\sqrt{\pi}}\cdot \frac{\Gamma(\frac{d}{2})}{\Gamma(\frac{d-1}{2})}\cdot a.
    \end{align*}
\end{proof}

By using Legendre’s duplication formula $\Gamma(\frac{d}{2})\Gamma(\frac{d-1}{2})=\frac{\sqrt{\pi}}{2^{d-2}}\Gamma(d-1)$~\cite[Chap.~8.21, Eq.~102]{rudin1976principles} and the fact $\Gamma(d)=(d-1)!$, we obtain 
\begin{align*}
    \frac{\Gamma(\frac{d}{2})}{\Gamma(\frac{d-1}{2})}=\begin{cases}
        \frac{\Gamma(\frac{d}{2})^2}{\Gamma(\frac{d}{2})\Gamma(\frac{d-1}{2})}=\frac{2^{d-2}((\frac{d}{2}-1)!)^2}{\sqrt{\pi}(d-2)!}=\frac{2^{d-2}}{\sqrt{\pi}}\cdot \Big(^{d-2}_{\frac{d-2}{2}}\Big)^{-1}, \text{ if $d$ is even,}\\
        \frac{\Gamma(\frac{d}{2})\Gamma(\frac{d-1}{2})}{\Gamma(\frac{d-1}{2})^2}=\frac{\sqrt{\pi}(d-2)!}{2^{d-2}(\frac{d-3}{2}!)^2}=\frac{\sqrt{\pi}(d-2)}{2^{d-2}}\cdot \Big(^{d-3}_{\frac{d-3}{2}}\Big), \text{ if $d$ is odd.}
    \end{cases}
\end{align*}
Plugging the above into \Cref{thm:sphericalcapratio}, we have the following corollary. 
\begin{corollary}\label{cor:innerproductsphere}
        Let $d\geq 4$ be an integer, $v\in \mathbb{R}^d$ be a unit vector, and $c \geq 1$. Then we have
    \[
    \Pr\limits_{u\sim S^{d-1}}\big[|\inProd{v}{u}|< \frac{1}{c\sqrt{d}}\big]<\frac{1}{c}.
    \]
\end{corollary}

\begin{proof}
    By \Cref{thm:sphericalcapratio}, it suffices to show $\frac{2}{\sqrt{\pi}}\cdot \frac{\Gamma(\frac{d}{2})}{\Gamma(\frac{d-1}{2})}\cdot\frac{1}{\sqrt{d}}<1$ for every $d\geq 4$. When $d$ is even, by using Robbins' bound $\frac{4^m}{\sqrt{\pi m}}\exp(-\frac{1}{6m})\leq\binom{2m}{m}\leq \frac{4^m}{\sqrt{\pi m}}$~\cite[consequence of Eq.~1]{robbins1955RemarkOnStrilingsFormula}, we have 
    \begin{align*}
        \frac{\Gamma(\frac{d}{2})}{\Gamma(\frac{d-1}{2})}=\frac{2^{d-2}}{\sqrt{\pi}}\cdot \Big(^{d-2}_{\frac{d-2}{2}}\Big)^{-1}\leq\sqrt{\frac{d-2}{2}}\exp(\frac{1}{3d-6})\leq \sqrt{\frac{d}{2}}\exp(\frac{1}{6}),
    \end{align*}
    implying that $\frac{2}{\sqrt{\pi}}\cdot \frac{\Gamma(\frac{d}{2})}{\Gamma(\frac{d-1}{2})}\cdot\frac{1}{\sqrt{d}}\leq\sqrt{\frac{2}{\pi}}\cdot \exp(\frac{1}{6})<1$. Similarly, when $d$ is odd, we have 
    \begin{align*}
        \frac{\Gamma(\frac{d}{2})}{\Gamma(\frac{d-1}{2})}=\frac{\sqrt{\pi}(d-2)}{2^{d-2}}\cdot \Big(^{d-3}_{\frac{d-3}{2}}\Big)\leq{\frac{d-2}{\sqrt{2({d-3})}}},
    \end{align*}
    implying that $\frac{2}{\sqrt{\pi}}\cdot \frac{\Gamma(\frac{d}{2})}{\Gamma(\frac{d-1}{2})}\cdot\frac{1}{\sqrt{d}}\leq\sqrt{\frac{2}{\pi}}\cdot \frac{d-2}{\sqrt{d(d-3)}}<1$.
\end{proof}

\section{Time-efficient unbiased pure-state tomography}

In this section we design efficient methods for obtaining a good classical description of a pure quantum state (i.e., tomography), by manipulating and measuring multiple copies of that state.

\subsection{Pure-state tomography by computational-basis measurements}\label{sec:qTomoComp}

A direct corollary of \Cref{cor:SqrtChernoff} as observed in \cite{apeldoorn2022QTomographyWStatePrepUnis} is that computational-basis measurements yield a good approximation of the absolute values of the amplitudes of a (sub)normalized quantum state vector.

\begin{corollary}\label{cor:AbsTomoMeas}
	Suppose that $\eps,\delta\in(0,1]$, $\psi\in \mathbb{C}^d$ has $\ell_2$-norm at most~$1$, and we are given $n\geq \frac{1}{\eps^2}\ln\left(\frac{2d}{\delta}\right)$ copies of the pure quantum state $\ket{\varphi}:=\ket{\bar{0}}\ket{\psi}+\ket{\bar{0}^\perp}$, where $(\ketbra{\bar{0}}{\bar{0}}\otimes I)\ket{\bar{0}^\perp} = 0$.
	If we measure each copy in the computational basis and denote by
	$s_{i}$ the normalized number (i.e., frequency) of outcomes $\ket{\bar{0}}\ket{i}$
	then the vector 
	\begin{align*}
		\bar{\psi}_i:=\sqrt{s_{i}}
	\end{align*}
	with probability at least $1-\delta$ gives an $\eps$-$\ell_\infty$ approximation of $|\psi|$. Moreover, $\lVert\bar{\psi}\rVert_2\leq 1$ with certainty and if $\nrm{\psi}_2=1$, then also $\lVert\bar{\psi}\rVert_2=1$, and in general $\left|\lVert\bar{\psi}\rVert_2-\nrm{\psi}_2\right|\geq \eps$ holds with probability $\leq \frac{\delta}{d}$.
\end{corollary}

\begin{proof}
	By \Cref{cor:SqrtChernoff} we have that
	\begin{align*}
		\Pr\left[\left|\sqrt{s_i}-|\psi_i|\right|\geq\eps\right]\leq\frac{\delta}{d},
	\end{align*}
	and similarly 
	\begin{align*}
		\Pr\left[\left|\lVert\bar{\psi}\rVert_2-\lVert\psi\rVert_2\right|\geq\eps\right]
		=\Pr\left[\left|\sqrt{\sum_{i=0}^{d-1}s_i}-\sqrt{\sum_{i=0}^{d-1}|\psi_i|^2}\right|
		\geq\eps\right]\leq\frac{\delta}{d}.
	\end{align*}
	Finally, $\lVert\bar{\psi}\rVert^2_2=\!\sum_{i=0}^{d-1}s_i\leq 1$, where the last inequality is an equality if  $\nrm{\psi}_2=1.$
\end{proof}

\subsection{Pure-state tomography using conditional samples}\label{sec:qTomoCond}

Now we show how to produce an unbiased estimator of $\psi$ itself (not just of the magnitudes of its entries) with bounded variance using computational-basis measurements with the help of a reference state $\bar{\psi}$. Our approach is inspired by \cite{kerenidis2018QIntPoint} but improves over their biased estimator by making it unbiased.

\begin{lemma} \label{lem:UnbiasedTomoMeas}
	Suppose that $\psi\in \mathbb{C}^d$ has $\norm{\psi}_2\leq 1$, and we are given a copy of the state $\ket{\varphi'}:=\Big(\ket{+}\!\big(\ket{\bar{0}}\!\ket{\psi}+\ket{\bar{0}^\perp}\big)+\ket{-}\!\big(\ket{\bar{0}}\!\ket{\bar{\psi}}+\ket{\bar{0}'^\perp}\big)\Big)/\sqrt{2}$, where $\ket{\bar{0}}=\ket{0^a}$ for some $a\in\mathbb{N}$, $(\ketbra{\bar{0}}{\bar{0}}\otimes I)\ket{\bar{0}^\perp} = 0$, and $(\ketbra{\bar{0}}{\bar{0}}\otimes I)\ket{\bar{0}'^\perp} = 0$. If we measure $\ket{\varphi'}$ in the computational basis and denote by $X\in\{0,1\}^{2d}$ the indicator of the measurement outcomes $\ket{b}\ket{\bar{0}}\ket{i}$ (this $X$ is a weight-1 Boolean vector indexed by $(b,i)$ where $b\in \{0,1\}$ and $i\in[d]-1$), then the random vector $\psi'\in\mathbb{C}^d$ with coordinates
	\begin{align*}
		\psi'_i:=\frac{X_{0,i}-X_{1,i}}{|\bar{\psi}_i|}
	\end{align*}
	is an unbiased estimator of $\psi^{\Re}_i:=\mathrm{Re}\Big(\psi_i\frac{\bar{\psi}^*_i}{|\bar{\psi}_i|}\Big)$, with $\nrm{\psi'}_2\leq \frac{1}{\min\{|\bar{\psi}_i|\colon i\in [d]-1\}}$ with certainty, and covariance matrix $\mathrm{Cov}(\psi')=\frac{I}{2}+\mathrm{diag}\Big(\frac{|\psi_i|^2}{2|\bar{\psi_i|}^2}\Big)-\ketbra{\psi^{\Re}}{\psi^{\Re}}$.
\end{lemma}

\begin{proof}
	The probabilities of getting measurement outcomes $\ket{0}\ket{\bar{0}}\ket{i}$ and $\ket{1}\ket{\bar{0}}\ket{i}$ are
	\begin{align*}
		p_{0,i}&:=|\bra{0}\bra{\bar{0}}\bra{i}\ket{\varphi'}|^2
		=\left|\frac{\psi_i+\bar{\psi}_i}{2}\right|^2
		=\frac{|\psi_i|^2+\psi_i\bar{\psi}^*_i+\psi^*_i\bar{\psi}_i+|\bar{\psi}_i|^2}{4}
		=\mathbb{E}[X_{0,i}],\\
		p_{1,i}&:=|\bra{1}\bra{\bar{0}}\bra{i}\ket{\varphi'}|^2
		=\left|\frac{\psi_i-\bar{\psi}_i}{2}\right|^2
		=\frac{|\psi_i|^2-\psi_i\bar{\psi}^*_i-\psi^*_i\bar{\psi}_i+|\bar{\psi}_i|^2}{4}		
		=\mathbb{E}[X_{1,i}],
	\end{align*}
	and therefore 
	\begin{align*}
		\mathbb{E}[X_{0,i}-X_{1,i}]
		=p_{0,i}-p_{1,i}
		=\frac{\psi_i\bar{\psi}^*_i+\psi^*_i\bar{\psi}_i}{2}
		=|\bar{\psi}_i|\frac{\psi_i\frac{\bar{\psi}^*_i}{|\bar{\psi}_i|}+\psi^*_i\frac{\bar{\psi}_i}{|\bar{\psi}_i|}}{2}
		=|\bar{\psi}_i|\psi^{\Re}_i.
	\end{align*}
	For the norm bound observe that
	\begin{align*}
		\nrm{{\psi'}}_2
		\leq \sum_{i=0}^{d-1}|{\psi'}|_i
		\leq \sum_{i=0}^{d-1}\frac{X_{0,i}+X_{1,i}}{|\bar{\psi}|_i}	
		\leq \sum_{i=0}^{d-1}\frac{X_{0,i}+X_{1,i}}{\min\{|\bar{\psi}_i|\colon i\in [d]-1\}}
		\leq \frac{1}{\min\{|\bar{\psi}_i|\colon i\in [d]-1\}}.
	\end{align*}
	We can compute the covariance matrix directly as follows
	\begin{align*}
		\mathrm{Cov}(\psi')_{ij}
		=\mathbb{E}[\psi'_i\psi'_j]-\mathbb{E}[\psi'_i]\mathbb{E}[\psi'_j]
		=\delta_{ij}\frac{p_{0,i}+p_{1,i}}{|\bar{\psi}_i|^2}-\psi^{\Re}_i\psi^{\Re}_j=\delta_{ij}\frac{|\psi_i|^2/|\bar{\psi}_i|^2+1}{2}-\psi^{\Re}_i\psi^{\Re}_j,%\tag*{\qedhere}
	\end{align*}
    where the second equality uses that $X_{a,i}X_{b,j}=\delta_{ab}\cdot\delta_{ij}\cdot X_{a,i}$  because $X$ is a weight-$1$ Boolean vector.
\end{proof}

By an analogous argument as in the proof of \Cref{lem:UnbiasedTomoMeas}, we can obtain an unbiased estimator of the imaginary parts $\psi^{\Im}_j:=\mathrm{Im}\left(\psi_j\frac{\bar{\psi}^*_j}{|\bar{\psi}_j|}\right)$ (with the same $\ell_2$-norm and covariance matrix gaurantee) by measuring $\ket{\varphi''}:=\left(\ket{+}\left(\ket{\bar{0}}\ket{\psi}+\ket{\bar{0}^\perp}\right)+i\ket{-}\left(\ket{\bar{0}}\ket{\bar{\psi}}+\ket{\bar{0}'^\perp}\right)\right)/\sqrt{2}$ in the computational basis.

We now give a procedure (the first part of \Cref{thm:UnbiasedTomoMeas}) to find an unbiased estimator $\tilde{\psi}$ of $\psi$ that simultaneously has a good bound on the error of the estimator (with overwhelming probability) in some $k$ fixed directions using \Cref{lem:UnbiasedTomoMeas}.  The second part of \Cref{thm:UnbiasedTomoMeas} shows that the output $\tilde{\psi}$ will be close (in total variation distance) to an ``almost ideal'' unbiased estimator $\widecheck{\psi}$ that simultaneously has a good bound on the error of the estimator \emph{with certainty} in some $k$ fixed directions.
This will be used later when estimating a matrix-vector product $Aw$ in order to avoid an estimation-error that has too much overlap with $k$ of the eigenvectors of~$A$.

\begin{theorem} \label{thm:UnbiasedTomoMeas}
	Let $\psi\in \mathbb{C}^d$ such that $\nrm{\psi}_2\leq 1$, $\eps,\delta\in(0,1]$, $\eta\in\mathbb{R}_+$, $k\in\mathbb{N}$, $n\geq \frac{4d}{\eps^2}\left(\frac{4}{3}+\frac{1}{\eta}\right)\ln\left(\frac{8k}{\delta}\right)$. Suppose there exists a ``reference state'' (not necessarily known to the algorithm) $\bar{\psi}\in \mathbb{C}^d$ such that $|\bar{\psi}_j|^2\geq \max\{\frac{\eps^2}{d},\eta|\psi_j|^2\}$ $\forall j\in[d]-1$, and $\nrm{\bar{\psi}}\leq 1$. Given $n$ copies of the pure quantum states
	\begin{align*}
		\ket{\varphi'}&:=\left(\ket{+}\left(\ket{\bar{0}}\ket{\psi}+\ket{\bar{0}^\perp}\right)+\ket{-}\left(\ket{\bar{0}}\ket{\bar{\psi}}+\ket{\bar{0}'^\perp}\right)\right)/\sqrt{2},\\
		\ket{\varphi''}&:=\left(\ket{+}\left(\ket{\bar{0}}\ket{\psi}+\ket{\bar{0}^\perp}\right)+i\ket{-}\left(\ket{\bar{0}}\ket{\bar{\psi}}+\ket{\bar{0}'^\perp}\right)\right)/\sqrt{2},
	\end{align*} where $\ket{\bar{0}}=\ket{0^a}$ for some $a\in\mathbb{N}$, $(\ketbra{\bar{0}}{\bar{0}}\otimes I)\ket{\bar{0}^\perp} = 0$, and $(\ketbra{\bar{0}}{\bar{0}}\otimes I)\ket{\bar{0}'^\perp} = 0$, 
	if we measure each copy in the computational basis and denote by
	$s'_{b,j}$, $s''_{b,j}$ the normalized number of measurement outcomes $\ket{b}\ket{\bar{0}}\ket{j}$ from measuring the states $\ket{\varphi'}$ and $\ket{\varphi''}$ respectively, 
	then the random vector $\tilde{\psi}\in\mathbb{C}^d$ with coordinates
	\begin{align*}
		\tilde{\psi}_j:= \left(s'_{0,j}-s'_{1,j}+is''_{0,j}-is''_{1,j}\right)\frac{\bar{\psi}_j}{|\bar{\psi}_j|^2}
	\end{align*}
	is an unbiased estimator of $\psi$. Moreover, for every set $V=\{v^{(j)}\colon j\in[k]\}\subset\mathbb{C}^d$ of vectors we have
\begin{align}\label{eq:InnerConcentration}
		\Pr\left[\forall v\in V\colon |\braket{\tilde{\psi}-\psi|v}|<\frac{\eps}{\sqrt{d}}\nrm{v}_2\right]\geq 1-\delta.
	\end{align}
	In particular if $k\geq d$, then $\Pr\left[\nrm{\tilde{\psi}-\psi}_2<\eps\right]\geq 1-\delta$. 
 
    Finally, let $\{v^{(j)}\colon j\in[k]\}$ be a fixed set of orthonormal vectors and $\Pi_k$ be the projector to their span. Let $A$ be the event that $\exists j\in [k] \colon |\braket{\tilde{\psi}-\psi|v^{(j)}}| > \frac{\eps}{\sqrt{d}}$, $\bar{A}$ be the complement of $A$, and $X_\zeta$ be an independent Bernoulli random variable such that $\Pr[X_\zeta=0]=\zeta:=\frac{\delta-p}{1-p}$ for $p:=\Pr[A]$. Define $\widecheck{\psi}\in\mathbb{C}^d$ as follows\footnote{Here we introduce $X_\zeta$ to ensure that $\Pr[A\cup (X_\zeta=0)]$ exactly equals $\delta$, which is helpful because we use both upper and lower bounds on this probability in the proof.}
    \begin{align*}
        \widecheck{\psi} = \left\{\begin{array}{ll}
            \tilde{\psi} & \text{on } \bar{A}\cap (X_\zeta=1)\\
             (I-\Pi_k)\tilde{\psi}+ \sum\limits_{j\in[k]}\ket{v^{(j)}}\mathbb{E}[\braket{v^{(j)}|\tilde{\psi}}|A\cup (X_\zeta=0)] & \text{on } A\cup (X_\zeta=0).
            \end{array}\right.
    \end{align*}
    Then $\mathbb{E}[\widecheck{\psi}]=\psi$, $\Pr[\forall j\in[k] \colon |\braket{\widecheck{\psi}-\psi|v^{(j)}}|\leq \frac{k+3}{k} \frac{\eps}{\sqrt{d}} ]=1$ (which is why we call $\widecheck{\psi}$ an ``almost ideal'' unbiased estimator), the total variation distance between $\tilde{\psi}$, and $\widecheck{\psi}$ is at most $\delta$, and $\norm{\mathrm{Cov}(\Pi_k\widecheck{\psi})}\leq \norm{\mathrm{Cov}(\Pi_k\tilde{\psi})}+25\delta\eps^2\frac{k}{d}\leq\left(\frac{1}{4\ln\left(\frac{8k}{\delta}\right)}+25\delta k\right)\frac{\eps^2}{d}$.
\end{theorem}

\begin{proof}
	We prove the first part of \Cref{thm:UnbiasedTomoMeas} first. Let us define the random vectors $\psi',\psi''\in\mathbb{C}^d$ with coordinates
	\begin{align*}
		\psi'_j:=\frac{X'_{0,j}-X'_{1,j}}{|\bar{\psi}_j|}, \qquad
		\psi''_j:=\frac{X''_{0,j}-X''_{1,j}}{|\bar{\psi}_j|},
	\end{align*}
	where $X',X''\in\{0,1\}^{2d}$ denote the indicator of the measurements outcomes $\ket{b}\ket{\bar{0}}\ket{j}$ for the states $\ket{\varphi'}$ and $\ket{\varphi''}$, respectively. 
	Then by \Cref{lem:UnbiasedTomoMeas} and the discussion after the proof of \Cref{lem:UnbiasedTomoMeas}, $\psi', \psi''$ are unbiased estimators of $\psi^{\Re}_j:=\mathrm{Re}\left(\psi_j\frac{\bar{\psi}^*_j}{|\bar{\psi}_j|}\right)$, and 
	$\psi^{\Im}_j:=\mathrm{Im}\left(\psi_j\frac{\bar{\psi}^*_j}{|\bar{\psi}_j|}\right)$ respectively,
	such that $\nrm{\psi'}_2,\nrm{\psi''}_2\leq \sqrt{d}/\eps$ with certainty and 
	\begin{align}\label{eq:covPsis}
		\mathrm{Cov}(\psi')+\mathbb{E}[\psi']\mathbb{E}[\psi'^T]\preceq \left(\frac{1}{2}+\frac{1}{2\eta}\right)I, \quad
		\mathrm{Cov}(\psi'')+\mathbb{E}[\psi'']\mathbb{E}[\psi''^T]
		\preceq \left(\frac{1}{2}+\frac{1}{2\eta}\right)I,
	\end{align}
 where for the latter psd inequalities we used that $|\bar{\psi}_j|^2\geq\eta|\psi_j|^2$ and hence $\mathrm{diag}\Big(\frac{|\psi_i|^2}{2|\bar{\psi_i|}^2}\Big)\preceq \frac{1}{2\eta}I$.
 
	 Let $w\in \mathbb{R}^d$, then the random variables $\psi'_w:=\braket{\psi'|w},\psi''_w:=\braket{\psi''|w}$ satisfy $|\psi'_w|\leq \nrm{\psi'}_2\nrm{w}_2\leq \sqrt{d}\nrm{w}_2/\eps$, $|\psi''_w|\leq \nrm{\psi''}_2\nrm{w}_2\leq \sqrt{d}\nrm{w}_2/\eps$ with certainty. Also, $\mathbb{E}[|\psi'_w|^2]$ and $\mathbb{E}[|\psi''_w|^2]$ are both $\leq \left(\frac{1}{2}+\frac{1}{2\eta}\right)\nrm{w}_2^2$, because for both $\phi=\psi'$ and $\phi=\psi''$, we have
	\begin{align*}
		\mathbb{E}[|\braket{\phi|w}|^2]
  =\bra{w}\mathbb{E}[\phi \phi^T]\ket{w}
		=\bra{w}(\mathrm{Cov}(\phi)+\mathbb{E}[\phi]\mathbb{E}[\phi^T])\ket{w}
		\le\nrm{\mathrm{Cov}(\phi)+\mathbb{E}[\phi]\mathbb{E}[\phi^T]}\cdot\nrm{w}_2^2.
	\end{align*}
	Let $\Psi',\Psi''\in\mathbb{R}^d$ be the sum of $n$ i.i.d.\ copies of $\psi',\psi''$, respectively, obtained from the measurement outcomes of the $n$ copies of $\ket{\varphi'}$ and $\ket{\varphi''}$, so that 
	\begin{align}\label{eq:ImReDecomp}
		\tilde{\psi}_j=(\Psi'_j+i\Psi''_j)\frac{\bar{\psi}_j}{n|\bar{\psi}_j|}.
	\end{align}
	Let us analogously define $\Psi'_w:=\braket{\Psi'|w}$, and $\Psi''_w:=\braket{\Psi''|w}$. Then clearly $\mathbb{E}[\Psi'_w]=n\braket{\psi^\Re|w}$, and $\mathbb{E}[\Psi''_w]=n\braket{\psi^\Im|w}$. 
    For all $\tau\geq 1$, the Bennett-Bernstein tail bound (\Cref{prop:BennettB}) implies
    \begin{align}\label{eq:BBCor}
		\Pr\left[|\Psi'_w-n\braket{\psi^\Re|w}|\geq \tau\frac{\eps n \nrm{w}_2}{2\sqrt{d}}\right]
		&\leq2\exp\left(-\frac{\tau^2\eps^2 n^2 \nrm{w}_2^2/(4d)}{\left(1+\frac{1}{\eta}\right)\nrm{w}_2^2 n+\frac{\tau}{3}\nrm{w}_2^2 n}\right)\nonumber\\&
		\leq2\exp\left(-\tau\frac{\eps^2 n/(4d)}{\frac{4}{3}+\frac{1}{\eta}}\right)
		\leq 2\left(\frac{\delta}{8k}\right)^{\!\!\tau},
	\end{align}
	and similarly $\Pr[|\Psi''_w-n\braket{\psi^\Im|w}|\geq \tau\eps n \nrm{w}_2/\sqrt{4d}]\leq 2\left(\frac{\delta}{8k}\right)^{\!\tau}$.
	
	Let $\tilde{v}_j:=v_j\frac{\bar{\psi}_j^*}{|\bar{\psi}_j|}$, $\tilde{v}^\Re_j:=\mathrm{Re}(\tilde{v}_j)$, $\tilde{v}^\Im_j:=\mathrm{Im}(\tilde{v}_j)$, and observe that for any $v\in V$ 
	\begin{align*}
		&\braket{\tilde{\psi}|v} -\braket{\psi|v}
		=\frac{1}{n}\braket{\Psi'+i\Psi''|\tilde{v}}-\braket{\psi^\Re+i\psi^\Im|\tilde{v}}\\&
		\kern6mm=\!\Bigg(\underset{:= a\nrm{\tilde{v}^\Re}_2}{\underbrace{\frac{\braket{\Psi'|\tilde{v}^\Re}}{n}-\braket{\psi^\Re|\tilde{v}^\Re}}}\underset{ := b\nrm{\tilde{v}^\Im}_2}{\underbrace{-\frac{\braket{\Psi''|\tilde{v}^\Im}}{n}+\braket{\psi^\Im|\tilde{v}^\Im}}}\!\Bigg)
		+i\Bigg(\underset{:= c\nrm{\tilde{v}^\Re}_2}{\underbrace{\frac{\braket{\Psi'|\tilde{v}^\Im}}{n}-\braket{\psi^\Re|\tilde{v}^\Im}}}\underset{ := d\nrm{\tilde{v}^\Im}_2}{\underbrace{-\frac{\braket{\Psi''|\tilde{v}^\Re}}{n}+\braket{\psi^\Im|\tilde{v}^\Re}}}\Bigg),
	\end{align*}
	so 
	\begin{align*}
		|\braket{\tilde{\psi}|v}-\braket{\psi|v}|
		&\leq \sqrt{2}\max\left\{\left|\mathrm{Re}(\braket{\tilde{\psi}|v}-\braket{\psi|v})\right|,\left|\mathrm{Im}(\braket{\tilde{\psi}|v}-\braket{\psi|v})\right|\right\}\\&
        = \sqrt{2}\max\left\{\left|a\nrm{\tilde{v}^\Re}_2+b\nrm{\tilde{v}^\Im}_2\right|,\left|c\nrm{\tilde{v}^\Re}_2+d\nrm{\tilde{v}^\Im}_2\right|\right\}\\&
		\leq \sqrt{2}\max\{|a|,|b|,|c|,|d|\}(\nrm{\tilde{v}^\Re}_2+\nrm{\tilde{v}^\Im}_2)\\&
		\leq 2\max\{|a|,|b|,|c|,|d|\}\nrm{v}_2,
	\end{align*}
	where the last step uses that $|\tilde{v}_j|=|v_j|$ for all $j$, and Cauchy-Schwarz. Using Eq.~\eqref{eq:BBCor} four times with different choices of~$w$ , and the union bound over 4 events, we have $\max\{|a|,|b|,|c|,|d|\}< \frac{\tau\eps }{2\sqrt{d}}$ except with probability $\leq 8\left(\frac{\delta}{8k}\right)^{\!\tau}$. Eq.~\eqref{eq:InnerConcentration} now follows by choosing $\tau=1$ and taking the union bound over all $k$ vectors $v\in V$. If $k\geq d$ then we can apply the statement for a set $V$ containing the computational basis, and then $\nrm{\tilde{\psi}-\psi}_\infty\leq \frac{\eps }{\sqrt{d}}$ implies $\nrm{\tilde{\psi}-\psi}_2\leq \eps$ by Cauchy-Schwarz. 

\medskip
 
    Now we prove the second part of \Cref{thm:UnbiasedTomoMeas}, where $V$ is an orthonormal set (the ``Finally'' part). Note that the previous paragraph already proved that for every $v\in\mathbb{C}^d$ and for all $\tau\geq 1$, 
    \begin{equation}
    \Pr[|\braket{v|\tilde{\psi}-\psi}|> \tau{\frac{\eps}{\sqrt{d}}}\nrm{v}_2]\leq 8\left(\frac{\delta}{8k}\right)^{\!\!\tau}. \label{eq:inprodconcentration}
    \end{equation}
    Defining $\mathring{\eps}:=\eps/\sqrt{d}$ and $\mathring{\psi}^{j}:=\braket{v^{(j)}|\tilde{\psi}-\psi}$, we bound
    \begin{align}
        \mathbb{E}\left[|\mathring{\psi}^{j}|\mathds{1}_{(\mathring{\eps},\infty)}(|\mathring{\psi}^{j}|)\right]\!
        &=\sum_{\ell=0}^\infty\mathbb{E}\left[|\mathring{\psi}^{j}|\mathds{1}_{(2^\ell\mathring{\eps},2^{\ell+1}\mathring{\eps}]}(|\mathring{\psi}^{j}|)\right]%\\&
        \leq \sum_{\ell=0}^\infty2^{\ell+1}\mathring{\eps}\Pr\left[|\mathring{\psi}^{j}|\in(2^\ell\mathring{\eps},2^{\ell+1}\mathring{\eps}]\right]\kern25mm\nonumber\\&
        \leq \sum_{\ell=0}^\infty2^{\ell+1}\mathring{\eps}\Pr\left[|\mathring{\psi}^{j}|> 2^\ell\mathring{\eps}\right]%\\&    
        \leq \sum_{\ell=0}^\infty2^{\ell+1}\mathring{\eps}\,8\left(\frac{\delta}{8k}\right)^{\!\!2^\ell}\!\!\!\!\tag*{$\left(\!\Pr[|\mathring{\psi}^{j}|> \tau\mathring{\eps}]\leq 8\left(\frac{\delta}{8k}\right)^{\!\!\tau}\right)$\!}\nonumber\\&
        \leq \sum_{\ell=0}^\infty2^{\ell+1}\mathring{\eps}\,8\left(\frac{\delta}{8k}\right)^{\!\!\ell+1}\tag*{(since $\delta/k\leq 1$, and $2^\ell\geq \ell+1$)}\nonumber\\&
        = 2\frac{\delta\eps}{k\sqrt{d}} \sum_{\ell=0}^\infty\left(\frac{\delta}{4k}\right)^{\!\!\ell}%\\&   
        \leq 2\frac{\delta\eps}{k\sqrt{d}} \sum_{\ell=0}^\infty\left(\frac{1}{4}\right)^{\!\!\ell}\tag*{(since $\delta/k\leq 1$)}\nonumber\\&     
        < \frac{3\delta}{k} \frac{\eps}{\sqrt{d}}.\label{eq:expBound}
    \end{align}
    We have already proven Eq.~\eqref{eq:InnerConcentration}, implying that $p=\Pr[A]\leq \delta$.
    Observe that 
    \[
    \mathbb{E}[\braket{v^{(j)}|\tilde{\psi}}|A\cup (X_\zeta=0)]=\braket{v^{(j)}|\psi}+\mathbb{E}[\mathring{\psi}^{j}|A\cup (X_\zeta=0)], 
    \]
    and hence
    \[
    \mathbb{E}[\braket{v^{(j)}|\tilde{\psi}-\psi}|A\cup (X_\zeta=0)] = \mathbb{E}[\mathring{\psi}^{j}|A\cup (X_\zeta=0)].
    \]
    Since $\Pr[A\cup (X_\zeta=0)]=1-\Pr[\bar{A}\cap (X_\zeta=1)]=\delta$, using \eqref{eq:expBound} we have 
    \begin{align}\label{eq:exDev}
        |\mathbb{E}[\mathring{\psi}^{j}|A\cup (X_\zeta=0)]|
        \leq \mathbb{E}[|\mathring{\psi}^{j}|\mid A\cup (X_\zeta=0)]
        = \frac{\mathbb{E}[|\mathring{\psi}^{j}|\mathds{1}_{A\cup (X_\zeta=0)}]}{\delta}
        \leq \frac{\mathring{\eps}\delta+\frac{3\delta}{k} \frac{\eps}{\sqrt{d}}}{\delta}
        =\left(1+\frac{3}{k}\right) \frac{\eps}{\sqrt{d}}.
    \end{align}
    Since we modified $\tilde{\psi}$ on an event of probability $\delta$ to get $\widecheck{\psi}$, the total variation distance between the distributions of the random variables $\widecheck{\psi}$ and $\tilde{\psi}$ is at most $\delta$. The boundedness of~$\widecheck{\psi}$ is by construction and the unbiasedness is inherited from that of $\tilde{\psi}$, as follows: abbreviating the event $\bar{A}\cap (X_\zeta=1)$ to $B$, we have
    \begin{align*}
    \Exp[\widecheck{\psi}]& =\Pr[B]\cdot \Exp[\widecheck{\psi}\mid B] + \Pr[\bar{B}]\cdot \Exp[\widecheck{\psi}\mid \bar{B}]\\
   & =\Pr[B]\cdot \Exp[\tilde{\psi}\mid B] + \Pr[\bar{B}]\cdot \Exp[(I-\Pi_k)\tilde{\psi}+ \sum\limits_{j\in[k]}\ket{v^{(j)}}\mathbb{E}[\braket{v^{(j)}|\tilde{\psi}}]\mid \bar{B}]\\
      & =\Pr[B]\cdot \Exp[\tilde{\psi}\mid B] + \Pr[\bar{B}]\cdot \Exp[(I-\Pi_k)\tilde{\psi}+ \sum\limits_{j\in[k]}\ket{v^{(j)}}\bra{v^{(j)}}\cdot\tilde{\psi}\mid \bar{B}]\\
            & =\Pr[B]\cdot \Exp[\tilde{\psi}\mid B] + \Pr[\bar{B}]\cdot \Exp[(I-\Pi_k)\tilde{\psi}+ \Pi_k\tilde{\psi}\mid \bar{B}]\\
   &  =\Pr[B]\cdot \Exp[\tilde{\psi}\mid B] + \Pr[\bar{B}]\cdot \Exp[\tilde{\psi}\mid \bar{B}]=\Exp[\tilde{\psi}]=\psi.
\end{align*}
    Finally, defining $\bar{\eps}:=\eps\sqrt{\frac{k}{d}}$  and $\mathring{\psi}:=\Pi_k(\tilde{\psi}-\psi)=\sum_{j\in[k]}\mathring{\psi}^{j}v^{(j)}$ we bound
    \begin{align}
    \norm{\mathbb{E}\left[\ketbra{\mathring{\psi}}{\mathring{\psi}}\mathds{1}_{(\bar{\eps},\infty)}(\nrm{\mathring{\psi}}_2)\right]}
        &\leq\mathbb{E}\left[\nrm{\mathring{\psi}}_2^2\mathds{1}_{(\bar{\eps},\infty)}(\nrm{\mathring{\psi}}_2)\right]
        =\sum_{\ell=0}^\infty\mathbb{E}\left[\nrm{\mathring{\psi}}_2^2\mathds{1}_{(2^\ell\bar{\eps},2^{\ell+1}\bar{\eps}]}(\nrm{\mathring{\psi}}_2)\right]\nonumber\\&
        \leq \sum_{\ell=0}^\infty4^{\ell+1}\bar{\eps}^2\Pr\left[\nrm{\mathring{\psi}}_2\in(2^\ell\bar{\eps},2^{\ell+1}\bar{\eps}]\right]%\nonumber\\&
        \leq \sum_{\ell=0}^\infty4^{\ell+1}\bar{\eps}^2\Pr\left[\nrm{\mathring{\psi}}_2> 2^\ell\bar{\eps}\right]\nonumber\\&
        \leq \sum_{\ell=0}^\infty4^{\ell+1}\bar{\eps}^2 8k\left(\frac{\delta}{8k}\right)^{\!\!2^\ell}\tag*{$\left(\Pr[\nrm{\mathring{\psi}}_2> \tau\bar{\eps}]\leq 8k\left(\frac{\delta}{8k}\right)^{\!\tau}\right)$}\nonumber\\&
        \leq \sum_{\ell=0}^\infty4^{\ell+1}\bar{\eps}^28k\left(\frac{\delta}{8k}\right)^{\!\!\ell+1}\tag*{(since $\delta/k\leq 1$, and $2^\ell\geq \ell+1$)}\nonumber\\&
        = 4\delta\eps^2\frac{k}{d} \sum_{\ell=0}^\infty\left(\frac{\delta}{2k}\right)^{\!\!\ell}%\nonumber\\&
        \leq 4\delta\eps^2\frac{k}{d} \sum_{\ell=0}^\infty\left(\frac{1}{2}\right)^{\!\!\ell}\tag*{(since $\delta/k\leq 1$)}\nonumber\\&
        = 8\delta\eps^2\frac{k}{d}.\label{eq:SqDev}
    \end{align}
    This then implies that 
    \begin{align}
        \norm{\mathrm{Cov}(\Pi_k\tilde{\psi})-\mathrm{Cov}(\Pi_k\widecheck{\psi})}
        &=\norm{\mathbb{E}\left[\ketbra{\mathring{\psi}}{\mathring{\psi}}-\Pi_k\ketbra{\widecheck{\psi}-\psi}{\widecheck{\psi}-\psi}\Pi_k\right]}\nonumber \\&
        =\norm{\mathbb{E}\left[\left(\ketbra{\mathring{\psi}}{\mathring{\psi}}-\Pi_k\ketbra{\widecheck{\psi}-\psi}{\widecheck{\psi}-\psi}\Pi_k\right)\mathds{1}_{A\cup (X_\zeta=0)}\right]}\nonumber \\&
        \leq \norm{\mathbb{E}\left[\ketbra{\mathring{\psi}}{\mathring{\psi}}\mathds{1}_{A\cup (X_\zeta=0)}\right]}
        +\norm{\mathbb{E}\left[\Pi_k\ketbra{\widecheck{\psi}-\psi}{\widecheck{\psi}-\psi}\Pi_k\cdot \mathds{1}_{A\cup (X_\zeta=0)}\right]}\label{eq:CovTriangle}, 
    \end{align}    
       where the second equality is because $\tilde{\psi}=\widecheck{\psi}$ on the complement of the event $A\cup (X_\zeta=0)$.
       Using the definition of $\widecheck{\psi}$, and the fact that $\Pi_k(I-\Pi_k)=0$, we can see that $\Pi_k\widecheck{\psi}$ conditioned on $A\cup (X_\zeta=0)$ is actually a fixed vector $\mathbb{E}[\Pi_k\tilde{\psi}\mid A\cup (X_\zeta=0)]$, not a random variable anymore.
       We now have
       \[
       \norm{\mathbb{E}\left[\Pi_k\ketbra{\widecheck{\psi}-\psi}{\widecheck{\psi}-\psi}\Pi_k\cdot \mathds{1}_{A\cup (X_\zeta=0)}\right]}=\Pr[A\cup (X_\zeta=0)]\cdot\norm{\mathbb{E}\left[\Pi_k\ket{\tilde{\psi}-\psi}\mid {A\cup (X_\zeta=0)}\right]}^2_2.
       \]
       Continuing with Eq.~\eqref{eq:CovTriangle}, we have
     \begin{align*}  
     \hspace*{-0.8em}\norm{\mathrm{Cov}(\Pi_k\tilde{\psi})-\mathrm{Cov}(\Pi_k\widecheck{\psi})}
     &\leq\norm{\mathbb{E}\left[\ketbra{\mathring{\psi}}{\mathring{\psi}}\mathds{1}_{A\cup (X_\zeta=0)}\right]}+\delta \nrm{\mathbb{E}[\mathring{\psi}\mid A\cup (X_\zeta=0)]}_2^2\\&
        \leq\mathbb{E}\left[\nrm{\mathring{\psi}}_2^2\mathds{1}_{A\cup (X_\zeta=0)}\right]
        +16\delta\eps^2\frac{k}{d}\tag{by \eqref{eq:exDev} and $(1+3/k\leq4)$}\\&
        =  \mathbb{E}\left[\nrm{\mathring{\psi}}_2^2\mathds{1}_{A\cup (X_\zeta=0)}\Big(\mathds{1}_{[0,\bar{\epsilon}]}(\nrm{\mathring{\psi}}_2)+\mathds{1}_{(\bar{\eps},\infty)}(\nrm{\mathring{\psi}}_2)\Big)\right]
        +16\delta\eps^2\frac{k}{d}   \\&
        \leq \delta\bar{\eps}^2+8\delta\eps^2\frac{k}{d}+16\delta\eps^2\frac{k}{d}=25\delta\eps^2\frac{k}{d}.\tag*{(by \eqref{eq:SqDev} and $\Pr[A\cup (X_\zeta=0)]=\delta$)}
    \end{align*}  
We obtain 
\[
\norm{\mathrm{Cov}(\Pi_k\widecheck{\psi})} \leq  \norm{\mathrm{Cov}(\Pi_k\tilde{\psi})-\mathrm{Cov}(\Pi_k\widecheck{\psi})}+\norm{\mathrm{Cov}(\Pi_k\tilde{\psi})} \leq 25\delta\eps^2\frac{k}{d}+ \frac{\eps^2}{4d\ln\left(\frac{8k}{\delta}\right)}
\]
because $\norm{\mathrm{Cov}(\Pi_k\tilde{\psi})}\leq \norm{\mathrm{Cov}(\tilde{\psi})}=\frac{1}{n}\norm{\mathrm{Cov}(\psi')+\mathrm{Cov}(\psi'')}$, and the matrix inside the latter norm can be upper bounded by $2(\frac{1}{2}+\frac{1}{2\eta})I$ using Eq.~\eqref{eq:covPsis}.
\end{proof}

If we have $n$ conditional samples $\ket{\varphi}:=(\ket{0}\left(\ket{\bar{0}}\ket{\psi}+\ket{\bar{0}^\perp}\right)+\ket{1}\ket{\bar{0}}\ket{0})/\sqrt{2}$, then we can first use \Cref{cor:AbsTomoMeas} to produce (with success probability $\geq 1-\frac{\delta}{2}$) a $\frac{1}{\sqrt{d}}$-$\ell_\infty$ approximation $\psi'$ of the vector $|\psi|$ of the magnitudes of entries, which has $\nrm{\psi'}_2\leq 1$. Setting $\bar{\psi}_i:=\frac{|\psi'_i|+\frac{1}{\sqrt{d}}}{2}$, and building a KP-tree for $\bar{\psi}$ to be able to efficiently prepare a state that is coordinate-wise $\frac{1}{4\sqrt{d}}$-close to $\ket{\bar{\psi}}/\nrm{\bar{\psi}}_2$, we can transform the conditional copies $\ket{\varphi}$ to the form required by \Cref{thm:UnbiasedTomoMeas} using $\bigO{n\log^2(d)}$ classical operations, ordinary quantum gates and QRAM read-out calls. Since $\eta=\Omega(1)$, we get a time-efficient unbiased tomography algorithm using $\bigO{\frac{d}{\eps^2}\ln\left(\frac{2d}{\delta}\right)}$ conditional samples.

\subsection{Improved pure-state tomography using state-preparation oracles}\label{sec:tomoCondition}

If we have a state-preparation oracle available, rather than copies of the state, then 
the precision-dependence can be quadratically improved using iterative refinement \cite{gilyen2023IterativeRefineTomoLinEq}:

\begin{corollary}\label{cor:EffPrepTomo}
Let $\psi\in \mathbb{C}^d$ such that $\nrm{\psi}\leq 1$, and $\eps,\delta\in(0,\frac12]$. Suppose we have access to a controlled unitary $U$ (and its inverse) that prepares the state $U\ket{0^{\otimes a'}}=\ket{{0}^{\otimes a}}\ket{\psi}+\ket{{0^{\otimes a}}^\perp}$, where $a',a=\bigO{\poly\log(d/(\delta \epsilon))}$. 
There is a quantum algorithm that outputs a random vector $\tilde{\psi}\in \mathbb{C}^d$ such that, for every set $V=\{v^{(1)},v^{(2)},\ldots,v^{(k)}\}$ of unit vectors, with probability at least $1-\delta$, $|\braket{\psi-\tilde{\psi}|v}|\leq \epsilon/\sqrt{d}$ for all $v\in V$, using $\mathcal{O}(\frac{d}{\epsilon}\poly\log(kd/(\eps\delta))$ applications of controlled $U$, $U^\dagger$, two-qubit quantum gates, read-outs of a QRAM of size $\mathcal{O}(d\cdot\poly\log(kd/(\eps\delta)))$, and classical computation.

If $k=d$ and $V$ is an orthonormal set, then $\tilde{\psi}$ is $\delta$-close in total variation distance to an ``almost ideal'' discrete random variable $\widecheck{\psi}\in \mathbb{C}^d$ such that $\mathbb{E}[\widecheck{\psi}] = \psi$, $\Pr[\forall v\in V \colon |\braket{\widecheck{\psi}\!-\!\psi|v}|\leq \frac{\eps}{\sqrt{d}} ]=1$, and $\nrm{\mathrm{Cov}(\widecheck{\psi})}\leq {\frac{\eps^2}{d}}$.
\end{corollary}

\begin{proof}
    The idea is to use the tomography algorithm of \cite{gilyen2023IterativeRefineTomoLinEq} to get an estimator $\psi'$ with $\ell_2$-error~$\eps$, with success probability $\geq 1- \frac{\delta}{4}$, using $\bigO{\frac{d}{\epsilon}\log(d/\delta)}$ queries in time $\mathcal{O}(\frac{d}{\epsilon}\log(1/\delta)\cdot\poly\log(d/\epsilon))$. In case of failure we set $\widecheck{\psi}=\psi$.

    We first build the KP-tree for $\psi'$ in QRAM. We can now prepare a state $\ket{0}\ket{\psi'}+\ket{1}\ket{.}$, and thus also the state $\ket{00}\ket{(\psi-\psi')/2}+\ket{1}\ket{..}$, and using linearized amplitude amplification \cite[Theorem 30]{gilyen2018QSingValTransf}, we can also prepare a subnormalized state $\phi$ such that $\nrm{\phi-(\psi-\psi')/(2\eps)}\leq\frac{\delta}{16\sqrt{d}}$ with $\bigO{\log(d/\delta)/\eps}$ (controlled) uses of $U$ and $U^\dagger$.

    As discussed at the start of this subsection, by \Cref{cor:AbsTomoMeas} using $d\ln(\frac{6d}{\delta})$ copies of $\ket{\phi}$ we can output a vector $\vec{\mu}\in[0,1]^d$ such that with probability at least $1-\frac{\delta}{4}$, $|\vec{\mu}_j-|\phi_j||\leq \frac{1}{\sqrt{d}}$ for every $j\in[d]$. (In case of failure we once again set $\widecheck{\psi}=\psi$.) Upon success, the vector $\vec{\mu'}:=\frac{1}{2}\vec{\mu}+\frac{1}{2\sqrt{d}}\mathbf{1}_d$ where $\mathbf{1}_d$ is the $d$-dimensional all-$1$ vector, satisfies $|\vec{\mu}'_j|^2\geq \frac{1}{4}\max\{|\phi_j|^2,\frac{1}{d}\}$ for every $j\in[d]$. Also, by using $\mathcal{\tilde{O}}(d)$ time and QRAM bits, we can construct a KP-tree KP$_{\vec{\mu'}}$ for $\vec{\mu'}$. Thus, by using one query to KP$_{\vec{\mu'}}$ and $\mathcal{\tilde{O}}(1)$ time, we can prepare a state $\ket{\bar{0}}\ket{\vec{\mu'}}+\ket{\bar{0}'^\perp}$, where $\ket{\vec{\mu'}}=\sum\limits_{j\in[d]} \vec{\mu'}_j\ket{j}$. 

    By \Cref{thm:UnbiasedTomoMeas} we can output an unbiased estimator $\tilde{\phi}$ of $\phi$ such that $\Pr[\forall v\in V \colon |\braket{\tilde{\phi}-\phi|v}|\leq \frac{1}{32\sqrt{d}} ]\geq 1 - \frac{\delta}{4}$. Defining $\tilde{\psi}:=\psi'+2\eps\tilde{\phi}$ we then have $\Pr[\forall v\in V \colon |\braket{\tilde{\psi} - \psi|v}|\leq \frac{\eps}{8\sqrt{d}} ]\geq 1 - \frac{\delta}{4}$ since $\nrm{\phi-(\psi-\psi')/(2\eps)}\leq\frac{\delta}{16\sqrt{d}}$. If $V$ is an orthonormal basis, then furthermore $\tilde{\phi}$ is $\delta/4$-close to an ``ideal'' (though not error-free) unbiased estimator $\phi'$ of $\phi$ such that $\Pr[\forall v\in V \colon |\braket{\phi'-\phi|v}|\leq \frac{1}{8\sqrt{d}} ]=1$. 
    Since $\nrm{\phi-(\psi-\psi')/(2\eps)}\leq\frac{\delta}{16\sqrt{d}}$ there is another discrete-valued estimator $\widecheck{\phi}$ within total variation distance $\frac{\delta}{4}$ to $\phi'$ that satisfies $\mathbb{E}[\widecheck{\phi}]=(\psi-\psi')/(2\eps)$ and $\Pr[\nrm{\widecheck{\phi}-\phi'} > \frac{1}{4\sqrt{d}} ]=0$ in turn implying $\Pr[\forall v\in V  \colon |\braket{\widecheck{\phi}-(\psi-\psi')/(2\eps)|v}|\leq \frac{1}{2\sqrt{d}} ]=1$. We then set $\widecheck{\psi}:=\psi'+2\eps\widecheck{\phi}$ (in case no failure happened). We have $\nrm{\mathrm{Cov}(\widecheck{\phi})}\leq 2\nrm{\mathrm{Cov}(\widecheck{\phi}-\phi')}+2\nrm{\mathrm{Cov}(\phi')}$, because for all vectors $a,b$, the matrix $2aa^\dagger+2bb^\dagger-(a+b)(a+b)^\dagger=(a-b)(a-b)^\dagger$ is psd. Therefore the claimed properties of $\tilde{\psi},\widecheck{\psi}$ follow from those of $\tilde{\phi},\phi'$ as guaranteed by \Cref{thm:UnbiasedTomoMeas}, assuming without loss of generality that $\delta\leq \frac{1}{k}$.
\end{proof}

\section{Quantum noisy power method}\label{sec:qnoisypowermethod}

In this section we introduce quantum algorithms for approximating the top eigenvector or top-$q$ eigenvectors. For simplicity, we assume the input matrix $A$ is real and Hermitian, and has operator norm $\norm{A}\leq 1$ (which implies all entries are in $[-1,1]$). Since $A$ is Hermitian, its eigenvalues $\lambda_1,\ldots,\lambda_d$ are real, and we assume them to be ordered in descending order according to their absolute value.\footnote{Sometimes the eigenvalues are ordered $1\geq\lambda_1\geq\cdots\geq\lambda_d\geq-1$ according to their value (instead of their absolute value). To find the $v_1$ associated with $\lambda_1$ in this situation, one can just let $A'=A/3+2I/3$. Then the eigenvectors of $A$ and $A'$ are the same, and the eigenvalues of $A'$ now are all between $1/3$ and $1$ (and hence one can use \Cref{Alg:NPM_top} to find $v_1$). This trick can also be used to find $v_d$ by simply considering $A''=-A/3+2I/3$. Note that the gap between the top and the second eigenvalues of $A'$ might be different from the gap between the top and the second eigenvalues of $A''$.\label{footnote:gaps}} 
Since the entries of $A$ are real, there is always an associated orthonormal basis of \emph{real} eigenvectors $v_1,\ldots,v_d$. For simplicity and without loss of generality, when we mention the $q$th eigenvector of $A$, we mean $v_q$ in this basis. The goal of the algorithms in this section is to find a unit vector $w$ which has large overlap with $v_1$ in the sense that $|\inProd{w}{v_1}|\geq 1-\eps^2/2$. Note that this is equivalent to finding a unit vector $w$ satisfying that either $\|w-v_1\|_2$ or $\|w+v_1\|_2$ is small (at most $\eps$), and hence we say $w$ approximates $v_1$ with small $\ell_2$-error. 

\subsection{Classical noisy power method for approximating the top eigenvector}

For the sake of completeness and pedagogy, we start with the noisy power method of Hardt and Price~\cite{HP15NoisyPowerMethod}, given in \Cref{Alg:NPM_top} below. Like the usual power method, it works by starting with a random vector and applying $A$ some $K$ times to it; the resulting vector will converge to the top eigenvector after relatively small $K$, assuming some gap between the first and second eigenvalues of~$A$. We include a short proof explaining how the noisy power method can approximate the top eigenvector of~$A$ even if there is a small noise vector $G_k$ in the $k$th matrix-vector computation that does not have too much overlap with $v_1$.
\begin{algorithm}[hbt]
\SetKwData{Left}{left}\SetKwData{This}{this}\SetKwData{Up}{up}
\SetKwFunction{Union}{Union}\SetKwFunction{FindCompress}{FindCompress}
\SetKwInOut{Input}{input}\SetKwInOut{Output}{output}

\Input{a Hermitian matrix $A\in[-1,1]^{d \times d}$ with operator norm $\norm{A}\leq 1$;}
Let $w_0$ be a unit vector randomly chosen from $S^{d-1}$;\\
\SetAlgoLined

  \For{$k\leftarrow 0$ \KwTo $K-1$}{
   $y_k= Aw_k+G_k$\;
   $w_{k+1}=y_k/\|y_k\|_2$\;
  }
  \Output{$w_K$;}
 \caption{Noisy power method (NPM) for approximating the top eigenvector of~$A$}
 \label{Alg:NPM_top}
\end{algorithm}

\begin{theorem}\label{thm:npm}
Let $A$ be a $d\times d$ Hermitian matrix with top eigenvector $v_1$, first and second eigenvalues $\lambda_1$ and $\lambda_2$, and  $\gamma=|\lambda_1|-|\lambda_2|$. 
Let $\eps\in(0,0.5)$ and $K=\frac{10|\lambda_1|}{\gamma}\log(20d/\eps)$ (for larger $K$ the theorem still holds).
Suppose $|\langle G_k,v_1\rangle |\leq \gamma/(50\sqrt{d})$ and $\|G_k\|_2\leq \eps\gamma/5$ for all $k\in [K]-1$. Then the unit vector $w_K$ in \Cref{Alg:NPM_top} satisfies $|\inProd{w_K}{v_1}|\geq 1-\eps^2/2$ with probability $\geq 0.9$.
\end{theorem}

\begin{proof}
Let $w_0=\sum_{i\in [d]}\alpha^{(0)}_iv_i$. 
Because $w_0$ is chosen uniformly at random over the unit sphere, by \Cref{cor:innerproductsphere} (without loss of generality assuming $d\geq 4$), with probability $\geq 0.9$, we have $|\alpha^{(0)}_1|\geq 1/(10\sqrt{d})$ and hence we assume $|\alpha^{(0)}_1|\geq 1/(10\sqrt{d})$ below for simplicity. Suppose $w_k= \sum_{i\in [d]}\alpha^{(k)}_iv_i$. We define the tangent angle of $w_k$ as $\tan{\theta_k}=\frac{\sin\theta_k}{\cos\theta_k}= \frac{\sqrt{\sum_{i=2}^d (\alpha_i^{(k)})^2}}{|\alpha_1^{(k)}|}$, and hence $\tan\theta_0\leq 10\sqrt{d}$.  It suffices to show that $\tan\theta_K\leq \eps/2$, because that implies $|\inProd{w_K}{v_1}|=\cos\theta_K\geq 1-\eps^2/4$.

Since $w_k=\sum_{i\in [d]}\alpha^{(k)}_iv_i$, we have $Aw_k=\sum_{i\in [d]}\alpha^{(k)}_i\lambda_iv_i$. Also because $G_k$ satisfies $|\inProd{G_k}{v_1}| \leq \gamma/(50\sqrt{d})$ and $\|G_k\|_2\leq \eps\gamma/50$, we can give an upper bound for $\tan\theta_{k+1}$ as follows:
\begin{align*}
    \tan\theta_{k+1} &\leq \frac{\sqrt{\sum_{i=2}^d (\lambda_i^2\alpha_i^{(k)})^2}+\|G_k\|_2}{|\lambda_1|\cdot|\alpha_1^{(k)}|-|\langle G_k,v_1\rangle|} \leq \frac{|\lambda_2|\sqrt{\sum_{i=2}^d (\alpha_i^{(k)})^2}+\eps\gamma/5}{|\lambda_1|\cdot|\alpha_1^{(k)}|-\gamma/(50\sqrt{d})}\leq \frac{1}{\cos\theta_k}\cdot \frac{|\lambda_1|\sin\theta_k+\eps\gamma/5}{|\lambda_1|-\gamma/5}\\
    &= \frac{1}{\cos\theta_k}\cdot \frac{|\lambda_1|\sin\theta_k+\eps\gamma/5}{|\lambda_2|+4\gamma/5}=  \frac{\sin\theta_k}{\cos\theta_k}\cdot \frac{|\lambda_1|}{|\lambda_2|+4\gamma/5}+\frac{1}{\cos\theta_k}\cdot \frac{\eps\gamma/5}{|\lambda_2|+4\gamma/5}\\
    &\leq \tan\theta_k\cdot \frac{|\lambda_1|}{|\lambda_2|+4\gamma/5}+(1+\tan\theta_k)\cdot \frac{\eps\gamma/5}{|\lambda_2|+4\gamma/5}\\
    &=\big(1-\frac{\gamma/5}{|\lambda_2|+\gamma/5}\big)\frac{|\lambda_2|+\eps\gamma/5}{|\lambda_2|+3\gamma/5}\tan\theta_k+\frac{\gamma/5}{|\lambda_2|+4\gamma/5}\eps\leq \max\{\eps, \frac{|\lambda_2|+\eps\gamma/5}{|\lambda_2|+3\gamma/5}\tan\theta_k\}.
\end{align*}
Note that $\frac{|\lambda_2|+\eps\gamma/5}{|\lambda_2|+3\gamma/5}\leq \max \{\eps, \frac{|\lambda_2|}{|\lambda_2|+2\gamma/5}\}$ because the left-hand side is a weighted mean of the components on the right ($\frac{|\lambda_2|+\eps\gamma/5}{|\lambda_2|+3\gamma/5}=\eps\cdot\frac{\gamma/5}{|\lambda_2|+3\gamma/5}+\frac{|\lambda_2|}{|\lambda_2|+2\gamma/5}\cdot \frac{|\lambda_2|+2\gamma/5}{|\lambda_2|+3\gamma/5}$).
Also, $\frac{|\lambda_2|}{|\lambda_2|+2\gamma/5}\leq (\frac{|\lambda_2|}{|\lambda_2|+5\gamma/5})^{2/5}=(\frac{|\lambda_2|}{|\lambda_1|})^{2/5}$, so we have $\tan\theta_{k+1}\leq \max\{\eps,\tan\theta_k\cdot\max\{\eps,(|\frac{\lambda_2}{\lambda_1}|)^{2/5}\}\}$. By letting $K=\frac{10|\lambda_1|}{\gamma}\log(20d/\eps)$, we obtain $\tan\theta_K\leq \eps/2$, which concludes the proof.
\hfill
\end{proof}

\subsection{Quantum Gaussian phase estimator}\label{sec:QGAE}

Before we explain our quantum noisy power method, we introduce another tool which we call the ``quantum Gaussian phase estimator''. Its aim is to do phase estimation with (approximately) Gaussian error on the estimate. The high-level idea of this estimator is to replace the initial uniform superposition in the algorithm of phase estimation~\cite{kitaev1995QMeasAndAbelianStabilizer,cleve1997QAlgsRevisited} by a discrete Gaussian quantum state, with standard deviation $s$; then the distribution of the error $\tilde{a}-a$ between the amplitude~$a$ and the estimator $\tilde{a}$ produced by the quantum Gaussian amplitude estimator, is also a discrete Gaussian distribution, now with standard deviation $1/s$.\footnote{\rnote{added in response to Review 1}There have been other variations of quantum phase estimation with non-uniform initial state to improve the statistics of the outcome, for instance in the context of Hamiltonian simulation~\cite{Childs2009CDQRW}, but to the best of our knowledge ours is the first application with a discrete \emph{Gaussian} initial state.} Since the latter distribution is sub-Gaussian with parameter $1/s$, with probability at least $1-\delta$ the output is at most $\sqrt{\log(1/\delta)}/s$ away from $a$. Recall that $\propDens_s$ is the pdf for the Gaussian with standard deviation $s$, defined in \Cref{sec:subGaussian}.

\begin{theorem}\label{thm:GAE_a}
   Let $\delta\in (0,0.1]$, $s\geq20\sqrt{2\log(1/\delta)}$, $ a\in [0,1]$, $N=200\cdot \lceil s\sqrt{\log(100/\delta)}\rceil$, $U$ be a unitary, $\ket{\psi}$ be an eigenvector of $U$ such that $U\ket{\psi}=\exp(\pi i a/4)\ket{\psi}$.  There exists a quantum algorithm that for every such $U$ and $a$, given one copy of~$\ket{\psi}$, outputs an estimator $\tilde{a}$ satisfying 
   that $a-\tilde{a}$ distributes $\delta$-close to $\mathcal{D}_{\frac{8}{N}\cdot \mathbb{Z}-\frac{8}{N}\nu,\frac{8}{\sqrt{2}s}}$ for some $\nu \in [0,1)$, using $\mathcal{O}(s\cdot\polylog(s/ \delta))$ applications of controlled-$U$, controlled-$U^{-1}$ and $\mathcal{\tilde{O}}(s\cdot\polylog(s/ \delta))$ time.
\end{theorem}

\begin{proof}
We first explain the algorithm of our quantum Gaussian phase estimator. Let $\ket{\tilde{\propDens_s}}=$ $\frac{1}{\sqrt{\tilde{G}}}\sum\limits_{z\in \{-N/2,\ldots,N/2-1\}}\propDens_s(z)\ket{z}$, where $\tilde{G}$ is a normalizing constant. Let $U_z$ be a unitary that maps $\ket{z}\ket{\psi}\rightarrow \ket{z}U^z\ket{\psi}$ for $z\in\{-N/2,\ldots,N/2-1\}$, and $U_\pi$ be a unitary that maps $\ket{z}\rightarrow (-1)^z\ket{z}$ (this $U_\pi$ is basically a $Z$-gate on the least-significant bit of $z$). 

The algorithm is as follows.
We first prepare the state $\ket{\tilde{\propDens}_s}\ket{w'}$. We then apply $U_z$ to this state, apply $U_\pi$, apply QFT$^{-1}_N$ on the first register, and then measure the first register in the computational basis, divide the outcome value by $N/8$, subtract $4$ from it, and output this value.

We first explain the time complexity of the above algorithm. To prepare $\ket{\tilde{\propDens_s}}$, it suffices to compute all values of $\propDens_s(z)$ for $z\in \{-N/2,\ldots,N/2-1\},$\footnote{Once we have those values, we can do the controlled-rotation tricks similar to how the KP-tree routine produces the quantum state. Since we only need to prepare $\ket{\tilde{\propDens}_s}$ once, it is fine for us to prepare the state using $N\log N$ time. This procedure does not require the use of QRAM. See the discussion above Theorem 2.13 in~\cite{CdW21QLasso}.} and computing all those values ($\propDens_s(-N/2),\ldots,\propDens_s(N/2-1)$) and constructing a KP-tree with those values stored in its leaves takes $\mathcal{O}(N\cdot\poly\log N)$ time.  To construct the unitary $U_z$, it suffices to use $\mathcal{O}(N\poly\log N)$ applications of  controlled-$U$, controlled-$U^{-1}$ and time.  Also, QFT$_N$ can be implemented using $\mathcal{O}(\log^2N)$ elementary gates. As a result, the total cost here is $\mathcal{O}(N\log N\cdot\log (1/\delta)+\log^2N)=\mathcal{O}(s\cdot\polylog(s/\delta))$ time and $\mathcal{O}(N\cdot\poly\log N)=\mathcal{O}(s\cdot\polylog(s/\delta))$ applications of controlled-$U$, controlled-$U^{-1}$.

Now we show the correctness of the above algorithm. 
Applying $U_\pi U_z$ to $\ket{\tilde{\propDens}_s}\ket{w'}$ gives the state $\frac{1}{\sqrt{G'}}\sum\limits_{{z\in \{-N/2,\ldots,N/2-1\}}}\propDens_s(z)(-1)^z\exp(\pi i az/4)\ket{z}\ket{w'}$. If we discard the second register, which is in tensor product with the rest of the state, then the remaining state is also $9\delta$-close to 
\[
\ket{\Psi}=\frac{1}{\sqrt{G}}\sum\limits_{z\in \mathbb{Z}}\propDens_s(z)\exp(2\pi i (a/8+1/2)z)\ket{(z+N/2\text{ mod } N)-N/2}
\] 
because $s\geq 8\sqrt{2\log (1/\delta)}$, $N > 16s\sqrt{2\ln (1/\delta)}$, and by \Cref{thm:truroughlymod}, where $G$ is a normalizing constant. Therefore, the distribution of the outcome of the quantum Gaussian phase estimator is $9\delta$-close to the distribution obtained by measuring the following state (using $a'=a/8+1/2$)
\begin{align*}
\text{QFT}^{-1}_N\ket{\Psi}=&\frac{1}{\sqrt{NG}}\sum_{y\in [N]}\sum\limits_{z\in \mathbb{Z}}\propDens_s(z)\exp(2\pi i a'z)\exp(-2\pi i y ((z+N/2 \text{ mod }N)-N/2)/N)\ket {y}\\
    =& \frac{1}{\sqrt{NG}}\sum_{y\in [N]}\sum\limits_{z\in \mathbb{Z}}\propDens_s(z)\exp(2\pi i z(Na'-y)/N)\ket{y}&\hidewidth{\text{ ($e^{-2\pi iyNz/N}=1$)}}&\\
    =&\frac{1}{\sqrt{NG}}\sum_{y\in [N]}\Big\{\sum\limits_{z\in \frac{1}{N}\cdot\mathbb{Z}}\propDens_{s/N}(z)\exp(2\pi i z(Na'-y))\Big\}\ket {y}&\hidewidth{\text{($z\leftarrow z/N$)}}&\\
        =&\frac{1}{\sqrt{NG}}\sum_{y\in [N]}\Big\{N\cdot\sum\limits_{x\in {N}\cdot\mathbb{Z}}\reallywidehat{\propDens_{{s}/{N}}\exp{(2\pi i z(Na'-y))}}(x)\Big\}\ket {y}&\hidewidth{\text{(by \Cref{thm:Poisson}})}&\\
    =&\frac{1}{\sqrt{NG}}\sum_{y\in [N]}\Big\{\frac{s}{N}\cdot N\cdot\sum\limits_{x\in N\cdot\mathbb{Z}}\propDens_{N/s}(x-Na'+y)\Big\}\ket {y}&\hidewidth{\text{($\widehat{\propDens_{s/N}}=\frac{s}{N}\cdot\propDens_{N/s}$})}&\\
    =& \frac{s}{\sqrt{NG}}\sum_{y\in [N]}\propDens_{N/s}(N\cdot \mathbb{Z}-Na'+y)\ket {y}\\
    =& \frac{s}{\sqrt{NG}}\sum_{y\in \mathbb{Z}}\propDens_{N/s}(-Na'+y)\ket {y\mod N }.
\end{align*}
By \Cref{thm:truroughlymod} again, because $ a' \in [1/2,5/8]$, $N/s\geq 8\sqrt{2\log(1/\delta)}$, and $N \geq16(N/s)\sqrt{2\ln(1/\delta)}$, we know $\text{QFT}^{-1}_N\ket{\Psi}$ is $9\delta$-close to $\frac{1}{\sqrt{G''}}\sum_{y\in [N]}\propDens_{N/s}(-Na'+y)\ket {y}$ where $G''$ is a normalizing constant. Therefore, the probability distribution of $y-Na'$ (letting $y$ be the measurement outcome) is $9\delta$-close to $\mathcal{D}^{[-N/2,N/2-1]}_{\mathbb{Z}-Na',\frac{N}{\sqrt{2}s}}$. 
Moreover, since $\sqrt{\log(12/\delta)/\pi}\leq \frac{N}{\sqrt{2}s}$ and $ 10\frac{N}{\sqrt{2}s}\sqrt{2\ln(1/\delta)}\leq N/2$, by \Cref{cor:DGSvariants} we know $y-Na'$ is also $9\delta+4\delta\exp(\delta)$-close to $\mathcal{D}_{\mathbb{Z}-Na',\frac{N}{\sqrt{2}s}}=\mathcal{D}_{\mathbb{Z}-\nu,\frac{N}{\sqrt{2}s}}$ for some $\nu \in[0,1)$, implying that the distribution of $8y/N-4-a$ is $9\delta+4\delta\exp(\delta)$-close to $\mathcal{D}_{\frac{8}{N}\cdot\mathbb{Z}-\frac{8}{N}\nu,\frac{8}{\sqrt{2}s}}$. As a result, the output of the algorithm in the second paragraph is $9\delta+9\delta+4\delta\exp(\delta)$-close to $\mathcal{D}_{\frac{8}{N}\cdot\mathbb{Z}-\frac{8}{N}\nu,\frac{8}{\sqrt{2}s}}$. Rescaling $\delta$ by a multiplicative constant, we finish the proof.
\hfill
\end{proof}

Using \Cref{thm:subGaussianDGS}, since $\frac{4\sqrt{2}}{s}\geq 8\sqrt{\log(12/\delta)/\pi}/N$ by the choice of $N$, we can see $\mathcal{D}_{\frac{8}{N}\cdot\mathbb{Z}-8\nu,\frac{4\sqrt{2}}{s}}$ is $\delta$-sub-Gaussian with parameter $\frac{4\sqrt{2}}{s}$. 
By letting $s=\frac{4\sqrt{2}}{\eps}$, we have the following corollary.

\begin{corollary}[Sub-Gaussian phase estimator, subGPE($U,\eps,\tau$)]\label{Cor:GAE}
   Let $\eps,\tau\in (0,0.1]$, $a \in [0,1]$, $U$ be a unitary, $\ket{\psi}$ be an eigenvector of $U$ such that $U\ket{\psi}=\exp(\pi i a/4)\ket{\psi}$. There exists a quantum algorithm that, given one copy of $\ket{\psi}$,  outputs an estimator $\tilde{a}$ satisfying 
   that $a-\tilde{a}$ is $\tau$-close to $\tau$-sub-Gaussian with parameter $\eps$ using $\mathcal{\tilde{O}}(\poly\log(1/\tau)/\eps)$ applications of controlled-$U$, controlled-$U^{-1}$ and elementary gates.
\end{corollary}

\subsection{Quantum noisy power method using Gaussian phase estimator}\label{sec:qNPM_G}

In this subsection we combine the noisy power method and the quantum Gaussian phase estimator (introduced in the previous subsection) to get a quantum version of the noisy power method.
It approximates the top eigenvector of a given matrix~$A$ with additive $\ell_2$-error $\eps$ in $\tilde{\mathcal{O}}(d^{1.75}/(\gamma^2\eps))$ time, which is a factor $d^{0.25}$ faster in its $d$-dependence than the best-possible classical algorithm (see \Cref{ssec:clasLBtopeigenvector} for the $\Omega(d^2)$ classical lower bound). 

We first prove the following theorem, which helps us estimate an individual entry of a matrix-vector product $Aw$ (the vector $u$ would be one of the rows of $A$).

\begin{theorem}[Inner product estimator, IPE$(\tau, \delta,\eps)$]\label{thm:IE}
Let $\tau,\delta,\eps\in (0,0.1]$, and $u,w\in B_2^d$ s.t.\ $\|w\|_2=1$. Suppose we can access a KP-tree KP$_{w}$ of $w$ and have quantum query access to entries of~$u$ by a unitary~$O_u$. There is a quantum algorithm that with probability at least $1-\delta$, outputs an estimator $\tilde{\mu}$ satisfying that $\tilde{\mu}-\langle u,w\rangle$ is $\tau$-close to $\tau$-subG$(\eps^2)$, using $\mathcal{\tilde{O}}(d^{0.75}\poly\log(d/\delta)+d^{0.25}\poly\log(1/\tau)/\eps)$ time. 
\end{theorem}

\begin{proof}
    In this proof, we index entries of vectors starting from $0$. Let $I_1=[-d^{-0.25},d^{-0.25}]$ and $I_2=[-1,1]\setminus I_1$. Let $u=u_1+u_2$, where $(u_1)_j=u_j\mathds{1}_{I_1}(u_j)$ and $(u_2)_j=u_j\mathds{1}_{I_2}(u_j)$ for every $j\in[d]-1$.
    Informally, $u_1$ is the vector with smallish entries, and $u_2$ is the vector with largish entries. We separately estimate 
    $\langle u_1,w\rangle$ and $\langle u_2,w\rangle$.

    \paragraph{Finding $u_2$ and computing $\langle u_2,w\rangle$.}
    The number of nonzero entries of $u_2$ is at most $\sqrt{d}$ because $\nrm{u}_2\leq 1$.
We first find (with success probability at least $1-\delta/2$) all the nonzero entries of $u_2$ in $\mathcal{O}{(\sqrt{d\cdot \sqrt{d}}\cdot\poly\log(d/\delta))}=\mathcal{O}{(d^{0.75}\cdot\poly\log(d/\delta))}$ time using \Cref{thm:findallsolutions}. 
Since we can query the entries of $w$ through its KP-tree, we can now compute $\langle u_2,w\rangle$ in time $\mathcal{\tilde{O}}{(\sqrt{d})}$.

\paragraph{Estimating $\langle u_1,w\rangle$.}
Our goal below is first  to show how to prepare (in time $\mathcal{\tilde{O}}(1)$) a superposition corresponding to the vector $d^{-0.25}u_1$, using the fact that all entries of $u_1$ are small; and then to use this to estimate $\langle u_1,w\rangle$ with Gaussian error in time $\mathcal{\tilde{O}}(d^{0.25}/\eps)$.

We can implement $O_{u_1}$ using 2 queries to $O_u$ and $\mathcal{\tilde{O}}(1)$ elementary gates: query $O_u$, and then apply $O_u^{-1}$ conditional on the magnitude of the value being $>d^{-0.25}$  to set the value back to 0 for entries that are in the support of $u_2$ rather than $u_1$.
Let $CR$ be a controlled rotation such that for every $a\in[-d^{-0.25},d^{-0.25}]$
    $$
      CR\ket{a}\ket{0}= \ket{a}(a\cdot d^{0.25}\ket{0}+\sqrt{1-a^2\cdot d^{0.5}}\ket{1}).
    $$
    This can be implemented up to negligibly small error by $\mathcal{\tilde{O}}(1)$ elementary gates.
    Using one application each of $O_{u_1}$, $O^{-1}_{u_1}$, and $CR$, and $\mathcal{\tilde{O}}(1)$ elementary gates, we can map
    \begin{align*}
        \ket{0^{\otimes\log d}}\ket{0}\ket{0}&\xrightarrow{H^{\otimes \log d}\otimes I}\sum\limits_{j\in[d]-1}d^{-0.5}\ket{j}\ket{0}\ket{0}\xrightarrow{O_{u_1}\otimes I}\sum\limits_{j\in[d]-1}d^{-0.5}\ket{j}\ket{(u_1)_j}\ket{0}\\
        &\xrightarrow{I_d\otimes CR} \sum\limits_{j\in[d]-1}(d^{-0.25}(u_1)_j\ket{j}\ket{(u_1)_j}\ket{0}+(\sqrt{d^{-1}-(u_1)_j^2\cdot d^{-0.5}}\ket{j}\ket{(u_1)_j}\ket{1})\\
        &\xrightarrow{O_{u_1}^{-1}\otimes I} \sum\limits_{j\in[d]-1}(d^{-0.25}(u_1)_j\ket{j}\ket{0}\ket{0}+(\sqrt{d^{-1}-(u_1)_j^2\cdot d^{-0.5}}\ket{j}\ket{0}\ket{1}).
    \end{align*}
    Swapping the second register to the front of the state, we showed how to implement the state-preparation unitary $U_{d^{-0.25}u_1}$ that maps
    $$
U_{d^{-0.25}u_1}:\ket{0}\ket{0^{\otimes\log d}}\rightarrow\ket{0}\sum\limits_{j\in[d]-1}(d^{-0.25}(u_1)_j)\ket{j}+\ket{1}\ket{\Phi},
    $$
for some arbitrary unnormalized state $\ket{\Phi}$.
    
    Now we show how to estimate $\langle u_1,w\rangle$. Since we have a KP-tree of $w$, we can implement the state-preparation unitary $U_w$ that maps $\ket{0^{\otimes \log d}}\rightarrow \sum_{j\in[d]-1}w_j\ket{j}$ using $\mathcal{\tilde{O}}(1)$ time by \Cref{thm:StatePrepare}. Let $W=I\otimes U_w\otimes Z$ and $V= U_{d^{-0.25}u_1}\otimes I$. Using \Cref{thm:blockEncodeIP} (with $x$ ranging over 2 cases), we can implement a $(1,\log d+2,0)$ block-encoding of diag($\{d^{-0.25}\langle u_1,w\rangle,-d^{-0.25}\langle u_1,w\rangle\}$) in $\mathcal{\tilde{O}}(1)$ time. In order to be able to use our Gaussian phase estimator, we want to convert $\langle u_1,w\rangle$ into an eigenphase.
    To that end, by \Cref{thm:exp_simluation}, we implement a unitary $U_{exp}$ which is a $(1,\log d+2,0)$-block-encoding of $W=\exp(i\pi(\langle u_1,w\rangle/4) Z)$ using $\mathcal{\tilde{O}}(d^{0.25})$ time. 
Since $\ket{0}$ is an eigenvector of $W$ with eigenvalue $\exp(i\pi \langle u_1,w\rangle/4)$, our Gaussian phase estimator (\Cref{Cor:GAE}) can output (with success probability $\geq 1-\delta/2$) an estimator $\tilde{\eta}$ such that $\tilde{\eta}-\langle u_1,w\rangle$ is $\tau$-close to $\tau$-subG($\eps^2$) using $\mathcal{\tilde{O}}(\poly\log(1/\tau)/\eps)$ applications of controlled-$U_{exp}$, controlled-$U^{-1}_{exp}$, and time. 

Because $\langle u,w\rangle=\langle u_1,w\rangle+\langle u_2,w\rangle$, with probability at least $1-\delta$ we obtain an estimator $\tilde{\mu}=\tilde{\eta}+\langle u_2,w\rangle$ such that $\tilde{\mu}-\langle u,w\rangle=\tilde{\eta}-\langle u_1,w\rangle$ is $\tau$-close to $\tau$-subG($\eps^2$).
\hfill
\end{proof}

Note that the two terms in the time complexity of the above theorem are roughly equal if $\eps$ is a small constant times $1/\sqrt{d}$. Such a small error-per-coordinate translates into a small overall $\ell_2$-error for a $d$-dimensional vector. Accordingly, such a setting of $\eps$ is what we use in our quantum noisy power method for estimating the top eigenvector. 

\begin{algorithm}[hbt]
\SetKwData{Left}{left}\SetKwData{This}{this}\SetKwData{Up}{up}
\SetKwFunction{Union}{Union}\SetKwFunction{FindCompress}{FindCompress}
\SetKwInOut{Input}{input}\SetKwInOut{Output}{output}

\Input{a Hermitian matrix $A\in[-1,1]^{d \times d}$ with operator norm at most~$1$;}
Let $\gamma =|\lambda_1(A)|-|\lambda_2(A)|$; $\eps \in(0,1)$; $K=\frac{10|\lambda_1|}{\gamma}\log(20d/\eps)$; \\
Let $\delta=1/(1000Kd)$; $\tau=\delta/(1000Kd^2)$; $\zeta=\frac{\eps\gamma}{100d^{0.5}\sqrt{\log(1000Kd/\delta)} };$\\
Let $w_0$ be a unit vector randomly chosen from $S^{d-1}$;\\
\SetAlgoLined

  \For{$k\leftarrow 0$ \KwTo $K-1$}{
   Prepare a KP-tree for $w_k$;\\
       For every $j\in[d]$, compute an estimator $(y_k)_j$ of  $\langle A_j,w_k \rangle$ using IPE$(\tau,\delta,\zeta)$ of~\Cref{thm:IE} ($A_j$ is the $j$th row of $A$);\\

   $w_{k+1}=y_k/\|y_k\|_2$\;
  }
  \Output{$w_K$;}
 \caption{Quantum noisy power method using Gaussian phase estimator}
 \label{Alg:qNPM_top}
\end{algorithm}

\begin{theorem}[Quantum noisy power method using Gaussian phase estimator, \Cref{Alg:qNPM_top}]\label{thm:qNPM_top_Gaussian}
    Let $A\in[-1,1]^{d \times d}$ be a symmetric matrix with operator norm at most~$1$, first and second eigenvalues $\lambda_1(A)$ and $\lambda_2(A)$, $\gamma=|\lambda_1(A)|-|\lambda_2(A)|$, $v_1=v_1(A)$ be the top eigenvector of $A$, and $\eps \in(0,1]$. Suppose we have quantum query access to entries of $A$. 
    There exists a quantum algorithm (namely \Cref{Alg:qNPM_top}) that with probability at least $0.89$, outputs a $d$-dimensional vector $w$ such that $|\inProd{w}{v_1}|\geq 1-\eps^2/2$, using $\tilde{\mathcal{O}}(d^{1.75}/(\gamma^2\eps))$ time and $\tilde{\mathcal{O}}(d)$ QRAM bits.
\end{theorem}

\begin{proof}
    Each iteration of \Cref{Alg:qNPM_top} uses $\tilde{\mathcal{O}}(d)$ time and QRAM bits to build the KP-tree for $w_k$, and for every $j\in[d]$, we use $\mathcal{\tilde{O}}(d^{0.75}\poly\log(d/\delta)+d^{0.25}\poly\log(1/\tau)/\zeta)$ time and $\mathcal{\tilde{O}}(\sqrt{d})$ QRAM bits for estimating $\langle A_j,w_\ell\rangle$ by IPE($\tau, \delta,\zeta$) in \Cref{thm:IE}. Hence the total number of elementary gates we used and queries to entries of $A$ is $\tilde{\mathcal{O}}(d\cdot(d^{0.75}+d^{0.25}/\zeta)\cdot K)=\tilde{\mathcal{O}}(d^{1.75}/(\gamma^2\eps))$.

    Now we are ready to show the correctness of \Cref{Alg:qNPM_top}. By \Cref{thm:npm} and the union bound, it suffices to show that for each $k\in[K]-1$, both $\|y_k-Aw_k\|_2\leq \gamma\eps/5$ and $|\langle y_k-Aw_k, v_1\rangle|\leq \gamma/(50\sqrt{d})$ hold with probability $\geq 1-1/(100K)$. Fix $k$. By \Cref{thm:IE}, we know that for every $j\in[d]$, with probability at least $1-\delta$, $(y_k-Aw_k)_j$ is $\tau$-close to $\tau$-subG($\zeta^2$) for every $j\in[d]$. Let $e=y_k-Aw_k$. 
    There are three different kinds of bad events, whose probabilities we now analyze.
    Firstly, we can see that for every $j\in[d]$, with probability at least $1-\delta$, 
    \begin{align*}
        \Pr\left[|e_j|>\frac{\gamma\eps}{5\sqrt{d}}\right]\leq 2\exp(\tau)\cdot \exp\left(-\frac{(\frac{\gamma\eps}{5\sqrt{d}})^2}{2\zeta^2}\right)+\tau=2\exp(\tau-200\log(1000Kd/\delta)) +\tau  
        \leq \frac{\delta}{100Kd}.
    \end{align*}
    Therefore, with probability at least $1-d\delta-\delta/(100K)$, $|e_j|\leq \frac{\gamma\eps}{5\sqrt{d}}$ for every $j\in[d]$, implying that $\|e\|_2\leq \frac{\gamma\eps}{5}$. Secondly, by the properties of sub-Gaussians from \Cref{sec:subGaussian}, we know that with probability at least $1-d\delta$, $\langle e, v_1\rangle$ is $d\tau$-close to $d\tau$-subG($\sum\limits_{j\in[d]} (v_1)_j^2 \zeta^2$). Thirdly, since $v_1$ is a unit vector, we have that (with probability at least $1-d\delta$) $\langle e, v_1\rangle$ is $d\tau$-close to $d\tau$-subG($\zeta^2$) and hence 
    \begin{align*}
        \Pr\left[|\langle e, v_1\rangle|>\frac{\gamma}{50\sqrt{d}}\right] \leq & ~2\exp(d\tau)\cdot \exp\left(-\frac{(\frac{\gamma}{50\sqrt{d}})^2}{2\zeta^2}\right)+d\tau\\
        = & ~2\exp(d\tau)\cdot \exp\left(-2\frac{\log(1000Kd/\delta)}{\eps^2}\right)+d\tau\leq \frac{\delta}{100Kd}.
    \end{align*}
    As a result, by the union bound over the three kinds of error probabilities, for each $k\in[K]\cup \{0\}$, with probability $1-2d\delta-\delta/(50K)\geq 1-1/(100K)$,  we have both $\|y_k-Aw_k\|_2\leq \gamma\eps/5$ and $|\langle y_k-Aw_k, v_1\rangle|\leq \gamma/(50\sqrt{d})$. This proves correctness of \Cref{Alg:qNPM_top}. 
    \hfill
 \end{proof}

\subsection{Almost optimal process-tomography of ``low-rank'' reflections}\label{ssec:tomolowrankreflections}

In this subsection we describe an essentially optimal algorithm for the ``tomography'' of projectors $\Pi$ of rank at most $q$ (or of the corresponding unitary reflection $2\Pi-I$).\footnote{Having access to a controlled reflection $2\Pi-I$ is equivalent up to constant factors to having access to controlled $U_\Pi^{\pm 1}$ (i.e., controlled-$U_\Pi$ and its inverse) for a block-encoding $U_\Pi$ of the projector $\Pi$, as follows from the QSVT framework~\cite{gilyen2018QSingValTransf}.} 
We will use this in the next subsection to approximate the eigensubspace spanned by the top-$q$ eigenvectors. For generality, we will from here on allow our matrices to have complex entries, not just real entries like in the earlier subsections.

Our algorithm is inspired by the noisy power method and has query complexity $\bigOt{dq/\eps}$ and time complexity $\bigOt{dq/\eps+dq^2}$ (using QRAM). When $q\ll d$ this gives a better complexity than the optimal unitary process-tomography algorithm of Haah, Kothari, O’Donnell, and Tang~\cite{HKOT23QueryOptUnitaryProcTomo}. Also in the special case when $q=1$ this gives a qualitative improvement over prior pure-state tomography algorithms~\cite{kerenidis2018QIntPoint,apeldoorn2022QTomographyWStatePrepUnis} which required a state-preparation unitary, while for us it suffices to have a reflection about the state, which is a strictly weaker input model.\footnote{Indeed, we can implement a reflection about a state $\ket{\psi}$ by a state-preparation unitary and its inverse as in amplitude amplification. However, if we only have access to a reflection about an unknown classical basis state $\ket{i}$ for $i\in[d]$, then we need to use this reflection $\Omega(\sqrt{d})$ times to find $i$ (because of the optimality of Grover search) showing that the reflection input is substantially weaker than the state-preparation-unitary input.} Surprisingly, it turns out that this weaker input essentially does not affect the query and time complexity. 

Observe that if we have two projectors $\Pi,\Pi'$ of rank $r$, then
\begin{align}\label{eq:TrToOp}
\nrm{\Pi-\Pi'}_1/(2r)\leq \nrm{\Pi-\Pi'} \leq \nrm{\Pi-\Pi'}_1/2.
\end{align}
This implies that the $\eps$-precise estimation of a rank-$1$ projector $\Pi$ (i.e., the density matrix corresponding to a pure state) is equivalent to $\eps$-precise approximation of the quantum state $\Pi$. Since the complexity of pure-state tomography using state-preparation unitaries is known to be $\widetilde{\Theta}(d/\eps)$~\cite{apeldoorn2022QTomographyWStatePrepUnis}, it follows that our algorithm is optimal up to log factors in the $q=1$ case.

Similarly, the query complexity of our algorithm is optimal for $\eps=\frac16$ up to log factors, as can be seen by an information-theoretic argument using $\eps$-nets, since by \Cref{lem:epsSupspaceNet} there is an ensemble of $\exp(\Omega(dq))$ rank-$q'$ projectors for $q'=\min(q,\lfloor\frac{d}{2}\rfloor)$ such that each pair of distinct projectors is more than $\frac13$ apart in operator norm.
We conjecture that the query complexity of the task is actually $\widetilde{\Omega}(dq/\eps)$, meaning that our algorithm has essentially optimal query complexity for all $\eps \in(0,\frac16]$.

Now we briefly explain our algorithm for finding an orthonormal basis of the image of $\Pi$, assuming for ease of exposition that its rank is exactly~$q$. 
We want to find a $d\times q$ isometry $W$ such that $\|WW^\dagger-\Pi\|$ is small. 
Our algorithm (\Cref{alg:qNPMopt} below) can be seen as a variant of the noisy power method: We start with generating $m$ Gaussian vectors $g_1,\ldots,g_{m}$ ($m=\widetilde{\Theta}(q)$ will be slightly bigger than $q$), with i.i.d.\ (complex) normal entries each having standard deviation $\sim 1/\sqrt{d}$. 
Subsequently, we repeat the following process $K \approx \log\sqrt{d/m}$ times: we estimate $\Pi g_i$ for every $i\in[m]$ using our quantum state tomography algorithm (\Cref{thm:UnbiasedTomoMeas} and \Cref{cor:EffPrepTomo}), and multiply the outcome by $2$, resulting in $m$ new vectors (these factors of~2 allows our analysis to treat all the errors together in one geometric series later). 
Finally, let $V$ be the $d\times m$ matrix whose columns are the versions of those $m$ vectors after the last iteration. The algorithm classically computes the singular value decomposition (SVD) of this $V$, and outputs those left-singular vectors $w_1,w_2,\ldots, w_q$ that have singular value greater than the threshold of $\frac{1}{14}$.

Note that the span of $g_1,\ldots,g_m$ includes the $q$-dimensional image of $\Pi$ almost surely.
Hence if no error occurs in the tomography step, then $V=2^K\Pi [g_1,\ldots,g_m]$, and the image of $V$ is the image of $\Pi$ (almost surely). Thus, if $\sum_{i=1}^q \varsigma_i w_i u_i^\dagger$ is the SVD of $V$, then $\sum_{i\in[q]}w_iw_i^\dagger=\Pi$. That is, we get the desired output $W$ by rounding the singular values of $V$ appropriately: the first $q$ singular values are rounded to~1, and the others are rounded to~0. To show the stability of this approach it remains to show that the singular values $\varsigma_1,\ldots,\varsigma_q$ of the final $V$ in the non-error-free case are still large ($\approx 1$), the other $d-q$ singular values are still essentially~0, and the error incurred by the quantum state tomography is small.
Our analysis relies on the concentration of the singular values of sufficiently random matrices, see \Cref{sec:MatConc}. 

To bound the effect of errors induced by tomography we combine the operator norm bound on random matrices of \Cref{thm:heavyRows} with our unbiased tomography algorithm (\Cref{thm:UnbiasedTomoMeas}). The key observation is that \Cref{cor:EffPrepTomo} gives an estimator $\tilde{\psi}$ that is $\delta$-close in total variation distance to an ``ideal'' (though not error-free) estimator $\widecheck{\psi}$ that satisfies $\mathbb{E}[\widecheck{\psi}]=\psi$, whose covariance matrix has operator norm at most  $S^2\leq\frac{\eps^2}{d}$, and $\nrm{\Pi\widecheck{\psi}-\Pi\psi}_2\leq\eps\sqrt{\frac{q}{d}}$ with certainty. 
For the sake of analysis we can assume that we work with $\widecheck{\psi}$, because we only notice the difference between $\widecheck{\psi}$ and $\tilde{\psi}$ with probability at most $\delta$, which can be made negligibly small at a logarithmic cost in \Cref{cor:EffPrepTomo}. Finally, we use perturbation bounds on singular values and vectors from \Cref{sec:SVPert}.

\begin{algorithm}[hbt]
    \DontPrintSemicolon
    \SetAlgoLined
    \SetKwInOut{Input}{Input}
    \SetKwInOut{Output}{Output}
    \SetKwInput{kwInit}{Init}    
    \SetKwFor{For}{for}{do}{endfor}
    \SetKwIF{If}{ElseIf}{Else}{if}{then}{elseif}{else}{endif}
    \SetKwRepeat{Do}{do}{while}
    \SetKw{kwABORT}{ABORT}

    \Input{Dimension $d$, maximum failure probability $\delta'\in (0,\frac12]$, target precision $\eps\in (0,\frac12]$, $\bigOt{\frac{\eps\delta'}{dq}}$-approximate block-encoding $U_\Pi$ of a  projector $\Pi$ of rank $\leq q$}

    \Output{$\eps$-approximate orthonormal basis $W$ of the image of $\Pi$, i.e., $\norm{WW^\dagger-\Pi}\leq \eps$}
    \kwInit{$m\!:=\!\min\!\Big(\Big\lceil\!\max\!\Big(16C^2q,8c\ln\left(\frac{10}{\delta'}\right)\Big)\Big\rceil,d\Big)$, $K\!:=\!\Big\lceil\log_2\!\Big(\!\sqrt{\!\frac{d}{m}}+1\Big)\Big\rceil $, $\eps'\!:= \frac{\eps}{65+98\sqrt{\ln\!\big( \frac{10d}{\delta'}\big)/c'}}$\;
    \tcp*[h]{the constants $c,C$ come from \Cref{cor:wellConditioned} and in the real case$^\textsuperscript{\ref{foot:realVsComplexBound}}$ $c,C=1$}\;
    \tcp*[h]{the constant $c'$ comes from \Cref{thm:heavyRows}}\;}
    
    \BlankLine
    \textbf{1.)}
    \lIf{$m=d$}{set $ g_j = \ket{j}/4$,
    \lElse{ generate $m$ random vectors $ g_j $ with i.i.d.\ complex (or real if $\Pi\in\mathbb{R}^{d\times d}$) standard normal entries multiplied by $\eta:=\frac{4\cdot 2^{-K}}{7\sqrt{m}}$ \tcp*[h]{note $\eta < \frac{4}{7\sqrt{d}}$}
    }}
	\vskip-4.5mm
    \textbf{2.)} \For{ $j = 1$ \KwTo $m$}{
      Set $g_j^{(0)} = g_j$; \lIf{$\nrm{g_j} > 2$}{\kwABORT}
      \For{$ k = 0 $ \KwTo $ K-1 $}{
        Classically compute $\nrm{g_j^{(k)}}$ and store $g_j^{(k)}/\nrm{g_j^{(k)}}$ in a KP-tree\\ (we can now unitarily prepare $\ket{\psi}=\Pi g_j^{(k)}\!/\nrm{g_j^{(k)}}$ using this KP-tree and $U_\Pi$)\;
        Obtain $ y_j^{(k+1)} $ via $\eps'/\nrm{g_j^{(k)}}$-precise         tomography (\Cref{cor:EffPrepTomo}) on $\ket{\psi}$ setting $ \delta\leftarrow \delta'/(5mK) $ \tcp*[h]{query complexity is $\bigOt{ d / \eps}$}\;
        Set $ g_j^{(k+1)} \leftarrow 2\nrm{g_j^{(k)}}y_j^{(k+1)}$;  \lIf{$\nrm{g_j^{(k+1)}} > 2$}{\kwABORT}
      }
    }
    \textbf{3.)} Output the left-singular vectors of $V\!=[g_1^{(K)}\!, \ldots, g_{m}^{(K)}]$ with singular value above $\frac{1}{14}$.\kern-1mm\;
    \tcp*[h]{Classical complexity is $\bigOt{dq^2}$ via direct diagonalization of $V^\dagger V$ }\;
     \tcp*[h]{The output is correct with probability at least $1-\delta'$}\;
    \tcp*[h]{The total $ U_\Pi$-query and quantum gate complexity is $\bigOt{dq/\eps} $}\;
     \caption{Time-efficient approximation of the top-$q$ eigensubspace}\label{alg:qNPMopt}
\end{algorithm}

We say that $U_\Pi$ is an $\eps$-approximate block-encoding of $\Pi$ if $\nrm{U_\Pi-U}\leq \eps$ for some $U$ satisfying $\Pi=(\bra{0^a}\otimes I)U(\ket{0^a}\otimes I)$.\footnote{The condition $\nrm{(\bra{0^a}\otimes I)U_\Pi(\ket{0^a}\otimes I)-\Pi}\leq \eps'$ appears similar, however is in some sense quadratically weaker. Consider, e.g., $\Pi=1$, $a=1$, and $U_\Pi=\left(\begin{array}{cc}
    \cos(x) & -\sin(x) \\
    \sin(x) &  \cos(x) 
\end{array}\right)$, then $1-\cos(x) = x^2/2 + \mathcal{O}(x^4)$, but for any unitary $U=\ketbra{0}{0}+z\ketbra{1}{1}$ we have $\nrm{U_\Pi-U}\geq |\sin(x)| = |x| +  \mathcal{O}(|x|^3)$. Nevertheless, because $\Pi$ is a projector, we can remedy this in general by converting $U_\Pi$ to an (approximate) block-encoding of $\Pi/2$ via linear combination of unitaries, and then applying quantum singular value transformation with the polynomial $-T_3(x)= 3 x-4 x^3$; the resulting unitary is then indeed $\mathcal{O}(\eps')$-close to a perfect block-encoding of $\Pi$, see for example the proof of \cite[Lemma 23]{gilyen2018QSingValTransf}.}

\begin{theorem}[Correctness of \Cref{alg:qNPMopt}]\label{thm:LearnProj}
    Let $\Pi\in\mathbb{C}^{d\times d}$ be an orthonormal projector of rank at most $q$, given via an $\bigOt{\frac{\eps\delta'}{dq}}$-approximate block-encoding $U_\Pi$. \Cref{alg:qNPMopt} outputs an isometry $W$ such that, with probability at least $1-\delta'$, $\|WW^\dagger-\Pi\|\leq \eps$, using $\bigOt{\frac{dq}{\eps}}$ controlled $U_\Pi$, $U_\Pi^\dagger$, two-qubit quantum gates, read-outs of a QRAM of size $\mathcal{\tilde{O}}(d)$, and $\bigOt{dq^2}$ classical computation. 
\end{theorem}
%\anote{The query complexity is in fact $\bigO{\frac{dq}{\eps}\polylog(d/\delta)}$, this a very minor logarithmic tightening of the statement.}

\begin{proof}
First consider the case when there is no error in tomography and in the implementation of~$\Pi$. 
Then we end up with $g_j^{(K)} = 2^K \Pi g_j^{(0)} $, and the corresponding matrix $V_{ideal}= 2^K[\Pi g_1^{(0)},\ldots,\Pi g_m^{(0)}]$ has almost surely $\mathrm{rank}(\Pi)$ nonzero singular values with associated left-singular vectors lying in the image of $\Pi$. The $m=d$ case is trivial. If $m<d$, then due to \Cref{cor:wellConditioned} all corresponding singular values are in $(\frac{1}{7},1)$ with probability at least $1-\frac{\delta'}{5}$.\footnote{The matrix $\frac{1}{\eta}[g_1^{(0)},\ldots,g_m^{(0)}]$ is a $d\times m$ random matrix with i.i.d.\ (complex) standard normal entries. After multiplying by $\Pi$ this effectively (up to a rotation) becomes a $q\times m$ random matrix with i.i.d.\ (complex) standard normal entries. 
We apply \Cref{cor:wellConditioned} (with $N=m$) to the latter matrix, obtaining the interval $[\frac{1}{4}\sqrt{m},\frac{7}{4}\sqrt{m}]$ for its singular values. Multiplying by $\eta 2^K=4/(7\sqrt{m})$ we get the interval $[\frac17,1]$ for the singular values of $V_{ideal}$.}
By \Cref{prop:rndVecLength} (and a union bound over all $j\in[m]$) we know that with probability at least $1-\frac{\delta'}{10}$ we have for all $j\in[m]$ that\footnote{We have $\eta(\sqrt{d} + \sqrt{2\ln(10m/\delta')})\leq \eta\sqrt{d} + \frac{2}{7\sqrt{m}}\sqrt{2\ln(10m/\delta')} < \frac{4}{7} + \frac{2}{7}\sqrt{\frac{2\ln(10m/\delta')}{m}}=\frac{4}{7} + \frac{2}{7}\sqrt{\frac{2\ln(m)+2\ln(10/\delta')}{m}}< \frac{6}{7}$, because $m\geq 8\ln(10/\delta')$ and $\frac{2\ln(x)}{x}$ takes its maximum at $x=e$, where it is less than $\frac{3}{4}$.} $\nrm{g_j} < \eta(\sqrt{d} + \sqrt{2\ln(10m/\delta')})<\frac{6}{7}$.
Thereby with probability at least $1-\frac{3\delta'}{10}$ \Cref{alg:qNPMopt} does not abort at the intialization of $g_j^{(0)}\!$, and the left-singular vectors of $V_{ideal}$ with singular value at least $1/7$ form the columns of the desired matrix $W$ such that $WW^\dagger=\Pi$; in the remainder of the proof we assume this is the case.

Second, we consider what happens when the tomography has error, but we can implement $\Pi$ exactly.
Let $\tilde{e}^{(k,\ell)}_j:=g_j^{(k)}-2^{k-\ell}\Pi g_j^{(\ell)}$ be the aggregate tomography error that occurred from the $\ell$-th iteration to the $k$-th iteration where $k\in[K]$, $\ell\in[k]-1$, and observe that $\tilde{e}^{(k,\ell)}_j= \tilde{e}^{(k,k-1)}_j + \sum_{i=\ell+1}^{k-1}2^{k-i}\cdot \Pi \tilde{e}^{(i,i-1)}_j$.
For each iteration, we do the tomography with precision $\eps'/\nrm{g_j^{(k)}}$ on $\Pi g_j^{(k)}/\nrm{g_j^{(k)}}$ via \Cref{cor:EffPrepTomo}, guaranteeing that the random variable $\tilde{e}^{(k,k-1)}_j$ is $\delta$-close in total variation distance to an ``ideal'' random variable $e^{(k,k-1)}_j$ such that $\nrm{e^{(k,k-1)}_j}\leq \eps'$ and $\nrm{\Pi e^{(k,k-1)}_j}\leq \eps'\sqrt{q/d}$ almost surely, and $\mathbb{E}[e^{(k,k-1)}_j]=0$, $\nrm{\mathrm{Cov}(e^{(k,k-1)}_j)}\leq \frac{\eps'}{d^2}$ (to see this, choose $V$ to be an orthonormal basis whose first $q$ elements span the image of $\Pi$). We define analogously $e^{(k,\ell)}_j := e^{(k,k-1)}_j + \sum_{i=\ell+1}^{k-1}2^{k-i}\cdot \Pi e^{(i,i-1)}_j$.
Note that the distribution of $\tilde{e}^{(k,k-1)}_j$ (and $e^{(k,k-1)}_j$) can depend on $\tilde{e}^{(i,i-1)}_j$ (and $e^{(i,i-1)}_{j'}$, respectively) only if $j=j'$ and $i\leq k$.
Let $\vec{\tilde{e}}_j:=(\tilde{e}^{(1,0)}_j,\tilde{e}^{(2,1)}_j,\ldots,\tilde{e}^{(K,K-1)}_j)$, and define $\vec{e}_j$ analogously.
We can apply \Cref{lem:condTVBound} recursively to show that $d_{TV}(\vec{\tilde{e}}_j,\vec{e}_j)\leq {K}{\delta}$.
Since the $\vec{\tilde{e}}_j$ are independent from each other, we can assume without loss of generality that so are the $\vec{e}_j$. 
Once again by \Cref{lem:condTVBound} we get that when comparing the two sequences of $m$ random variables, we have $d_{TV}((\vec{\tilde{e}}_j\colon j \in m),(\vec{e}_j\colon j \in m))\leq {mK}{\delta}=\frac{\delta'}{5}$. From now on we replace $(\vec{\tilde{e}}_j\colon j \in m)$ by $(\vec{e}_j\colon j \in m)$ throughout the analysis, which can therefore hide an additional failure probability of at most $\frac{\delta'}{5}$.

By the triangle inequality we have that 
\begin{align*}
\nrm{e^{(k,0)}_j} & \leq \nrm{e^{(k,k-1)}_j}+\sum_{i=1}^{k-1}2^{k-i}\nrm{\Pi e^{(i,i-1)}_j} \leq \eps'+\sum_{i=1}^{k-1}2^{k-i} \eps'\sqrt{q/d}\\ &
\leq (1+2^{k}\sqrt{q/d}) \eps'< (1+2(\sqrt{d/m}+1)\sqrt{q/d}) \eps'\leq 5\eps'
\end{align*}
for every $k\in[K]$. 
This also implies that for every $k\in[K]$ we have $\nrm{g_j^{(k)}}\leq \nrm{2^k\Pi g_j^{(0)}} + \nrm{e^{(k,0)}_j}\leq \nrm{2^K\Pi g_j^{(0)}} + 5\eps'\leq 1+ 5\eps'\leq 2$, and therefore \Cref{alg:qNPMopt} also does not abort in the for-loop.

Let us define $E_{tomo}:=[e^{(K,0)}_1,\ldots,e^{(K,0)}_m]=V-V_{ideal}$ as the matrix of accumulated tomography errors.
We can apply \Cref{lem:condCovBound} recursively with $X=\Pi e^{(k-1,0)}_j, Y=(I-\Pi) e^{(k-1,0)}_j, Z= e^{(k,k-1)}_j$ to show that for all $k\in [K]$ we have
\begin{align*}
    \nrm{\mathrm{Cov}(e^{(k,0)}_j)}&\leq \nrm{\mathrm{Cov}(e^{(k,k-1)}_j)}+\sum_{i=1}^{k-1}4^{k-i} \nrm{\mathrm{Cov}(\Pi e^{(i,i-1)}_j)}\\&
    \leq \nrm{\mathrm{Cov}(e^{(k,k-1)}_j)}+\sum_{i=1}^{k-1}4^{k-i} \nrm{\mathrm{Cov}(e^{(i,i-1)}_j)}\\&
    \leq \sum_{i=1}^{k}4^{k-i} \frac{\eps'^2}{d} 
    \leq \frac{4^K\eps'^2}{3d}
    \leq \frac{4(\sqrt{d/m}+1)^2\eps'^2}{3d}
    \leq \frac{16\eps'^2}{3m}.
\end{align*}
Applying \Cref{thm:heavyRows} for $t=\sqrt{\ln(10d/\delta')/c'}$ gives that with probability at least $1-\frac{\delta'}{5}$ we have
	\begin{equation}
		\nrm{ E_{tomo} } \le \frac{8\eps'}{\sqrt{3}} + 7 \eps' \sqrt{\ln(10d/\delta')/c'} \leq \frac{\eps}{14} .
	\end{equation}
Since $\nrm{ E_{tomo} }\leq \frac{\eps}{14}$, using the notation of \Cref{thm:Wedin} we have that $\Pi=\Pi_{>0}^{V_{ideal}}$ and the rank of $\Pi_{>\frac{1}{14}}^{V}$ is rank$(\Pi)\leq q$ due to Weyl's bound (\Cref{thm:Weyl}). Therefore, by \Cref{thm:Wedin} and \Cref{lem:ProjDiff} we have $\nrm{\Pi - \Pi_{>\frac{1}{14}}^{V}}\leq 14 \nrm{V-V_{ideal}}=14\nrm{ E_{tomo} } \leq \eps$ as desired. 

Finally, let us analyze the effect of implementation errors in $U_\Pi$. We perform tomography $m K$ times via \Cref{cor:EffPrepTomo}, each time using $T=\mathcal{O}(\frac{d}{\eps}\poly\log(d/(\eps\delta))$ applications of $U^{\pm 1}_\Pi$, therefore the induced total variation distance\footnote{With more careful tracking of error spreading in the estimated vectors it might be possible to show that it suffices to have access to a block-encoding satisfying the weaker condition $\nrm{(\bra{0^a}\otimes I)U_\Pi(\ket{0^a}\otimes I)-\Pi} \leq\frac{\eps'}{\sqrt{d m}}$.} in the output distribution is at most $TmK\cdot \bigOt{\frac{\eps\delta'}{dq}}\leq \frac{\delta'}{10}$. Preparing a KP-tree that allows a similar precision for the preparation of $g_j^{(k)}\!/\nrm{g_j^{(k)}}$ likewise induces at most an additional $\frac{\delta'}{10}$ total variation distance, implying that our algorithm outputs a sufficiently precise answer with probability at least $1-\delta'$ when all approximations are considered.
The quantum gate complexity comes entirely from \Cref{cor:EffPrepTomo}, which is gate efficient, while the final computation requires computing the SVD of a $d\times m$ matrix, which can be performed in $\bigOt{dm^2}=\bigOt{dq^2}$ classical time.
\hfill
\end{proof}

Using basic quantum information theory, one can see that recovering the $q$-dimensional subspace with small constant error~$\eps=1/6$ gains us $\Omega(dq)$ bits of information about the subspace, and hence requires $\widetilde{\Omega}(dq)$ quantum queries. This shows that the query complexity of our previous algorithm is essentially optimal in its $dq$-dependence. 
We first show that there exists a large set of $q$-dimensional projectors that are all far apart from each other. This might be a standard result, but we did not find a reference in the literature so provide our own proof.
 
\begin{lemma}\label{lem:epsSupspaceNet}
	Let $q\leq d/2$. There exists a set $S$ of $q$-dimensional subspaces of $\mathbb{R}^d$ of size $\exp(\Omega(dq))$ such that for any distinct $s,r\in S$ we have $\nrm{\Pi_s-\Pi_r}>\frac13$, where $\Pi_t$ denotes the projector to the subspace $t$.
\end{lemma}

\begin{proof}
We can assume without loss of generality that $d\geq 128^3$. 

	First let us assume that $q\leq d/64$; we show the existence of such a set $S$ via the probabilistic method, by showing that for any set $S$ of subspaces, if $|S|< \exp(qd/32-1)$, then with non-zero probability a Haar-random  $q$-dimensional subspace $r$ satisfies $\nrm{\Pi_s-\Pi_r}>\frac13$ for every $s\in S$ (thereby we can take $S\leftarrow S\cup\{r\}$). We sample $r$ as follows: generate a random matrix $R\in\mathbb{R}^{d\times q}$ with i.i.d.\ standard normal entries, and accept $R$ only if $\varsigma_{\min}(R)\geq \frac34\sqrt{d}$ (i.e., take a sample conditioned on this happening -- we know by  \Cref{thm:wellConditioned} that this happens with probability $\geq \frac{1}{e}$). 

	Upon acceptance we compute a singular value decomposition $R=U \Sigma V^\dagger$ and define $\Pi = U U^\dagger$ as the projector corresponding to the subspace. 	
	Since $\Pi = U U^\dagger$ is an orthogonal projection to the image of $R=U \Sigma V^\dagger$ we have $\Pi r_i=r_i$ for all columns $r_i$ of $R$ and thus
	\begin{align*}
		\nrm{\Pi_s r_i}&\geq \nrm{\Pi r_i}-\nrm{(\Pi-\Pi_s) r_i}\\&
		=\nrm{r_i}-\nrm{(\Pi-\Pi_s) r_i}\\&
		\geq (1-\nrm{\Pi-\Pi_s})\nrm{r_i}\\&		
		\geq (1-\nrm{\Pi-\Pi_s})\varsigma_{\min}(R)\\&
		\geq \frac{3}{4}\sqrt{d}(1-\nrm{\Pi-\Pi_s}).
	\end{align*}
	Hence $\nrm{\Pi-\Pi_s}\leq\frac13$ is only possible if $\nrm{\Pi_s r_i} \geq \sqrt{d}/2$ for all columns of $R$. Without the conditioning (on $\varsigma_{\min}(R)\geq \frac34\sqrt{d}$, the condition we accept $R$), $\Pi_s r_i$ is effectively a random $q$-dimensional vector with i.i.d.\ standard normal entries. Since $\sqrt{d}/2-\sqrt{q}\geq \sqrt{d}/4$, by \Cref{prop:rndVecLength} for any $s\in S$ we have $\Pr[\nrm{\Pi_s r_i} \geq \sqrt{d}/2]\leq\exp(-d/32)$, and due to the independece of the columns the probability that this happens for all $i\in [q]$ is less than $\exp(1-qd/32)$ even after conditioning. Taking the union bound over all $s\in S$ we can conclude that with non-zero probability $\nrm{\Pi-\Pi_s}>\frac13$ for all $s\in S$.

 	The statement for $q=\Omega(d)$ follows from \cite[Lemma 8]{haah2017OptTomography}. Alternatively, if $q> d/64$, we can set $q':=\lceil q/128\rceil$, $d':=d-(q-q')$ so that $q'\leq d'/64$. Then from a large set $S'$ of $q'$-dimensional subspaces of $\mathbb{R}^{d'}$ satisfying $\nrm{\Pi_{s'}-\Pi_{r'}}>\frac13$ for any distinct $s',r'\in S'$ we construct $S:=\{s \colon \Pi_s = \Pi_{s'}\oplus I_{q-q'} \text{ for some } s'\in S'\}$ so that also $\nrm{\Pi_{s}-\Pi_{r}}>\frac13$ for distinct $s,r\in S$. 
  \hfill
\end{proof}

\subsection{Time-efficiently approximating the subspace spanned by top-\texorpdfstring{$q$}{q} eigenvectors}\label{sec:qspaceapproximate}

In this subsection, we give a quantum algorithm to ``approximate the top-$q$ eigenvectors'' in a strong sense using $qd^{1.5+o(1)}$ time (and $q\sqrt{s}d^{1+o(1)}$ time if the matrix is $s$-sparse). 
In particular, when $q=1$, this algorithm outputs a vector that approximates the top eigenvector using $d^{1.5+o(1)}$ time. This is what we referred to as our ``second algorithm'' in \Cref{ssec:resultsalgos}. 

Consider the following situation: for $q\in [d]$, suppose we only know there is a significant eigenvalue gap between the $q$th eigenvalue $\lambda_q$ and the $(q+1)$th eigenvalue $\lambda_{q+1}$ (it doesn't really matter whether we order the eigenvalues by value or by absolute value, we can make the algorithm work in each of these two cases, see Footnote~\ref{footnote:gaps}). Is there a way we can learn the subspace spanned by the top-$q$ eigenvectors? Here we consider the subspace instead of the top-$q$ eigenvectors directly, because there might be degeneracy among $\lambda_1,\ldots,\lambda_q$, in which case the set of the top-$q$ eigenvectors is not uniquely defined. 

We first estimate the magnitude of $\lambda_q$ (with additive error $\gamma/100$) using the following theorem.

\begin{theorem}\label{thm:findGap}
    Let $\delta\in(0,1)$, $q<d$, $A\in\mathbb{C}^{d\times d}$ be a Hermitian matrix with operator norm at most~$1$, $v_1,\ldots,v_d$ be an orthonormal basis of eigenvectors of~$A$, and corresponding eigenvalues $\lambda_1,\ldots,\lambda_d$ such that $|\lambda_1|\geq\cdots\geq|\lambda_d|$, where we know the gap $\gamma=|\lambda_q|-|\lambda_{q+1}|$. Suppose $U_A=\exp(\pi iA)$. 
    There is a quantum algorithm that with probability at least $1-\delta$, estimates $|\lambda_q|$ with additive error $\gamma/100$, using ${\mathcal{O}}\Big(\frac{\sqrt{qd}}{\gamma}\log(\frac{\log(1/\gamma)}{\delta})\log(\frac{d}{\delta})\log(\frac{1}{\gamma})\Big)$ controlled applications of $U_A^\pm$ and $\tilde{\mathcal{O}}\Big(\frac{\sqrt{qd}}{\gamma}\log(\frac{\log(1/\gamma)}{\delta})\log(\frac{d}{\delta})\log(\frac{1}{\gamma})\Big)$ time.  
\end{theorem}

\begin{proof}
    Let $\delta'=\delta/(10d)$, $A=\sum\limits_{i\in[d]}\lambda_i \ketbra{v_i}{v_i}$, and $T=2^{\lceil \log (200\log(1/\delta')/\gamma)\rceil+2}$. 
    Unitary $W=\sum\limits_{t=0}^{T-1} \ketbra{t}{t}\otimes\exp(\pi i t A)$ does Hamiltonian simulation according to~$A$ on the second register, for an amount of time specified in the first register.
    Observe that $\frac{1}{\sqrt{d}}\sum\limits_{i\in[d]}\ket{i}\ket{i}=\frac{1}{\sqrt{d}}\sum\limits_{i\in[d]}\ket{v_i}\ket{v_i^*}$ because of the invariance of maximally entangled states under unitaries of the form $U\otimes U^\dagger$. Hence we can apply phase estimation with precision $\gamma/200$ and failure probability $\delta'$ to the quantum state $\frac{1}{\sqrt{d}}\sum\limits_{i\in[d]}\ket{v_i}\ket{v_i^*}\ket{0}$ using $W$, to obtain the state
    \begin{equation}\label{eq:statewithLi}
    \frac{1}{\sqrt{d}}\sum\limits_{i\in[d]}\ket{v_i}\ket{v_i^*}\ket{L_i}, 
    \end{equation}
    where the state $\ket{L_i}$ contains a superposition over different estimates $\tilde{\lambda_i}$ of $\lambda_i$. For each $i\in[d]$, if we were to measure $\ket{L_i}$ in the computational basis, then with probability at least $1-\delta'$ we get an outcome $\tilde{\lambda_i}$ such that $|\lambda_i-\tilde{\lambda_i}|\leq \gamma/200$.\footnote{There's a small technical issue here: the unitary $e^{\pi iA}$ (to which we apply phase estimation) has phases ranging between $-\pi$ and $\pi$ because the $\lambda_j$ range between $-1$ and 1, and phase estimation treats $-\pi$ and $\pi$ the same. However, we can easily fix that by applying phase estimation to the unitary $e^{\pi iA/2}$, whose phases range between $-\pi/2$ and $\pi/2$.} Let $\mu \in[0,1]$, and $R_\mu$ be a unitary that marks whether a number's absolute value is $<\mu$, i.e., for every $a\in[-1,1]$
    $$
    R_\mu\ket{a}\ket{0}=\begin{cases}
        \ket{a}\ket{0}, \text{ if } { |a|\geq \mu }\\
        \ket{a}\ket{1}, \text{ otherwise.}
    \end{cases}
    $$
    This unitary can be implemented up to negligibly small error by $\mathcal{\tilde{O}}(1)$ elementary gates. Applying $R_\mu$ on the last register of the state of Eq.~\eqref{eq:statewithLi} and an additional $\ket{0}$, we obtain $\sqrt{p_\mu}\ket{\phi_0}\ket{0}+\sqrt{1-p_\mu}\ket{\phi_1}\ket{1}$ for some $\ket{\phi_0}$ and $\ket{\phi_1}$, where $p_\mu$ is the probability of outcome~0 if we were to measure the last qubit. Note that if $\mu\geq |\lambda_q|+\gamma/150>|\lambda_q|+\gamma/200$, then $p_\mu \leq (q-1)/d+(d-q+1)\delta'/d$, where the first term on the right-hand side is the maximal contribution (to the probability~$p_\mu$ of getting outcome~0 for the last qubit) coming from $\ket{L_i}$ with $i\leq q-1$ and the second term is the maximal contribution coming from $\ket{L_i}$ with $i>q-1$. On the other hand, if $\mu\leq |\lambda_{q}|-\gamma/150<|\lambda_{q}|-\gamma/200$, then $p_\mu\geq (q/d)\cdot(1-\delta')$, which is the minimal contribution coming from $\ket{L_i}$ with $i\leq q$. The difference between the square-roots of these two values is therefore
    \begin{equation}\label{eq:diffinroots}
        \sqrt{\frac{q}{d}\cdot(1-\delta')}-\sqrt{\frac{q-1}{d}+\frac{d-q+1}{d}\delta'} \geq \sqrt{\frac{q}{d}-\delta'}-\sqrt{\frac{q-1}{d}+\delta'}=\frac{\frac{1}{d}-2\delta'}{\sqrt{\frac{q}{d}-\delta'}+\sqrt{\frac{q-1}{d}+\delta'}}, 
    \end{equation}
    where the last equality is because $a-b=(a^2-b^2)/(a+b)$.
    Because $\delta'=\delta/(10d)\in (0,1/(10d))$, both terms in the denominator are $\leq \sqrt{q/d}$, and hence the right-hand side of Eq.~\eqref{eq:diffinroots} is at least $2/(5\sqrt{qd})$. To estimate $|\lambda_q|$ with additive error $\gamma/100$, it therefore suffices to do binary search over the values of $\mu$ (with precision $\gamma/300$, that is, binary search over $\mu\in\{0,\gamma/300,2\gamma/300,\ldots,1\}$), in each iteration estimating $\sqrt{p_\mu}$ to within $\pm 1/(5\sqrt{qd})$. We can implement the unitary that maps $\ket{0}\ket{0}\rightarrow \sqrt{p_\mu}\ket{\phi_0}\ket{0}+\sqrt{1-p_\mu}\ket{\phi_1}\ket{1}$ using one application of $W$ and $\mathcal{\tilde{O}}(1)$ time. Let $\delta''=\delta/(10\log(300/\gamma))>0$. By \Cref{thm:amplitude_estimation} with additive error $\eta=1/(5\sqrt{qd})=\Theta(1/\sqrt{qd})$ and with failure probability $\delta'',$ one iteration of the binary search succeeds with probability at least $1-\delta''$ and uses $\mathcal{{O}}(\log(1/\delta'')/\eta)=\mathcal{{O}}\big(\sqrt{qd}\log(\frac{\log(1/\gamma)}{\delta})\big)$ applications of $W$ and $W^\dagger$, and $\mathcal{\tilde{O}}\big(\sqrt{qd}\log(\frac{\log(1/\gamma)}{\delta})\big)$ time. 
   Therefore by the union bound, with probability at least $1-(\lceil\log(300/\gamma)\rceil+1)\cdot \delta''\geq 1-\delta$, all iterations of the binary search give a sufficiently good estimate of the value $\sqrt{p_\mu}$ of that iteration, so the binary search gives us an estimate of $|\lambda_q|$ to within $\pm (\gamma/150+\gamma/300)=\pm\gamma/100$.
   
    Since we use $O(\log(1/\gamma))$ iterations of binary search, we use $\mathcal{{O}}\Big(\sqrt{qd}\log(\frac{\log(1/\gamma)}{\delta})\log(\frac{1}{\gamma})\Big)$ applications of $W$ and $W^\dagger$ and $\tilde{\mathcal{{O}}}\Big(\sqrt{qd}\log(\frac{\log(1/\gamma)}{\delta})\log(\frac{1}{\gamma})\Big)$ time for the whole binary search.
    We can implement $W$ and $W^\dagger$ using $\mathcal{O}(T)={\mathcal{O}}(\frac{\log(d/\delta)}{\gamma})$ controlled applications of $U_A^\pm$ and $\tilde{\mathcal{O}}(\frac{\log(d/\delta)}{\gamma})$ time.
   Thus we obtain a good estimate of $|\lambda_q|$ with probability $\geq1-\delta$, using  ${\mathcal{O}}\Big(\frac{\sqrt{qd}}{\gamma}\log(\frac{\log(1/\gamma)}{\delta})\log(\frac{d}{\delta})\log(\frac{1}{\gamma})\Big)$ controlled applications of $U_A^\pm$ and $\tilde{\mathcal{O}}\Big(\frac{\sqrt{qd}}{\gamma}\log(\frac{\log(1/\gamma)}{\delta})\log(\frac{d}{\delta})\log(\frac{1}{\gamma})\Big)$ time.
   \hfill
 \end{proof}

The above theorem assumes perfect access to $\exp(\pi iA)=U_A$ for doing phase estimation. If we only assume we have sparse-query-access to $A$, then by \Cref{thm:Low19OptimalHsim} we can implement a unitary $\widetilde{U}_A$ such that $\|\widetilde{U}_A-\exp(\pi iA)\|\leq \eps$ using $\mathcal{\tilde{O}}({s}^{0.5+o(1)}/\eps^{o(1)})$ time and queries.

Since the procedure in \Cref{thm:findGap} makes use of $D={\mathcal{O}}\Big(\frac{\sqrt{qd}}{\gamma}\log(\frac{\log(1/\gamma)}{\delta})\log(\frac{d}{\delta})\log(\frac{1}{\gamma})\Big)$ controlled applications of $U_A^\pm$, if we replace $U_A$ with $\widetilde{U}_A$, the algorithm still outputs the desired answer with success probability at least $1-\delta-D\eps$.\footnote{Here we use the fact that if two unitaries are $\eps$-close in operator norm, and they are applied to the same quantum state, then the resulting two states are $\eps$-close in Euclidean norm, and the two probability distributions obtained by measuring the resulting two states in the computational basis are $\eps$-close in total variation distance.} 
By plugging in the time complexity for constructing $\widetilde{U}_A$ with $\eps=\Theta\Big(\frac{\delta\gamma}{\sqrt{qd}\log(\log(1/\gamma)/\delta)\log(d/\delta)\log(1/\gamma)}\Big)$, with the constant in the $\Theta(\cdot)$ chosen such that $T\eps\leq \delta$ (and rescaling $\delta$ by factor of $2$), we immediately have the following corollary.

 \begin{corollary}\label{cor:est_q}
    Let $q<d$ and $A\in\mathbb{C}^{d\times d}$ be a Hermitian matrix with operator norm at most~$1$, $v_1,\ldots,v_d$ be an orthonormal basis of eigenvectors of~$A$, and eigenvalues $\lambda_1,\ldots,\lambda_d$ such that $|\lambda_1|\geq\cdots\geq|\lambda_d|$, where we know the gap $\gamma=|\lambda_q|-|\lambda_{q+1}|$, and $\delta \in(0,1)$.  Suppose $A$ has sparsity $s$ and we have sparse-query-access to $A$. There is a quantum algorithm that with success probability at least $1-\delta$, estimates $|\lambda_q|$ with additive error $\gamma/100$, using $\mathcal{\tilde{O}}\Big(\frac{1}{\delta^{o(1)}}(\frac{\sqrt{dqs}}{\gamma})^{1+o(1)}\Big)$ queries and time.
\end{corollary}

The following proposition shows that every bounded-error quantum algorithm needs $\Omega(\sqrt{ds})$ sparse-access queries to estimate the top eigenvalue of an $s$-sparse matrix with constant additive error. This implies the above corollary is near-optimal when $q=1$.

\begin{proposition}\label{prop:LB_top_eigenvalue}
    Let $A\in \mathbb{C}^{d \times d}$ be a Hermitian matrix with operator norm at most $3$. Suppose $A$ has sparsity $s$ and we have sparse-query-access to $A$. Every bounded-error quantum algorithm that estimates the top eigenvalue of $A$ with additive error $0.1$ uses $\Omega(\sqrt{ds})$ queries.
\end{proposition} 

\begin{proof} For simplicity and without loss of generality we assume $A$ has sparsity $2s+1$ and $d\geq 2s+1$ is a multiple of $s$.
    The idea is to encode an $s(d-s)$-bit string into a $2s+1$-sparse $d\times d$ matrix. Given a string $X\in\{0,1\}^{s(d-s)}\coloneq X^{(1)}X^{(2)}\ldots X^{(d/s)-1}$ with Hamming weight either $0$ or $1$, where $X^{(k)}$ is an $s^2$-bit Boolean string for each $k\in [d/s-1]$. For every $k\in [d/s-1]$, define $Y^{(k)}\in\{0,1\}^{s \times s}$ as $(Y^{(k)})_{ij}=X^{(k)}_{s\cdot i+j}$.  Let $A$ be defined by $d/s \cdot d/s=d^2/s^2$ many $s\times s$ square matrices such that for $i\geq j$ %$A_{ij}\in \mathbb{R}^{s\times s}$,
    $$A_{ij}=\begin{cases}
        I_s &  \text{ if } i=j\\
        Y^{(i)}+2^{-d}\cdot J_s  & \text{ if } i=j+1\\
        0_s & \text{ otherwise, }\\
    \end{cases}$$
    where $I_s$ is the $s\times s$ identity matrix, $J_s$ is the $s\times s$ all-$1$ matrix, and $0_s$ is the $s\times s$ all-$0$ matrix; and for $i<j$, $A_{ij}=A_{ji}^T$. One can easily see that $A$ has sparsity $2s+1$, and because the Hamming weight Ham($X$) of $X$ is at most $1$, the operator norm of $A$ is at most $2+2s\cdot 2^{-d}\leq 2+d\cdot 2^{-d} \leq 3$. Note that given access to the oracle $O_X $ that maps $\ket{i}\ket{0}\rightarrow\ket{i}\ket{X_i}$ for every $i\in[s(d-s)]$, one can construct an oracle that allows us to make sparse-query-access to $A$ using $2$ applications of $O_X^{\pm}$. %and $\tilde{\mathcal{O}}(1)$ time. 
    
  %  Now we are ready to show the quantum query lower bound for estimating the top eigenvector of $A$ with additive error $0.1$. Suppose there is a quantum algorithm $\mathcal{A}$ that solves this task in time $T$. 
    Observe that if Ham($X$)=0, then the operator norm of $A$ is at most $1+2s\cdot 2^{-d}\leq 1+d\cdot2^{-d} \leq 1+1/(e\cdot \ln 2)<  1.6$, while if Ham($X$)=1, then the operator norm of $A$ is at least $2$. Therefore, if there exists a $T$-query quantum algorithm $\mathcal{A}$ that estimates the top eigenvector of $A$ with additive error $0.1$, then $\mathcal{A}$ can also be used to decide whether the $s (d-s)$-bit string $X$ has Hamming weight~$0$ or~$1$ using $2T$ queries to $O_X$. By the well-known quantum query lower bound for search, every bounded-error quantum algorithm needs $\Omega(\sqrt{ds})$ queries for this, implying  $T=\Omega(\sqrt{ds})$.
    \hfill
\end{proof}

Note that a $\sqrt{ds}$ lower bound for estimating the top eigenvalue implies a $\sqrt{ds}$ lower bound for approximating the quantum state $\ket{v_1}$ of the top eigenvector: once we have $\ket{v_1}$, we can apply the quantum phase estimation algorithm with precision $\eps$ to estimate $\lambda_1$ with additive error $\eps$ using $\bigOt{1/\eps}$ controlled applications of $\exp(\pi iA)$ and $\exp(-\pi iA)$. Combining this with Hamiltonian simulation (\Cref{thm:Low19OptimalHsim}) and our lower for approximating $\lambda_1$ (\Cref{prop:LB_top_eigenvalue}), we conclude that every bounded-error quantum algorithm that outputs a state at $\ell_2$-distance $\leq 0.05$ from $\ket{v_1}$, needs $\Omega(\sqrt{ds})$ queries.

Once we know $|\lambda_q|$, we can apply \cite[Theorem 31]{gilyen2018QSingValTransf} to implement a block-encoding of $U_\Pi$ using $\bigOt{1/\gamma}$ applications of a block-encoding of~$A$, where $\Pi=\sum_{i\in[q]}v_iv_i^\dagger$. Combining the above argument with \Cref{thm:EntrytoBlock}, we directly get the following corollary of \Cref{thm:LearnProj}:

\begin{corollary}\label{cor:qeigenspacesparse}
    Let $q< d$ and $A\in\mathbb{C}^{d\times d}$ be a Hermitian matrix with operator norm at most~$1$, $v_1,\ldots,v_d$ be an orthonormal eigenbasis of~$A$ with respective eigenvalues $\lambda_1,\ldots,\lambda_d$ such that $|\lambda_1|\geq\cdots\geq|\lambda_d|$,  where we know the gap $\gamma=|\lambda_q|-|\lambda_{q+1}|$. Let $\eps,\delta \in(0,1)$ and $\Pi=\sum_{i\in[q]}v_iv_i^\dagger$. Suppose $A$ has sparsity $s$ and we have sparse-query-access to $A$.
    There exists a quantum algorithm that outputs a $d\times q$ matrix $W$ with orthonormal columns such that, with probability at least $1-\delta$, $\|WW^\dagger-\Pi\|\leq \eps$, using $\bigOt{\left(\frac{d\sqrt{s}q}{\gamma\eps}\right)^{1+o(1)}+\frac{1}{\delta^{o(1)}}\left(\frac{\sqrt{dqs}}{\gamma}\right)^{1+o(1)}+dq^2}$ time and $\mathcal{\tilde{O}}(d)$ QRAM bits. 
\end{corollary}

For the case of dense matrix $A$, we can set $s=d$ to get time complexity roughly $qd^{1.5}$.
The special case $q=1$ gives our main result for approximating the top eigenvector (with additive $\ell_2$-error $\eps$)\footnote{
Note that for every unit $w, v \in \mathbb{C}^d$ it holds that $\nrm{ww^\dagger-vv^\dagger}=2\sqrt{1-|\inProd{w}{v}|^2}$, thus $\eps \geq \nrm{ww^\dagger-vv^\dagger}$ implies $|\inProd{w}{v}|\geq \sqrt{1-\eps^2/4}\geq {1-\eps^2/4}\geq 1-\eps^2/2$.} in time $\bigOt{\big(d^{1.5}/(\gamma\eps)\big)^{1+o(1)}}$. The $\eps$-dependency is slightly worse than the algorithm in \Cref{sec:qNPM_G}, while both the $d$-dependency and $\gamma$-dependency are significantly better (for $d$, the power is $1.5+o(1)$ instead of $ 1.75$; for $\gamma$, the power is $1+o(1)$ instead of $2$). In \Cref{sec:LB} we show that its $d$-dependence to be essentially optimal. However, the complexity with respect to $q$ is sub-optimal for $q$ close to~$d$, because one can diagonalize the entire matrix $A$ classical in matrix-multiplication time $\mathcal{O}(d^\omega)$.

We also remark that one can (approximately) prepare the quantum state $\ket{v_1}$ of the top eigenvector of the matrix $A$  (with success probability $\geq 0.9$) in time  $\sim(\frac{\sqrt{ds}}{\gamma})^{1+o(1)}$ without using QRAM: first use~\Cref{cor:est_q} to estimate $\lambda_1$ (with additive error $\gamma/100$ and with success probability $\geq 0.99$) and then implement an $\epsilon'$-approximate block-encoding of projector $\ket{v_1}\bra{v_1}$ in time  $\tilde{\mathcal{O}}\Big(\frac{\sqrt{s}({s/\epsilon'})^{o(1)}}{\gamma}\Big)$ by \Cref{thm:EntrytoBlock} and~\cite[Theorem 31]{gilyen2018QSingValTransf}. After that, we generate a vector $ g $ with i.i.d.\ (complex) standard normal entries, prepare the corresponding normalized quantum state $\ket{g}$  (with probability $\geq 0.99-\exp(-d/2)$ it has overlap $\geq 1/(100\sqrt{2d})$ with $v_1$)\footnote{Since $ g $ has i.i.d.\ (complex) standard normal entries, $\braket{v_1|g}$ has (complex) standard normal distribution. Since by~\Cref{prop:rndVecLength}, with probability $\geq 1-\exp(-d/2)$, $\|g\|\leq \sqrt{2d}$, and since the pdf of the standard normal distribution is $p(x)=\frac{1}{\sqrt{2\pi}}\exp(-x^2/2)$), we know with probability $\geq 1-\exp(-d/2)-0.01\cdot 2 \cdot\frac{1}{\sqrt{2\pi}}\geq 0.99-\exp(-d/2)$, $|\braket{v_1|g}|/\|g\|\geq 1/(100\sqrt{2d})$.} and apply the block-encoding to project it down, obtaining a state that looks like $\alpha\ket{0}\ket{v_1}+\sqrt{1-|\alpha|^2}\ket{1}\ket{v_1^\perp}$ where $|\alpha|=\Omega(1/\sqrt{d})$.  By \Cref{thm:Fixed_AA}, $\bigO{\sqrt{d}\log (1/\eps)}$ amplitude amplification rounds then suffice to prepare $\ket{v_1}$ (up to $\ell_2$-norm error $\eps/2$), and the total cost is therefore $\bigOt{\frac{1}{\eps^{o(1)}}(\frac{\sqrt{ds}}{\gamma})^{1+o(1)}}$ by choosing $\eps'=\Theta(\eps/(\sqrt{d}\log (1/\eps)))$ such that $T\eps'\leq \eps/2$, where $T$ is number of amplitude amplification rounds. This algorithm is near-optimal in its $d$-dependence and $s$-dependence because of the argument after the proof of~\Cref{prop:LB_top_eigenvalue}.

\begin{corollary}\label{cor:prepare_v1}
    Let $A\in\mathbb{C}^{d\times d}$ be a Hermitian matrix with operator norm at most~$1$, $v_1,\ldots,v_d$ be an orthonormal eigenbasis of~$A$ with respective eigenvalues $\lambda_1,\ldots,\lambda_d$ such that $|\lambda_1|\geq\cdots\geq|\lambda_d|$,  where we know the gap $\gamma=|\lambda_1|-|\lambda_{2}|$. Let $\eps \in(0,1)$. 
    Suppose $A$ has sparsity $s$ and we have sparse-query-access to $A$.
    There exists a quantum algorithm that with probability at least $0.98-\exp(-d/2)$, outputs $\ket{\tilde{v_1}}$ that approximates $\ket{v_1}$ with $\ell_2$-error $\eps$ using $\mathcal{\tilde{O}}\Big(\frac{1}{\eps^{o(1)}}(\frac{\sqrt{ds}}{\gamma})^{1+o(1)}\Big)$ time. 
\end{corollary}

\section{Lower bounds for approximating the top eigenvector}\label{sec:LB}

In this section we prove essentially tight classical and quantum query lower bounds for approximating the top eigenvector of a matrix whose entries we can query. 

\subsection{The hard instance for the lower bound}\label{ssec:hardinstance}

Consider the following case, which is the ``hard instance'' for which we prove the lower bounds. Let $u\in\{-1,1\}^d$ be a vector, and define symmetric random matrix $A=\frac{1}{d}uu^T+N$ where the entries of $N$ are i.i.d.\ $N_{ij}\sim N(0,\frac{1}{4\cdot 10^6d})$ for all $1\leq i\leq j\leq d$ (and $N_{ij}=N_{ji}$ if $i>j$); the goal is to recover most (say, 99\%) of the entries of the vector~$u$. In this problem, the information about the $u_i$-s is hidden in the matrix $A$: the entry $A_{ij}$ is clearly a sample from $N(\frac{u_iu_j}{d},\frac{1}{4\cdot 10^6d})$. Hence to learn entries of~$u$, intuitively we should be able to distinguish the distribution $N(\frac{1}{d},\frac{1}{4\cdot 10^6d})$ from the distribution $N(-\frac{1}{d},\frac{1}{4\cdot 10^6d})$. In the classical case (where querying an entry of $A_{ij}$ is the same as obtaining one sample fom the distribution) it requires roughly $\Omega(d)$ queries to the entries of the $i$th row and column to learn one $u_i$, even if all $u_j$ with $j\neq i$ are already known. In the quantum case, it requires $\Omega(\sqrt{d})$ queries. Intuitively, learning $0.99d$ of the $u_i$-s should then require roughly $d$ times more queries, so $\Omega(d^2)$ and $\Omega(d^{1.5})$ classical and quantum queries in total, respectively. We show in the next two subsections that this is indeed the case.

First we show that (with high probability) $A$ has a large eigenvalue gap.
Note that $A=\frac{1}{d}uu^T+\frac{1}{2000\sqrt{d}}G$ where $G_{ij}\sim N(0,1)$ for $1\leq i\leq j\leq d$ and $G_{ij}=G_{ji}$ for the lower-triangular elements. 
Since $G$ itself is symmetric and its entries are i.i.d.\ $\sim N(0,1)$, by \Cref{thm:GauOperator} with $t=0.4\sqrt{d}$ (here $b_{max}=\sqrt{d}$ and $b^*_{max}=1$) we have that 
the operator norm of $G$ is upper bounded by $2.5\sqrt{d}+(0.4+o(1))\sqrt{d}\leq 3\sqrt{d}$ with probability
at least $1-\exp(-0.04d)$;
below we assume this is indeed the case. Therefore, by triangle inequality, the top eigenvalue of $A$ is upper bounded by $1+\frac{3}{2000}$ and lower bounded by $\|A\frac{u}{\sqrt{d}}\|_2\geq 1-\frac{3}{2000}$, implying that there is a unit top eigenvector $v_1=v_1(A)$ of~$A$ that has inner product nearly~1 with the unit vector $\frac{u}{\sqrt{d}}$:
$$
1-\frac{3}{2000}\leq \|Av_1\|_2 \leq \|\frac{1}{d}uu^Tv_1\|_2+\|\frac{1}{2000\sqrt{d}}Gv_1\|_2\leq |\frac{1}{\sqrt{d}}u^Tv_1|+ \frac{3}{2000},
$$
hence $|\inProd{v_1}{\frac{u}{\sqrt{d}}}|\geq 1-\frac{3}{1000}$. We may assume without loss of generality that the eigenvector~$v_1$ has been chosen such that $\inProd{v_1}{\frac{u}{\sqrt{d}}}$ is positive, so we can ignore the absolute value sign.
The signs of $v_1(A)$ have to agree with the signs of $u$ in at least $99.4\%$ of the $d$ entries, because each entry where the signs are different contributes at least $1/d$ to the squared distance between $u/\sqrt{d}$ and $v_1(A)$:
$$
 \frac{1}{d} \#\{j\in[d] \mid u_j\cdot (v_1(A))_j \leq 0\}\leq \|\frac{u}{\sqrt{d}}-v_1(A)\|_2^2 = 2-2\inProd{\frac{u}{\sqrt{d}}}{v_1(A)}\leq \frac{3}{500}.
$$
Moreover, the second eigenvalue $\lambda_2(A)$ of $A$ is at most $0.08$:
$$
\lambda_2(A)=\max\limits_{w:\|w\|_2=1, w\perp v_1 }\|Aw\|_2\leq \frac{3}{2000}+ \max\limits_{w:\|w\|_2=1, w\perp v_1 }\|\frac{uu^Tw}{d}\|_2\leq \frac{3}{2000}+\frac{\sqrt{3\cdot1997}}{997}<0.08,
$$
where the last inequality holds because $\inProd{v_1}{\frac{u}{\sqrt{d}}}\geq 1-\frac{3}{1000}$.\footnote{Let $v_1=\alpha_1\frac{u}{\sqrt{d}}+\beta_1w_1$ and $v_2=\alpha_2\frac{u}{\sqrt{d}}+\beta_2w_2$ for some unit vectors $w_1,w_2 \perp \frac{u}{\sqrt{d}}$ and for some $\alpha_1,\alpha_2,\beta_1,\beta_2\in[-1,1]$ satisfying $\alpha_1^2+\beta^2_1=\alpha_2^2+\beta^2_2=1$. Since $v_1 \perp v_2$, we have $\langle v_1,v_2\rangle=\alpha_1\alpha_2+\beta_1\beta_2\langle w_1,w_2\rangle=0$, implying $|\alpha_2|=\frac{|\beta_1\beta_2|}{|\alpha_1|}|\langle w_1,w_2\rangle|\leq \frac{|\beta_1|}{|\alpha_1|}=\frac{\sqrt{1-\alpha^2_1}}{|\alpha_1|}\leq \frac{\sqrt{3\cdot1997}}{{997}}$.} Hence there is a constant gap between the top and the second eigenvalue of $A$. 

If we have an algorithm that outputs a vector $\tilde{u}$ satisfying $\|\tilde{u}-v_1(A)\|_2\leq \frac{1}{1000}$, then we can use the signs of $\tilde{u}$ to learn $99\%$ of the $u_i$-s. Hence a (classical or quantum) query lower bound for recovering (most of) $u$ is also a query lower bound for approximating the top eigenvector of our hard instance.

\subsection{A classical lower bound}\label{ssec:clasLBtopeigenvector}

We first show that every classical algorithm that recovers $99\%$ of the $u_i$-s needs $\Omega(d^2)$ queries.

\begin{theorem}\label{thm:CLB}
    Let $u\in\{-1,1\}^d$ be a vector, $e_{11},e_{12},\ldots,e_{1d},e_{22},\ldots,e_{dd}$ be $d(d+1)/2$ independent samples drawn from $N(0,\frac{1}{4\cdot 10^6d})$, and $A\in \mathbb{R}^{d\times d}$ be the matrix defined by 
    $$
    A_{ij}=\begin{cases}
        \frac{1}{d}u_iu_j+e_{ij}, \text{ if } 1\leq i\leq j\leq d,\\
        A_{ji}, \text{ otherwise. }
    \end{cases}
    $$ 
    Suppose we have query access to entries of~$A$. Every bounded-error classical algorithm that computes a $\tilde{u}\in\{-1,1\}^d$ at Hamming distance $\leq d/100$ from $u$, uses $\Omega{(d^2)}$ queries.
\end{theorem}

\begin{proof}
    Suppose there exists a $T$-query bounded-error classical algorithm $\mathcal{A}$ to compute such a $\tilde{u}$, with worst-case error probability $\leq 1/20$. Note that the only entries that depend on $u_i$ are in the $i$th column and row of $A$. Let random variable $T_i$ be the number of queries that $\cal A$ makes in the $i$th column and row (here the randomness comes from the input distribution and from the internal randomness of~$\cal A$). Because every query is counted at most twice among the $T_i$-s (and only once if the query is to a diagonal entry of $A$), we have $\E[\sum_{i\in[d]}T_i]\leq 2T$. Define index $i$ as ``good'' if $\Pr[u_i=\tilde{u}_i]\geq 0.8$, and let $I_G$ be the set of good indices. Since $\mathcal{A}$ has error probability at most~$1/20$, we can bound the expected Hamming distance by
    $$
    \E[\text{Ham}(u,\tilde{u})]\leq \Pr[\text{Ham}(u,\tilde{u})\leq \frac{d}{100}]\cdot \frac{d}{100}+\Pr[\text{Ham}(u,\tilde{u})> \frac{d}{100}]\cdot d\leq 1\cdot \frac{d}{100}+\frac{1}{20}\cdot d\leq \frac{d}{10}.
    $$
    By plugging in the definition of $I_G$, we obtain
    $$
    \frac{d}{10}\geq \E[\text{Ham}(u,\tilde{u})]= \sum\limits_{i\in[d]}\Pr[u_i\neq \tilde{u}_i] 
    \geq\sum\limits_{i\in[d]\setminus I_G}\Pr[u_i\neq \tilde{u}_i]\geq\sum\limits_{i\in[d]\setminus I_G} \frac{1}{5}= 
    \frac{d- |I_G|}{5},
    $$
    implying that $|I_G|\geq d/2$. Because $\E[\sum_{i\in I_G}T_i]\leq\E[\sum_{i\in[d]}T_i]\leq 2T$, by averaging there exists an index $\mathbf{i}\in I_G$ such that $\E[T_i]\leq 4T/d$.  This implies that there is a classical algorithm $\mathcal{A'}$ that recovers $u_{\mathbf{i}}$ with probability at least $0.8$ using an expected number of at most $4T/d$ queries to entries in the $\mathbf{i}$th row and column of $A$.
    
    Now suppose we want to distinguish $u_{\mathbf{i}}=1$ from $u_{\mathbf{i}}=-1$ using samples from $N(\frac{u_{\mathbf{i}}}{d},\frac{1}{4\cdot 10^6d})$. We can use $\mathcal{A'}$ for this task, as follows. Generate $u_1,\ldots,u_{\mathbf{i}-1},u_{\mathbf{i}+1},\ldots,u_d$ uniformly at random from $\pm1$, and generate $e'_{ij}$ from $N(0,\frac{1}{4\cdot 10^6d})$ for all $1\leq i\leq j\leq d$. Then we define a $d\times d$ matrix $A'$ as
    \begin{align*}
        A'_{ij} = \begin{cases}
  \frac{1}{d}u_iu_j+e'_{ij},  & \text{if  $1\leq i\leq j\leq d$, $\mathbf{i}\not\in\{i,j\}$}\\
 u_j\cdot\text{sample from }N(\frac{u_{\mathbf{i}}}{d},\frac{1}{4\cdot 10^6d}) & \text{if $i=\mathbf{i}$ and $i<j$} \\
  u_i\cdot\text{sample from }N(\frac{u_{\mathbf{i}}}{d},\frac{1}{4\cdot 10^6d}) & \text{if $j=\mathbf{i}$ and $i<j$} \\
 A'_{ji}, & \text{otherwise}.
\end{cases}
    \end{align*}
    Note that for every $i,j\in[d]$, $A'_{ij}$ is a sample drawn from $N(\frac{u_iu_j}{d},\frac{1}{4\cdot 10^6d})$.
Given that our algorithm knows the values it generated itself (in particular, all $u_j$ with $j\neq\mathbf{i}$), it can implement one query to an entry in the $\mathbf{i}$th row or column of $A'$ by at most one new sample from $N(\frac{u_{\mathbf{i}}}{d},\frac{1}{4\cdot 10^6d})$, and queries to other entries of~$A'$ do not cost additional samples.
    
    By running $\mathcal{A}'$ on the input matrix $A'$, 
    we have $\Pr[u_{\mathbf{i}}=\tilde{u}_{\mathbf{i}}]\geq 0.8$.  Hence by using an expected number of $4T/d$ samples drawn from $N(\frac{u_{\mathbf{i}}}{d},\frac{1}{4\cdot 10^6d})$, we can distinguish $u_{\mathbf{i}}=1$ from $u_{\mathbf{i}}=-1$ with probability $\geq 0.8$. Then by Markov's inequality,  worst-case $40T/d$ samples suffice to distinguish $u_{\mathbf{i}}=1$ from $u_{\mathbf{i}}=-1$ with probability $\geq 0.7$.  
    
    The KL-divergence between $N(\frac{1}{d},\frac{1}{4\cdot 10^6d})$ and $N(-\frac{1}{d},\frac{1}{4\cdot 10^6d})$ is $\E_{x\sim N(\frac{1}{d},\frac{1}{4\cdot 10^6d})}[8\cdot 10^6x]=\frac{8\cdot 10^6}{d}$. As a result, by using the well-known Pinsker's inequality ($d_{TV}(P,Q)\leq \sqrt{\frac{1}{2}D_{KL}(P\,\|\,Q)}$ \cite{kullback1967PinskerTight,csiszar1967InfMeasProbDist,kemperman1969OptInfTransmitionRate}) and the fact that $D_{KL}(P^{\otimes t}\,\|\,Q^{\otimes t})=t\cdot D_{KL}(P\,\|\,Q)$ for any distributions $P,Q$ and natural number $t$, we obtain
    \begin{align*}
        \Omega(1)&\leq{d_{TV}\Big(N(\frac{1}{d},\frac{1}{4\cdot 10^6d})^{\otimes \frac{40T}{d}},N(-\frac{1}{d},\frac{1}{4\cdot 10^6d})^{\otimes \frac{40T}{d}}\Big)}\\
        &\leq\sqrt{\frac{1}{2}D_{KL}\Big(N(\frac{1}{d},\frac{1}{4\cdot 10^6d})^{\otimes \frac{40T}{d}}\,\Big\|\, N(-\frac{1}{d},\frac{1}{4\cdot 10^6d})^{\otimes \frac{40T}{d}}\Big)}\tag*{(by Pinsker's inequality)}\\
        &=\sqrt{\frac{1}{2}\cdot \frac{40T}{d}\cdot D_{KL}\Big(N(\frac{1}{d},\frac{1}{4\cdot 10^6d})\,\Big\|\, N(-\frac{1}{d},\frac{1}{4\cdot 10^6d})\Big)} =  \bigO{\sqrt{T/d^2}},
    \end{align*}
    implying $T=\Omega(d^2)$.
\end{proof}

By the discussion in \Cref{ssec:hardinstance}, we therefore obtain the following corollary.

\begin{corollary}\label{cor:classicalLB}
    Let $A$ be a $d\times d$ symmetric matrix with $\norm{A}=\mathcal{O}(1)$ and an $\Omega(1)$ gap between its top and second eigenvalues.
    Suppose we have query access to the entries of $A$. Every classical algorithm that with probability at least $\geq 99/100$, approximates the top eigenvector of $A$ with $\ell_2$-error at most $\frac{1}{1000}$ uses $\Omega(d^2)$ queries.
\end{corollary}

This query lower bound is tight up to the constant factor, since we can compute the top eigenvector exactly using $d^2$ queries: just query every entry of~$A$ and diagonalize the now fully known matrix~$A$ (without any further queries) to find the top eigenvector exactly.

\subsection{A quantum lower bound}

Now we move to the quantum case, still using the same hard instance. Our proof uses similar ideas as the hybrid method~\cite{BBBV97Hybrid} and adversary method~\cite{Amb02AdvMethod,Amb06weightedAdv} for quantum query lower bounds, but adjusted to continuous random variables instead of input bits.

\begin{theorem}
 Let $u\in\{-1,1\}^d$ be a uniformly random vector, $e_{11},e_{12},\ldots, e_{1d},e_{22},\ldots,e_{dd}$ be $d(d+1)/2$ independent samples drawn from $N(0,\frac{1}{4\cdot 10^6d})$, and $A\in \mathbb{R}^{d\times d}$ be the random matrix defined by 
    $$
    A_{ij}=\begin{cases}
        \frac{1}{d}u_iu_j+e_{ij}, \text{ if } 1\leq i\leq j\leq d,\\
        A_{ji}, \text{ otherwise. }
    \end{cases}
    $$  
    Suppose we have quantum query access to entries of $A$. Every bounded-error quantum algorithm that computes $\tilde{u}\in\{-1,1\}^d$ at Hamming distance $\leq d/100$ from $u$, uses $\Omega{(d^{1.5}/\sqrt{\log d})}$ queries. 
\end{theorem}

\begin{proof}
 Let $\nu$ denote the input distribution given in the theorem statement, and $X\sim\nu$ be a random $d\times d$ input matrix according to that distribution (with instantiations of random variable~$X$ denoted by lower-case~$x$). Let $O_x$ denote the query oracle to input matrix $x$. Suppose there exists a $T$-query quantum algorithm $\mathcal{A}=U_TO_xU_{T-1}\cdots U_1O_xU_0$, alternating queries and input-independent unitaries on some fixed initial state (say, all-0), to compute such a $\tilde{u}$ with error probability $\leq 1/20$, probability taken over both $\nu$ and the internal randomness of $\mathcal{A}$ caused by the measurement of its final state. 
 Our goal is to lower bound~$T$.
 
 For $t\in\{0,\ldots,T-1\}$, let $\ket{\psi^t_x}=\sum\limits_{i,j\in[d]}\alpha^t_{xij}\ket{i,j}\ket{\phi_{xij}^t}$ be the quantum state of algorithm $\cal A$ just before its $(t+1)$st query on input matrix $x$. Let $\ket{\psi^T_x}$ be the final state on input~$x$. 
    We define the query mass on $(i,j)\in[d]\times[d]$ on input $x$ as $p_{ijx}=\sum\limits_{t\in[T]}|\alpha_{xij}^t|^2$, and define the query mass on $i$ on input $x$ as $p_{ix}=\sum\limits_{t\in[T],j\in[d]}|\alpha_{xij}^t|^2+|\alpha_{xji}^t|^2$. 
    Note that $\sum_{i,j}|\alpha^t_{xij}|^2\leq 1$ for all $x$ and $t$.  Every $|\alpha^t_{xij}|^2$ is counted once in $p_{ix}$ and once in $p_{jx}$ if $i\neq j$, and counted only in $p_{ix}$ if $i=j$. Hence we have $\sum_{i\in[d]}p_{ix}\leq 2T$ for every~$x$.    Define $T_i=\Exp_{x\sim\nu}[p_{ix}]$ as the expected query mass on the $i$th row and column of the input matrix, then $\sum_{i\in[d]}T_i\leq 2T$.

Call index $i\in[d]$ ``good'' if $\Pr[u_i=\tilde{u}_i]\geq 0.8$, where the probability is taken over~$\nu$ and the internal randomness of the algorithm. 
    Let $I_G$ be the set of good indices.
    Since $\mathcal{A}$ has error probability at most~$1/20$, we have
    $$
    \E_x[\text{Ham}(u,\tilde{u})]\leq \Pr[\text{Ham}(u,\tilde{u})\leq \frac{d}{100}]\cdot \frac{d}{100}+\Pr[\text{Ham}(u,\tilde{u})> \frac{d}{100}]\cdot d
    \leq 1\cdot \frac{d}{100}+\frac{1}{20}\cdot d
    \leq \frac{d}{10}.
    $$
    Using linearity of expectation and the definition of $I_{G}$, we have
    $$
    \frac{d}{10}\geq \E_x[\text{Ham}(u,\tilde{u})]= \sum\limits_{i\in[d]}\Pr[u_i\neq \tilde{u}_i] \geq \sum\limits_{i\in[d]\setminus I_G}\Pr[u_i\neq \tilde{u}_i]\geq \frac{d- |I_G|}{5},
    $$
    which implies $|I_G|\geq d/2$. Since $\sum_{i\in I_G}T_i\leq \sum_{i\in[d]}T_i\leq 2T$, by averaging there exists an index $\mathbf{i}\in I_G$ such that $T_{\mathbf{i}}\leq 4T/d$.
   We fix this $\mathbf{i}$ for the rest of the proof. Note that because $\mathbf{i}$ has $\Pr_\nu[u_{\mathbf{i}}=\tilde{u}_{\mathbf{i}}]\geq 0.8$, we also have
    $\Pr_{\nu_+}[u_{\mathbf{i}}=\tilde{u}_{\mathbf{i}}]\geq 0.6$ and $\Pr_{\nu_-}[u_{\mathbf{i}}=\tilde{u}_{\mathbf{i}}]\geq 0.6$, where the distributions $\nu_+$ and $\nu_-$ are $\nu$ conditioned on $u_{\mathbf{i}}=1$ and $u_{\mathbf{i}}=-1$, respectively.

We now define an (adversarial) joint distribution $\mu$ on $(X,Y)$-pairs of input matrices, such that the marginal distribution of $X$ is $\nu_+$ and the marginal distribution of $Y$ is $\nu_-$.
First sample a matrix $x$ (with associated $u\in\{-1,1\}^d$ with $u_{\mathbf{i}}=1$) according to $\nu_+$. We want to probabilistically modify this into a matrix $y$ by changing only a small number of entries, and only in the $\mathbf{i}$th row and column of~$x$.
Let $f$ and $g$ be the pdf of $N(\frac{1}{d},\frac{1}{4\cdot 10^6d})$ and $N(-\frac{1}{d},\frac{1}{4\cdot 10^6d})$, respectively. Consider an entry $x_{\mathbf{i}j}$ in the $\mathbf{i}$th row of $x$, with $j\neq \mathbf{i}$. Conditioned on the particular $u$ we sampled, its pdf was $f$ if $u_j=1$ and $g$ if $u_j=-1$. If the pdf of $x_{\mathbf{i}j}$ was $f$, then obtain $y_{\mathbf{i}j}$ from $x_{\mathbf{i}j}$ as follows: if $x_{\mathbf{i}j} >0$, then negate it with probability $\frac{f(x_{\mathbf{i}j})-g(x_{\mathbf{i}j})}{f(x_{\mathbf{i}j})}$, else leave it unchanged.

\medskip

\begin{claim}
   If $x_{\mathbf{i}j}\sim N(\frac{1}{d},\frac{1}{4\cdot 10^6d})$, then $y_{\mathbf{i}j}\sim N(-\frac{1}{d},\frac{1}{4\cdot 10^6d})$.
\end{claim}

\begin{claimproof}
Let $h$ be the pdf of $y_{\mathbf{i}j}$. 
    For a value $z>0$, we have
    $
    h(z)=f(z)-f(z)\cdot \frac{f(z)-g(z)}{f(z)}=g(z).
    $\\
    For $z\leq 0$ we have
    $
    h(z)=f(z)+f(-z)\cdot \frac{f(-z)-g(-z)}{f(-z)}=f(z)+f(-z)-g(-z)=f(-z)=g(z).
    $
\end{claimproof}

\bigskip

If the pdf of $x_{\mathbf{i}j}$ was $g$ instead of $f$, then we do something analogous: if $x_{\mathbf{i}j}<0$, then negate it with probability $\frac{g(x_{\mathbf{i}j})-f(x_{\mathbf{i}j})}{g(x_{\mathbf{i}j})}$. This gives the analogous claim: the pdf of $y_{\mathbf{i}j}$ is then $f$.

Let matrix $y$ be obtained by applying this probabilistic process to all entries in the $\mathbf{i}$th row of~$x$, and changing the entries in the $\mathbf{i}$th column to equal the new $\mathbf{i}$th row (since the resulting $y$ needs to be a symmetric matrix). Outside of the $\mathbf{i}$th row and column, $x$ and $y$ are equal. Let $\mu$ be the resulting joint distribution on $(x,y)$ pairs. An important property of this distribution that we use below, is that the $d\times d$ matrices $x$ and $y$ typically only differ in roughly $\sqrt{d}$ entries, because the probability with which $x_{\mathbf{i}j}$ is modified (=negated) is $\mathcal{O}(1/\sqrt{d})$. The marginal distribution of $Y$ is $\nu_-$, because the change we made in the $X$-distribution corresponds exactly to changing $u_{\mathbf{i}}$ from 1 to $-1$. We could equivalently have defined $\mu$ by first sampling $Y\sim\nu_-$, and then choosing $x_{ij}$ by an analogous negating procedure on~$y_{ij}$.

We now use the general template of the adversary method~\cite{Amb02AdvMethod} together with our distribution $\mu$ to lower bound the total number~$T$ of queries that $\mathcal{A}$ makes. Define a progress measure $S_t=\E_{xy\sim \mu}[\langle{\psi_{x}^t}|{\psi_{y}^t}\rangle]$. As usual in the adversary method, this measure is large at the start of the algorithm and becomes small at the end: $S_0=1$ because $\langle{\psi_{x}^0}|{\psi_{y}^0}\rangle=1$ for all $x,y$ (since the initial state is fixed, independent of the input); and $S_T\leq 1-\Omega(1)$ because for $(x,y)\sim \mu$, our algorithm outputs~1 with probability at least 0.6 on $x$ and outputs $-1$ with probability at least 0.6 on $y$, meaning that $\langle{\psi_{x}^T}|{\psi_{y}^T}\rangle$ is typically bounded below~1.
 Let $\Delta_t=|S_{t+1}-S_t|$ be the change in the progress measure due to the $(t+1)$st query. We can upper bound that change as follows:
    \begin{align*}
        \Delta_t &= |\E_{xy\sim \mu}[\langle{\psi_{x}^{t+1}}|{\psi_{y}^{t+1}}\rangle-\langle{\psi_{x}^t}|{\psi_{y}^t}\rangle]|\\
         &=|\E_{xy\sim \mu}[\langle{\psi_{x}^t}|\big( O_x^\dagger O_y-I \big)|{\psi_{y}^t}\rangle]|\\
        &=|\E_{xy\sim \mu}\left[\sum\limits_{i,j\in[d]} \alpha^t_{xij}\bra{i,j}\bra{\phi^t_{xij}} \sum\limits_{i,j:x_{ij}\neq y_{ij}}\alpha^t_{yij}\ket{i,j}\ket{\phi'^t_{yij}} \right]|\\
      & \leq \E_{xy\sim \mu} \left[\sum\limits_{i,j:x_{ij}\neq y_{ij}}|\alpha^t_{xij}|\cdot|\alpha^t_{yij}|\right]\\ 
        & = \E_{xy\sim \mu}\left[\sum\limits_{j:x_{\mathbf{i}j}\neq y_{\mathbf{i}j}}|\alpha^t_{x\mathbf{i}j}|\cdot|\alpha^t_{y\mathbf{i}j}|+\sum\limits_{j:x_{j\mathbf{i}} \neq y_{j\mathbf{i}}}|\alpha^t_{xj\mathbf{i}}|\cdot|\alpha^t_{yj\mathbf{i}}|\right]\\
        & \leq  \frac{1}{2}\E_{xy\sim \mu}\left[\sum\limits_{j:x_{\mathbf{i}j}\neq y_{\mathbf{i}j}}\Big({|\alpha^t_{x\mathbf{i}j}|^2}+{|\alpha^t_{y\mathbf{i}j}|^2}\Big)+\sum\limits_{j:x_{j\mathbf{i}}\neq y_{j\mathbf{i}}}\Big({|\alpha^t_{xj\mathbf{i}}|^2}+{|\alpha^t_{yj\mathbf{i}}|^2}\Big)\right].   
        \end{align*}
    where we use that $\ket{\psi^{t+1}_x}=U_{t+1}O_x\ket{\psi^t_x}$ and $\ket{\psi^{t+1}_y}=U_{t+1}O_y\ket{\psi^t_y}$, that $O^\dagger_xO_y:\ket{i,j,b}\rightarrow \ket{i,j}\ket{b-x_{ij}+y_{ij}}$, that $x$ and $y$ only differ in the $\mathbf{i}$th row and column, and the AM-GM inequality ($ab\leq(a^2+b^2)/2$) in the last step.
    
   Now observe that
    \begin{align*}
        \E_{xy\sim \mu}\left[\sum\limits_{j:x_{\mathbf{i}j}\neq y_{\mathbf{i}j}}|\alpha^t_{x\mathbf{i}j}|^2\right] = &\sum\limits_{j\in[d]}\E_{x}[\Pr\limits_{y\sim\mu|x}[x_{\mathbf{i}j}\neq y_{\mathbf{i}j}]\cdot|\alpha^t_{x\mathbf{i}j}|^2]\\
        = & \sum\limits_{j}\int_{0}^\infty \frac{f(x_{\mathbf{i}j})-g(x_{\mathbf{i}j})}{f(x_{\mathbf{i}j})} \cdot |\alpha^t_{x\mathbf{i}j}|^2\cdot f(x_{\mathbf{i}j})\text{d}x_{\mathbf{i}j}\\
        =& \sum\limits_{j}\Big(\int_{0}^{\frac{10\sqrt{\log d}}{\sqrt{d}}} \big(1-\frac{g(x_{\mathbf{i}j})}{f(x_{\mathbf{i}j})}\big) \cdot |\alpha^t_{x\mathbf{i}j}|^2\cdot f(x_{\mathbf{i}j})\text{d}x_{\mathbf{i}j}\\
        &+ \int_{\frac{10\sqrt{\log d}}{\sqrt{d}}}^\infty {\big(f(x_{\mathbf{i}j})-g(x_{\mathbf{i}j})}\big) \cdot |\alpha^t_{x\mathbf{i}j}|^2\text{d}x_{\mathbf{i}j}\Big)\\
        \leq &\sum\limits_{j} \Big(\max\limits_{z\in[0, \frac{10\sqrt{\log d}}{\sqrt{d}}]}|1-\exp(-8\cdot 10^6z)| \cdot \E_{x}[|\alpha^t_{x\mathbf{i}j}|^2]\\
        &+\int_{\frac{10\sqrt{\log d}}{\sqrt{d}}}^\infty {\big(f(x_{\mathbf{i}j})-g(x_{\mathbf{i}j})}\big)\text{d}x_{\mathbf{i}j}\Big)\\
        \leq & \sum\limits_{j} \Big(\frac{8\cdot 10^7\sqrt{\log d}}{\sqrt{d}}\cdot \E_{x\sim\nu_+}[|\alpha^t_{x\mathbf{i}j}|^2]+ 2d^{-100}\Big),
    \end{align*}
    where the first part of the first inequality holds because $\frac{g(x_{\mathbf{i}j})}{f(x_{\mathbf{i}j})}=\exp(-8\cdot 10^6x_{\mathbf{i}j})$, the first part of the second inequality holds because for every $z$, $1-\exp(-z)\leq z$, and the second part of the second inequality holds because both $f$ and $g$ are Gaussians with variance $\frac{1}{4\cdot 10^6d}$ and $\frac{10\sqrt{\log d}}{\sqrt{d}}-\frac{1}{d}\geq 10\cdot \frac{1}{2000\sqrt{d}}$.

We can similarly upper bound 
$ \sum\limits_{t\in[T]} \E_{xy\sim \mu}[\sum\limits_{j:x_{j\mathbf{i}}\neq y_{j\mathbf{i}}}|\alpha^t_{xj\mathbf{i}}|^2]$, 
$\sum\limits_{t\in[T]} \E_{xy\sim \mu}[\sum\limits_{j:x_{\mathbf{i}j}\neq y_{\mathbf{i}j}}|\alpha^t_{y\mathbf{i}j}|^2]$,\\ 
and $ \sum\limits_{t\in[T]} \E_{xy\sim \mu}[\sum\limits_{j:x_{j\mathbf{i}}\neq y_{j\mathbf{i}}}|\alpha^t_{yj\mathbf{i}}|^2]$. Now we have
\begin{align*}
\Omega(1)& \leq |S_0-S_T|\\
    & \leq\sum\limits_{t\in[T]-1}\Delta_t\\
       & \leq  \frac{1}{2} \sum\limits_{t\in[T]-1} \E_{xy\sim \mu}\left[\sum\limits_{j:x_{\mathbf{i}j}\neq y_{\mathbf{i}j}}\Big({|\alpha^t_{x\mathbf{i}j}|^2}+{|\alpha^t_{y\mathbf{i}j}|^2}\Big)+\sum\limits_{j:x_{j\mathbf{i}}\neq y_{j\mathbf{i}}}\Big({|\alpha^t_{x\mathbf{i}j}|^2}+{|\alpha^t_{y\mathbf{i}j}|^2}\Big)\right]\\
        & \leq\sum\limits_{t\in[T]-1,j\in[d]} \Big(\frac{4\cdot 10^7\sqrt{\log d}}{\sqrt{d}}\big(\E_{x\sim\nu_+}[|\alpha^t_{x\mathbf{i}j}|^2+|\alpha^t_{x\mathbf{i}j}|^2]+ \E_{y\sim\nu_-}[|\alpha^t_{y\mathbf{i}j}|^2+|\alpha^t_{yj\mathbf{i}}|^2]\big)+ 4d^{-100}\Big)\\
        & \leq  \frac{4\cdot 10^7\sqrt{\log d}}{\sqrt{d}}\Big(\E_{x\sim\nu_+}[p_{\mathbf{i}x}]+\E_{y\sim \nu_-}[p_{\mathbf{i}y}]\Big)+ 4d^{-97}\\
        & \leq  \frac{8\cdot 10^7\sqrt{\log d}}{\sqrt{d}}\E_{x\sim\nu}[p_{\mathbf{i}x}]+ 4d^{-97}\\
        & \leq \frac{8\cdot 10^7\sqrt{\log d}}{\sqrt{d}}\cdot \frac{4T}{d}+ 4d^{-97},
\end{align*}
where the sixth inequality uses that $\nu_+ + \nu_-=2\nu$, and that $j$ ranges over $d$ values and $t$ ranges over $T$ values (and $T\leq d^2$ without loss of generality).
This implies $T= \Omega(d^{1.5}/\sqrt{\log d})$.
\end{proof}

Again invoking the discussion in \Cref{ssec:hardinstance}, we obtain the following corollary which shows that our second algorithm is close to optimal.

\begin{corollary}\label{cor:quantumLB}
    Let $A$ be a $d\times d$ symmetric matrix with $\norm{A}=\mathcal{O}(1)$ and an $\Omega(1)$ gap between its top and second eigenvalues. Suppose we have quantum query access to the entries of $A$. Every quantum algorithm that with probability at least $\geq 99/100$, approximates the top eigenvector of $A$ with $\ell_2$-error at most $\frac{1}{1000}$ uses $\Omega(d^{1.5}/\sqrt{\log d})$ queries.
\end{corollary}

\subsubsection*{Acknowledgements.}
We thank Alex Wang for discussions about his work on classical optimization algorithms~\cite{ding&wang:sharp} that stimulated the question of quantum speed-up for the power method, Hsin-Po Wang for useful discussions, in particular for the suggestion that Gaussian noise in entries might suffice for the power method, Jeroen Zuiddam for inducing us to look at matrix-matrix multiplication and verification, Jordi Weggemans for useful discussions, Fran\c{c}ois Le Gall for mentioning the middle-eigenvalue problem, Frédéric Magniez for pointing out a small technical issue with the wrap-around in Gaussian phase estimation (which we fixed in this version), and Christoph Lampert and Marco Mondelli for discussions about applications of algorithms for finding the top eigenvector.

\bibliographystyle{alphaUrlePrint}
\bibliography{ref.bib,qc_gily.bib}

\newcommand{\etalchar}[1]{$^{#1}$}
\providecommand{\noopsort}[1]{}\newcommand{\lName}{1}\newcommand{\arxiv}[1]{arXiv:
  \href{https://arxiv.org/abs/#1}{\ttfamily{#1}}\removefirstdot}\newcommand{\arXiv}[1]{arXiv:
  \href{https://arxiv.org/abs/#1}{\ttfamily{#1}}\removefirstdot}\def\removefirstdot#1{\if.#1{}\else#1\fi}\providecommand{\multiletter}[1]{#1}\renewcommand{\multiletter}[1]{#1}\DeclareRobustCommand{\dutchPrefix}[2]{#2}\providecommand{\dutchPrefix}[2]{#2}\renewcommand{\dutchPrefix}[2]{#2}\newcommand{\skp}[3]{#2}\newcommand{\focs
  }[1]{\if\lName1\skp{ }{Proceedings of the #1 {IEEE} Symposium on Foundations
  of Computer Science ({FOCS})}{ }\else{FOCS}\fi}\newcommand{\stoc
  }[1]{\if\lName1\skp{ }{Proceedings of the #1 {ACM} Symposium on the Theory of
  Computing ({STOC})}{ }\else{STOC}\fi}\newcommand{\soda }[1]{\if\lName1\skp{
  }{Proceedings of the #1 {ACM-SIAM} Symposium on Discrete Algorithms
  ({SODA})}{ }\else{SODA}\fi}\newcommand{\stacs }[1]{\if\lName1\skp{
  }{Proceedings of the #1 Symposium on Theoretical Aspects of Computer Science
  ({STACS})}{ }\else{STACS}\fi}\newcommand{\itcs }[1]{\if\lName1\skp{
  }{Proceedings of the #1 Innovations in Theoretical Computer Science
  Conference ({ITCS})}{ }\else{ITCS}\fi}\newcommand{\fsttcs
  }[1]{\if\lName1\skp{ }{Proceedings of the #1 International Conference on
  Foundations of Software Technology and Theoretical Computer Science
  ({FSTTCS})}{ }\else{FSTTCS}\fi}\newcommand{\mfcs }[1]{\if\lName1\skp{
  }{Proceedings of the #1 International Symposium on Mathematical Foundations
  of Computer Science ({MFCS})}{ }\else{MFCS}\fi}\newcommand{\ccc
  }[1]{\if\lName1\skp{ }{Proceedings of the #1 {IEEE} Conference on
  Computational Complexity ({CCC})}{ }\else{CCC}\fi}\newcommand{\isit
  }[1]{\if\lName1\skp{ }{Proceedings of the #1 {IEEE} International Symposium
  on Information Theory ({ISIT})}{ }\else{ISIT}\fi}\newcommand{\colt
  }[1]{\if\lName1\skp{ }{Proceedings of the #1 Conference On Learning Theory
  ({COLT})}{ }\else{COLT}\fi}\newcommand{\nips }[1]{\if\lName1\skp{ }{Advances
  in Neural Information Processing Systems #1 ({NIPS})}{
  }\else{NIPS}\fi}\newcommand{\aistats }[1]{\if\lName1\skp{ }{Proceedings of
  the #1 International Conference on Artificial Intelligence and Statistics
  ({AISTATS})}{ }\else{AISTATS}\fi}\newcommand{\icml }[1]{\if\lName1\skp{
  }{Proceedings of the #1 International Conference on Machine Learning
  ({ICML})}{ }\else{ICML}\fi}\newcommand{\icalp }[1]{\if\lName1\skp{
  }{Proceedings of the #1 International Colloquium on Automata, Languages, and
  Programming ({ICALP})}{ }\else{ICALP}\fi}\newcommand{\esa
  }[1]{\if\lName1\skp{ }{Proceedings of the #1 Annual European Symposium on
  Algorithms ({ESA})}{ }\else{ESA}\fi}\newcommand{\tqc }[1]{\if\lName1\skp{
  }{Proceedings of the #1 Conference on the Theory of Quantum Computation,
  Communication, and Cryptography ({TQC})}{}\else{TQC}\fi}\newcommand{\isaac
  }[1]{\if\lName1\skp{ }{Proceedings of the #1 International Symposium on
  Algorithms and Computation ({ISAAC})}{ }\else{ISAAC}\fi}\newcommand{\jacm
  }{\if\lName1\skp{ }{Journal of the ACM}{ }\else{J. ACM}\fi}\newcommand{\acmta
  }{\if\lName1\skp{ }{ACM Transactions on Algorithms}{ }\else{{ACM} Tr.
  Alg}\fi}\newcommand{\acmtct }{\if\lName1\skp{ }{ACM Transactions on
  Computation Theory}{ }\else{ACM Tr. Comp. Th.}\fi}\newcommand{\acmtqc
  }{\if\lName1\skp{ }{ACM Transactions on Quantum Computing}{ }\else{ACM Tr.
  Quant. Comp.}\fi}\newcommand{\jams }{\if\lName1\skp{ }{Journal of the AMS}{
  }\else{J. AMS}\fi}\newcommand{\pams }{\if\lName1\skp{ }{Proceedings of the
  AMS}{ }\else{Proc. AMS}\fi}\newcommand{\linalgappl }{\if\lName1\skp{ }{Linear
  Algebra and its Applications}{ }\else{Lin. Alg. \&
  App.}\fi}\newcommand{\jalgo }{\if\lName1\skp{ }{Journal of Algorithms}{
  }\else{J. Alg.}\fi}\newcommand{\jcss }{\if\lName1\skp{ }{Journal of Computer
  and System Sciences}{ }\else{J. Comp. Sys. Sci.}\fi}\newcommand{\cc
  }{\if\lName1\skp{ }{Computational Complexity}{ }\else{Comp.
  Comp.}\fi}\newcommand{\algor }{\if\lName1\skp{ }{Algorithmica}{
  }\else{Alg.}\fi}\newcommand{\comb }{\if\lName1\skp{ }{Combinatorica}{
  }\else{Comb.}\fi}\newcommand{\cacm }{\if\lName1\skp{ }{Communications of the
  ACM}{ }\else{Comm. ACM}\fi}\newcommand{\sigart }{\if\lName1\skp{ }{SIGART
  Bulletin}{ }\else{SIGART Bull.}\fi}\newcommand{\sigactn }{\if\lName1\skp{
  }{SIGACT News}{ }\else{SIGACT News}\fi}\newcommand{\eatcsbul
  }{\if\lName1\skp{ }{Bulletin of the {EATCS}}{ }\else{Bull.
  {EATCS}}\fi}\newcommand{\siamrev }{\if\lName1\skp{ }{SIAM Review}{
  }\else{SIAM Rev.}\fi}\newcommand{\siamjc }{\if\lName1\skp{ }{SIAM Journal on
  Computing}{ }\else{SIAM J. Comp.}\fi}\newcommand{\siamjo }{\if\lName1\skp{
  }{SIAM Journal on Optimization}{ }\else{SIAM J. Opt.}\fi}\newcommand{\siamjdm
  }{\if\lName1\skp{ }{SIAM Journal on Discrete Mathematics}{ }\else{SIAM J.
  Disc. Math.}\fi}\newcommand{\siamjnum }{\if\lName1\skp{ }{SIAM Journal on
  Numerical Analysis}{ }\else{SIAM J. Num. Anal.}\fi}\newcommand{\siamjmathanal
  }{\if\lName1\skp{ }{SIAM Journal on Mathematical Analysis}{ }\else{SIAM J.
  Math. Anal.}\fi}\newcommand{\discmath }{\if\lName1\skp{ }{Discrete
  Mathematics}{ }\else{Disc. Math.}\fi}\newcommand{\das }{\if\lName1\skp{
  }{Discrete Applied Mathematics}{ }\else{Disc. App.
  Math.}\fi}\newcommand{\amatstat }{\if\lName1\skp{ }{Annals of Mathematical
  Statistics}{ }\else{Ann. Math. Stat.}\fi}\newcommand{\rms }{\if\lName1\skp{
  }{Russian Mathematical Surveys}{ }\else{Russ. Math.
  Surv.}\fi}\newcommand{\invmath }{\if\lName1\skp{ }{Inventiones Mathematicae}{
  }\else{Inv. Math.}\fi}\newcommand{\jnumber }{\if\lName1\skp{ }{Journal of
  Number Theory}{ }\else{J. Num. Th.}\fi}\newcommand{\tcs }{\if\lName1\skp{
  }{Theoretical Computer Science}{ }\else{Theor. Comput.
  Sci.}\fi}\newcommand{\toc }{\if\lName1\skp{ }{Theory of Computing}{
  }\else{Th. Comp.}\fi}\newcommand{\cjtcs }{\if\lName1\skp{ }{Chicago Journal
  of Theoretical Computer Science}{}\else{Chic. J. Th. Comp.
  Sci.}\fi}\newcommand{\quantum }{\if\lName1\skp{ }{{Quantum}}{
  }\else{Quant.}\fi}\newcommand{\cmp }{\if\lName1\skp{ }{Communications in
  Mathematical Physics}{ }\else{Comm. Math. Phys.}\fi}\newcommand{\jmp
  }{\if\lName1\skp{ }{Journal of Mathematical Physics}{ }\else{J. Math.
  Phys.}\fi}\newcommand{\rspa }{\if\lName1\skp{ }{Proceedings of the Royal
  Society A}{ }\else{Proc. Roy. Soc. A}\fi}\newcommand{\qic }{\if\lName1\skp{
  }{Quantum Information and Computation}{ }\else{Quant. Inf. \&
  Comp.}\fi}\newcommand{\physrev }{\if\lName1\skp{ }{Physical Review}{
  }\else{Phys. Rev.}\fi}\newcommand{\pra }{\if\lName1\skp{ }{Physical Review
  A}{ }\else{Phys. Rev. A}\fi}\newcommand{\prb }{\if\lName1\skp{ }{Physical
  Review B}{ }\else{Phys. Rev. B}\fi}\newcommand{\pre }{\if\lName1\skp{
  }{Physical Review E}{ }\else{Phys. Rev. E}\fi}\newcommand{\prr
  }{\if\lName1\skp{ }{Physical Review Research}{ }\else{Phys. Rev.
  Research}\fi}\newcommand{\prx }{\if\lName1\skp{ }{Physical Review X}{
  }\else{Phys. Rev. X}\fi}\newcommand{\prxq }{\if\lName1\skp{ }{Physical Review
  X Quantum}{ }\else{Phys. Rev. X Quant.}\fi}\newcommand{\prl }{\if\lName1\skp{
  }{Physical Review Letters}{ }\else{Phys. Rev. Lett.}\fi}\newcommand{\njp
  }{\if\lName1\skp{ }{New Journal of Physics}{ }\else{New J.
  Phys.}\fi}\newcommand{\prapp }{\if\lName1\skp{ }{Physical Review Applied}{
  }\else{Phys. Rev. Appl.}\fi}\newcommand{\physrep }{\if\lName1\skp{ }{Physics
  Reports}{ }\else{Phys. Rep.}\fi}\newcommand{\rmp }{\if\lName1\skp{ }{Reviews
  of Modern Physics}{ }\else{Rev. Mod. Phys. }\fi}\newcommand{\phystoday
  }{\if\lName1\skp{ }{Physics Today}{ }\else{Phys.
  Today}\fi}\newcommand{\physics }{\if\lName1\skp{ }{Physics}{
  }\else{Phys.}\fi}\newcommand{\nature }{\if\lName1\skp{ }{Nature}{
  }\else{Nat.}\fi}\newcommand{\natcomm }{\if\lName1\skp{ }{Nature
  Communications}{ }\else{Nat. Comm.}\fi}\newcommand{\natphys }{\if\lName1\skp{
  }{Nature Physics}{ }\else{Nat. Phys.}\fi}\newcommand{\npjqi }{\if\lName1\skp{
  }{npj Quantum Information}{ }\else{npj Quant. Inf.}\fi}\newcommand{\scirep
  }{\if\lName1\skp{ }{Scientific Reports}{ }\else{Sci.
  Rep.}\fi}\newcommand{\science }{\if\lName1\skp{ }{Science}{
  }\else{Sci.}\fi}\newcommand{\jpa }{\if\lName1\skp{ }{Journal of Physics A:
  Mathematical and Theoretical}{ }\else{J. Phys. A}\fi}\newcommand{\ijtp
  }{\if\lName1\skp{ }{International Journal of Theoretical Physics}{
  }\else{Int. J. Th. Phys.}\fi}\newcommand{\jmo }{\if\lName1\skp{ }{Journal of
  Modern Optics}{ }\else{J. Mod. Opt.}\fi}\newcommand{\jstatph
  }{\if\lName1\skp{ }{Journal of Statistical Physics}{ }\else{J. Stat.
  Phys.}\fi}\newcommand{\pnas }{\if\lName1\skp{ }{Proceedings of the National
  Academy of Sciences}{ }\else{PNAS}\fi}\newcommand{\lncs }{\if\lName1\skp{
  }{Lecture Notes in Computer Science}{ }\else{L. Notes Comp.
  Sci.}\fi}\newcommand{\lnai }{\if\lName1\skp{ }{Lecture Notes in Artificial
  Intelligence}{ }\else{L. Notes Art. Int.}\fi}\newcommand{\lnm
  }{\if\lName1\skp{ }{Lecture Notes in Mathematics}{ }\else{L. Notes
  Math.}\fi}\newcommand{\tams }{\if\lName1\skp{ }{Transactions of the American
  Mathematical Society}{ }\else{Trans. AMS}\fi}\newcommand{\ieeetit
  }{\if\lName1\skp{ }{{IEEE} Transactions on Information Theory}{ }\else{{IEEE}
  Trans. Inf. Th.}\fi}\newcommand{\iscs }{\if\lName1\skp{ }{International
  Series in Computer Science}{ }\else{Int. Ser. Comp.
  Sci.}\fi}\newcommand{\tocl }{\if\lName1\skp{ }{Theory of Computing Library}{
  }\else{Th. Comp. Lib.}\fi}
\begin{thebibliography}{{\dutchPrefix{Apeldoorn}{v}}AGG{\dutchPrefix{Wolf}{d}}W20}

\bibitem[ADW{\etalchar{+}}24]{alman2024asymmetryyieldsfastermatrix}
Josh Alman, Ran Duan, Virginia~Vassilevska Williams, Yinzhan Xu, Zixuan Xu, and
  Renfei Zhou.
\newblock More asymmetry yields faster matrix multiplication, 2024.
\newblock \arxiv{2404.16349}.

\bibitem[AG23]{apers&gribling:QIPM}
Simon Apers and Sander Gribling.
\newblock Quantum speedups for linear programming via interior point methods.
\newblock \arXiv{2311.03215}, 2023.

\bibitem[Amb02]{Amb02AdvMethod}
Andris Ambainis.
\newblock \href{http://dx.doi.org/10.1006/jcss.2002.1826}{Quantum lower bounds
  by quantum arguments}.
\newblock {\em Journal of Computer and System Sciences}, 64(4):750--767, 2002.

\bibitem[Amb06]{Amb06weightedAdv}
Andris Ambainis.
\newblock \href{http://dx.doi.org/10.1016/j.jcss.2005.06.006}{Polynomial degree
  vs. quantum query complexity}.
\newblock {\em Journal of Computer and System Sciences}, 72(2):220--238, 2006.

\bibitem[{\dutchPrefix{Apeldoorn}{v}}ACGN23]{apeldoorn2022QTomographyWStatePrepUnis}
Joran {\dutchPrefix{Apeldoorn}{v}}an~Apeldoorn, Arjan Cornelissen, András
  Gilyén, and Giacomo Nannicini.
\newblock \href{http://dx.doi.org/10.1137/1.9781611977554.ch47}{Quantum
  tomography using state-preparation unitaries}.
\newblock In {\em \soda{36th}}, pages 1265--1318, 2023.
\newblock \arxiv{2207.08800}.

\bibitem[{\dutchPrefix{Apeldoorn}{v}}AGG{\dutchPrefix{Wolf}{d}}W20]{apeldoorn2017QSDPSolvers}
Joran {\dutchPrefix{Apeldoorn}{v}}an~Apeldoorn, Andr{\'a}s Gily{\'e}n, Sander
  Gribling, and Ronald {\dutchPrefix{Wolf}{d}}e~Wolf.
\newblock \href{http://dx.doi.org/10.22331/q-2020-02-14-230}{Quantum
  {SDP}-solvers: {B}etter upper and lower bounds}.
\newblock {\em \quantum}, 4:230, 2020.
\newblock Earlier version in FOCS'17. \arxiv{1705.01843}.

\bibitem[{\noopsort{Apeldoorn}v}AGN23]{vAGN23multifinding}
Joran {\noopsort{Apeldoorn}v}an~Apeldoorn, Sander Gribling, and Harold
  Nieuwboer.
\newblock Basic quantum subroutines: finding multiple marked elements and
  summing numbers.
\newblock To appear in Quantum. \arxiv{2302.10244}, 2023.

\bibitem[AS19]{angel2021PairwiseOptCoupling}
Omer Angel and Yinon Spinka.
\newblock Pairwise optimal coupling of multiple random variables, 2019.
\newblock \arxiv{1903.00632}.

\bibitem[A{\noopsort{Wolf}d}W22]{AdW19graphSpar}
Simon Apers and Ronald {\noopsort{Wolf}d}e~Wolf.
\newblock \href{http://dx.doi.org/10.1137/21m1391018}{Quantum speedup for graph
  sparsification, cut approximation, and {L}aplacian solving}.
\newblock {\em {SIAM} Journal on Computing}, 51(6):1703--1742, 2022.

\bibitem[BBBV97]{BBBV97Hybrid}
Charles~H. Bennett, Ethan Bernstein, Gilles Brassard, and Umesh~V. Vazirani.
\newblock \href{http://dx.doi.org/10.1137/S0097539796300933}{Strengths and
  weaknesses of quantum computing}.
\newblock {\em {SIAM} Journal on Computing}, 26(5):1510--1523, 1997.

\bibitem[BBHT98]{boyer1998TightBoundsOnQuantumSearching}
Michel Boyer, Gilles Brassard, Peter H{\o}yer, and Alain Tapp.
\newblock
  \href{http://dx.doi.org/10.1002/(SICI)1521-3978(199806)46:4/5<493::AID-PROP493>3.0.CO;2-P}{Tight
  bounds on quantum searching}.
\newblock {\em Fortschritte der Physik}, 46(4--5):493--505, 1998.
\newblock \arxiv{quant-ph/9605034}.

\bibitem[Bha97]{bhatia1997MatrixAnalysis}
Rajendra Bhatia.
\newblock \href{http://dx.doi.org/10.1007/978-1-4612-0653-8}{{\em Matrix
  Analysis}}, volume 169 of {\em Graduate Texts in Mathematics}.
\newblock Springer, 1997.

\bibitem[BHMT02]{brassard2002AmpAndEst}
Gilles Brassard, Peter H{\o}yer, Michele Mosca, and Alain Tapp.
\newblock \href{http://dx.doi.org/10.1090/conm/305/05215}{Quantum amplitude
  amplification and estimation}.
\newblock In {\em Quantum Computation and Quantum Information: A Millennium
  Volume}, volume 305 of {\em Contemporary Mathematics Series}, pages 53--74.
  AMS, 2002.
\newblock \arxiv{quant-ph/0005055}.

\bibitem[BLM13]{stephane2013ConcIneqThIndep}
Stéphane Boucheron, Gábor Lugosi, and Pascal Massart.
\newblock
  \href{http://dx.doi.org/10.1093/acprof:oso/9780199535255.001.0001}{{\em
  Concentration Inequalities: A Nonasymptotic Theory of Independence}}.
\newblock Oxford University Press, 2013.

\bibitem[B{\v{S}}06]{buhrman2006MatrixProduct}
Harry Buhrman and Robert {\v{S}}palek.
\newblock Quantum verification of matrix products.
\newblock In {\em \soda{17th}}, pages 880--889, 2006.
\newblock \arxiv{quant-ph/0409035}.

\bibitem[BvH16]{BH14Gaussianrandom}
Afonso~S. Bandeira and Ramon van Handel.
\newblock \href{http://www.jstor.org/stable/24735373}{Sharp nonasymptotic
  bounds on the norm of random matrices with independent entries}.
\newblock {\em The Annals of Probability}, 44(4):2479--2506, 2016.

\bibitem[CEMM98]{cleve1997QAlgsRevisited}
Richard Cleve, Artur Ekert, Chiara Macchiavello, and Michele Mosca.
\newblock \href{http://dx.doi.org/10.1098/rspa.1998.0164}{Quantum algorithms
  revisited}.
\newblock {\em \rspa}, 454(1969):339--354, 1998.
\newblock \arxiv{quant-ph/9708016}.

\bibitem[Che52]{chernoff1952Bound}
Herman Chernoff.
\newblock \href{http://dx.doi.org/10.1214/aoms/1177729330}{A measure of
  asymptotic efficiency for tests of a hypothesis based on the sum of
  observations}.
\newblock {\em The Annals of Mathematical Statistics}, 23(4):493 -- 507, 1952.

\bibitem[Chi09]{Childs2009CDQRW}
Andrew~M. Childs.
\newblock \href{http://dx.doi.org/10.1007/s00220-009-0930-1}{On the
  relationship between continuous- and discrete-time quantum walk}.
\newblock {\em Communications in Mathematical Physics}, 294(2):581–603, 2009.

\bibitem[CHJ22]{CHJ:multivar}
{Arjan} Cornelissen, Yassine Hamoudi, and Sofiene Jerbi.
\newblock Near-optimal quantum algorithms for multivariate mean estimation.
\newblock In {\em Proceedings of the 54st Annual {ACM} {SIGACT} Symposium on
  Theory of Computing}, pages 33--43. {ACM}, 2022.
\newblock arXiv:2111.09787.

\bibitem[{\multiletter{Cs}}is67]{csiszar1967InfMeasProbDist}
Imre {\multiletter{Cs}}iszár.
\newblock Information-type measures of difference of probability distributions
  and indirect observations.
\newblock {\em Studia Scientiarum Mathematicarum Hungarica}, 2:299--318, 1967.

\bibitem[C{\noopsort{Wolf}d}W23]{CdW21QLasso}
Yanlin Chen and Ronald {\noopsort{Wolf}d}e~Wolf.
\newblock \href{http://dx.doi.org/10.4230/LIPIcs.ICALP.2023.38}{Quantum
  algorithms and lower bounds for linear regression with norm constraints}.
\newblock In {\em Proceedings of 50th International Colloquium on Automata,
  Languages, and Programming}, volume 261 of {\em LIPIcs}, pages 38:1--38:21,
  2023.

\bibitem[DS01]{davidson2001RandMatBanachSpaces}
Kenneth~R. Davidson and Stanisław~J. Szarek.
\newblock \href{https://case.edu/artsci/math/szarek/TeX/DavSzaHB.pdf}{Local
  operator theory, random matrices and {B}anach spaces}.
\newblock In William~B. Johnson and Joram Lindenstrauss, editors, {\em Handbook
  of the Geometry of Banach Spaces, Vol. I}, volume~I, pages 317--366,
  North-Holland, Amsterdam, 2001. Elsevier Science.

\bibitem[DW23]{ding&wang:sharp}
Lijun Ding and Alex~L. Wang.
\newblock Sharpness and well-conditioning of nonsmooth convex formulations in
  statistical signal recovery.
\newblock \arXiv{2307.06873}, 2023.

\bibitem[Fre77]{freivalds:matrixmult}
Rusins Freivalds.
\newblock Probabilistic machines can use less running time.
\newblock In {\em Proceedings of 7th IFIP Congress}, pages 839--842, 1977.

\bibitem[Gil23]{gilyen2023IterativeRefineTomoLinEq}
András Gilyén.
\newblock Iterative refinement for improved tomography and quantum linear
  equation solving.
\newblock Poster presented at QIP'24, manuscript forthcoming, 2023.

\bibitem[GL13]{golub&vanloan:matrixcomp}
Gene~H. Golub and Charles F.~Van Loan.
\newblock {\em Matrix Computations}, volume~3 of {\em Johns Hopkins Studies in
  Mathematical Sciences}.
\newblock Johns Hopkins University Press, fourth edition, 2013.

\bibitem[Gor85]{gordon1985InequalitiesGaussianProcesses}
Yehoram Gordon.
\newblock \href{http://dx.doi.org/10.1007/BF02759761}{Some inequalities for
  {G}aussian processes and applications}.
\newblock {\em Israel Journal of Mathematics}, 50(4):265--289, 1985.

\bibitem[Gro96]{grover1996QSearch}
Lov~K. Grover.
\newblock \href{http://dx.doi.org/10.1145/237814.237866}{A fast quantum
  mechanical algorithm for database search}.
\newblock In {\em \stoc{28th}}, pages 212--219, 1996.
\newblock \arxiv{quant-ph/9605043}.

\bibitem[GSLW19]{gilyen2018QSingValTransf}
András Gilyén, Yuan Su, Guang~Hao Low, and Nathan Wiebe.
\newblock \href{http://dx.doi.org/10.1145/3313276.3316366}{Quantum singular
  value transformation and beyond: {E}xponential improvements for quantum
  matrix arithmetics}.
\newblock In {\em \stoc{51st}}, pages 193--204, 2019.
\newblock Full version in \arxiv{1806.01838}.

\bibitem[HHJ{\etalchar{+}}17]{haah2017OptTomography}
Jeongwan Haah, Aram~W. Harrow, Zhengfeng Ji, Xiaodi Wu, and Nengkun Yu.
\newblock \href{http://dx.doi.org/10.1109/TIT.2017.2719044}{Sample-optimal
  tomography of quantum states}.
\newblock {\em \ieeetit}, 63(9):5628--5641, 2017.
\newblock \arxiv{1508.01797}.

\bibitem[HHL09]{harrow2009QLinSysSolver}
Aram~W. Harrow, Avinatan Hassidim, and Seth Lloyd.
\newblock \href{http://dx.doi.org/10.1103/PhysRevLett.103.150502}{Quantum
  algorithm for linear systems of equations}.
\newblock {\em \prl}, 103(15):150502, 2009.
\newblock \arxiv{0811.3171}.

\bibitem[HKOT23]{HKOT23QueryOptUnitaryProcTomo}
Jeongwan Haah, Robin Kothari, Ryan O'Donnell, and Ewin Tang.
\newblock Query-optimal estimation of unitary channels in diamond distance.
\newblock \arXiv{2302.14066}, 2023.

\bibitem[Hoe63]{hoeffding1963ProbIneqSumsOfBoundedRVs}
Wassily Hoeffding.
\newblock \href{http://dx.doi.org/10.2307/2282952}{Probability inequalities for
  sums of bounded random variables}.
\newblock {\em Journal of the American Statistical Association},
  58(301):13--30, 1963.

\bibitem[HP14]{HP15NoisyPowerMethod}
Moritz Hardt and Eric Price.
\newblock
  \href{https://proceedings.neurips.cc/paper_files/paper/2014/file/729c68884bd359ade15d5f163166738a-Paper.pdf}{The
  noisy power method: A meta algorithm with applications}.
\newblock In {\em Proceedings of the Advances in Neural Information Processing
  Systems}, volume~27, 2014.

\bibitem[Jol02]{Jol86PCA}
Ian Jolliffe.
\newblock \href{http://dx.doi.org/https://doi.org/10.1007/b98835}{{\em
  Principal Components in Regression Analysis}}.
\newblock Springer New York, second edition, 2002.

\bibitem[Kem69]{kemperman1969OptInfTransmitionRate}
Johannes~H.B. Kemperman.
\newblock \href{http://dx.doi.org/10.1007/BFb0079123}{On the optimum rate of
  transmitting information}.
\newblock In M.~Behara, K.~Krickeberg, and J.~Wolfowitz, editors, {\em
  Probability and Information Theory}, volume~89 of {\em Lecture Notes in
  Mathematics}, pages 126--169. Springer Berlin Heidelberg, 1969.

\bibitem[Kit95]{kitaev1995QMeasAndAbelianStabilizer}
Alexei~Yu. Kitaev.
\newblock Quantum measurements and the {A}belian stabilizer problem.
\newblock \arxiv{quant-ph/9511026}, 1995.

\bibitem[KP17]{kerenidis2016QRecSys}
Iordanis Kerenidis and Anupam Prakash.
\newblock \href{http://dx.doi.org/10.4230/LIPIcs.ITCS.2017.49}{Quantum
  recommendation systems}.
\newblock In {\em \itcs{8th}}, pages 49:1--49:21, 2017.
\newblock \arxiv{1603.08675}.

\bibitem[KP20]{kerenidis2018QIntPoint}
Iordanis Kerenidis and Anupam Prakash.
\newblock \href{http://dx.doi.org/10.1145/3406306}{A quantum interior point
  method for {LP}s and {SDP}s}.
\newblock {\em \acmtqc}, 1(1), 2020.
\newblock \arxiv{1808.09266}.

\bibitem[Kul67]{kullback1967PinskerTight}
Solomon Kullback.
\newblock \href{http://dx.doi.org/10.1109/TIT.1967.1053968}{A lower bound for
  discrimination information in terms of variation}.
\newblock {\em \ieeetit}, 13(1):126--127, 1967.
\newblock \href{https://doi.org/10.1109/TIT.1970.1054514}{Correction, volume
  16(5), p. 652, 1970.}

\bibitem[KV09]{KV09spectral}
Ravindran Kannan and Santosh Vempala.
\newblock
  \href{http://dx.doi.org/http://dx.doi.org/10.1561/0400000025}{Spectral
  algorithms}.
\newblock {\em Foundations and Trends in Theoretical Computer Science},
  4:57--288, 2009.

\bibitem[LC19]{low2016HamSimQubitization}
Guang~Hao Low and Isaac~L. Chuang.
\newblock \href{http://dx.doi.org/10.22331/q-2019-07-12-163}{Hamiltonian
  simulation by qubitization}.
\newblock {\em \quantum}, 3:163, 2019.
\newblock \arxiv{1610.06546}.

\bibitem[Led01]{ledoux2001ConcentrationOfMeasure}
Michel Ledoux.
\newblock \href{https://bookstore.ams.org/surv-89-s}{{\em The Concentration of
  Measure Phenomenon}}.
\newblock Mathematical surveys and monographs. American Mathematical Society,
  2001.

\bibitem[Li11]{Li11}
Shengqiao Li.
\newblock \href{https://docsdrive.com/pdfs/ansinet/ajms/2011/66-70.pdf}{Concise
  formulas for the area and volume of a hyperspherical cap}.
\newblock {\em Asian Journal of Mathematics and Statistics}, pages 66--70,
  2011.

\bibitem[LMR14]{lloyd2013QPrincipalCompAnal}
Seth Lloyd, Masoud Mohseni, and Patrick Rebentrost.
\newblock \href{http://dx.doi.org/10.1038/nphys3029}{Quantum principal
  component analysis}.
\newblock {\em \natphys}, 10:631--633, 2014.
\newblock \arxiv{1307.0401}.

\bibitem[Low19]{low2018HamSimNearlyOptSpecNorm}
Guang~Hao Low.
\newblock \href{http://dx.doi.org/10.1145/3313276.3316386}{Hamiltonian
  simulation with nearly optimal dependence on spectral norm}.
\newblock In {\em \stoc{51st}}, pages 491--502, 2019.
\newblock \arxiv{1807.03967}.

\bibitem[MP12]{MP11Trapdoors}
Daniele Micciancio and Chris Peikert.
\newblock \href{http://dx.doi.org/10.1007/978-3-642-29011-4\_41}{Trapdoors for
  lattices: Simpler, tighter, faster, smaller}.
\newblock In {\em Proceedings of the Advances in Cryptology 31st Annual
  International Conference on the Theory and Applications of Cryptographic
  Techniques}, volume 7237, pages 700--718, 2012.

\bibitem[MR07]{MR04GaussianMeasure}
Daniele Micciancio and Oded Regev.
\newblock \href{http://dx.doi.org/10.1137/S0097539705447360}{Worst-case to
  average-case reductions based on {G}aussian measures}.
\newblock {\em {SIAM} Journal on Computing}, 37(1):267--302, 2007.

\bibitem[NW23]{NWlargesteigenvalue}
Nhat~A. Nghiem and Tzu-Chieh Wei.
\newblock
  \href{http://dx.doi.org/https://doi.org/10.1016/j.physleta.2023.129138}{Quantum
  algorithm for estimating largest eigenvalues}.
\newblock {\em Physics Letters A}, 488(12913857), 2023.

\bibitem[Oka59]{okamoto1958IneqSqrtBinomial}
Masashi Okamoto.
\newblock \href{http://dx.doi.org/10.1007/BF02883985}{Some inequalities
  relating to the partial sum of binomial probabilities}.
\newblock {\em Annals of the Institute of Statistical Mathematics},
  10(1):29--35, 1959.

\bibitem[Par87]{parlett:symeigenvalue}
Beresford~N. Parlett.
\newblock {\em The Symmetric Eigenvalue Problem}.
\newblock Number~20 in Classics in Applied Mathematics. SIAM, 1987.

\bibitem[Pra14]{Pra14}
Anupam Prakash.
\newblock {\em Quantum algorithms for linear algebra and machine learning}.
\newblock PhD thesis, University of California, Berkeley, 2014.

\bibitem[Rav21]{Rav21Stack}
Alex Ravsky.
\newblock
  \href{https://math.stackexchange.com/questions/3922615/inner-product-of-unit-vector-with-a-vector-uniformly-distributed-on-a-hyperspher}{Reply
  on {S}tack{E}xchange}.
\newblock StackExchange, 2021.

\bibitem[Reg09]{Reg09lattices}
Oded Regev.
\newblock On lattices, learning with errors, random linear codes, and
  cryptography.
\newblock {\em Journal of the ACM}, 56(6):1--40, 2009.

\bibitem[Rob55]{robbins1955RemarkOnStrilingsFormula}
Herbert Robbins.
\newblock \href{http://dx.doi.org/10.2307/2308012}{A remark on {S}tirling's
  formula}.
\newblock {\em The American Mathematical Monthly}, 62(1):26--29, 1955.

\bibitem[Rud76]{rudin1976principles}
Walter Rudin.
\newblock \href{https://books.google.nl/books?id=kwqzPAAACAAJ}{{\em Principles
  of Mathematical Analysis}}.
\newblock International series in pure and applied mathematics. McGraw-Hill,
  1976.

\bibitem[SY21]{SY21QPM}
Kazuhiro Seki and Seiji Yunoki.
\newblock \href{http://dx.doi.org/10.1103/PRXQuantum.2.010333}{Quantum power
  method by a superposition of time-evolved states}.
\newblock {\em PRX Quantum}, 2:010333, 2021.

\bibitem[SZ23]{SZ:qspeedupstochastic}
Aaron Sidford and Chenyi Zhang.
\newblock Quantum speedups for stochastic optimization.
\newblock In {\em Proceedings of NeurIPS}, 2023.
\newblock arXiv:2308.01582.

\bibitem[Ver12]{vershynin2012IntroNonAsyAnalyRandMat}
Roman Vershynin.
\newblock \href{http://dx.doi.org/10.1017/CBO9780511794308.006}{Introduction to
  the non-asymptotic analysis of random matrices}.
\newblock In Yonina~C. Eldar and Gitta Kutyniok, editors, {\em Compressed
  Sensing: Theory and Applications}, pages 210--268. Cambridge University
  Press, 2012.
\newblock \arxiv{1011.3027}.

\bibitem[Wed72]{wedin1972DavisKahanSinThSingularVec}
Per-{\AA}ke Wedin.
\newblock \href{http://dx.doi.org/10.1007/BF01932678}{Perturbation bounds in
  connection with singular value decomposition}.
\newblock {\em BIT Numerical Mathematics}, 12(1):99--111, 1972.

\bibitem[WXXZ24]{williams2023newboundsmatrixmultiplication}
Virginia~Vassilevska Williams, Yinzhan Xu, Zixuan Xu, and Renfei Zhou.
\newblock \href{https://doi.org/10.1137/1.9781611977912.134}{New bounds for
  matrix multiplication: from alpha to omega}.
\newblock In {\em Proceedings of the 2024 Symposium on Discrete Algorithms},
  pages 3792--3835, 2024.

\bibitem[YLC14]{yoder2014FixedPointSearch}
Theodore~J. Yoder, Guang~Hao Low, and Isaac~L. Chuang.
\newblock \href{http://dx.doi.org/10.1103/PhysRevLett.113.210501}{Fixed-point
  quantum search with an optimal number of queries}.
\newblock {\em \prl}, 113(21):210501, 2014.
\newblock \arxiv{1409.3305}.

\end{thebibliography}

\end{document}